\documentclass[11pt]{article} %11-point
\usepackage{hyperref}
\usepackage{changepage}
\usepackage[letterpaper, margin=1in]{geometry} %letter-size, 1-inch margins all around
\usepackage{setspace}
\usepackage{url}
\usepackage{graphicx} % Required for inserting images
\usepackage{amssymb,amsmath,amsthm}
\usepackage{thmtools, thm-restate}
\usepackage{algorithm}
\usepackage{algorithmic}
\usepackage{enumitem} 
\usepackage{cases} 

\newtheorem{theorem}{Theorem}
\newtheorem{lemma}{Lemma}

\theoremstyle{definition}
\newtheorem{definition}{Definition}[section]

\usepackage{xcolor}
\definecolor{WildStrawberry}{RGB}{255,67,164}
\definecolor{Dousha}{RGB}{47, 136, 67}
\definecolor{BlueBerry}{RGB}{44, 96, 189}

\newcommand{\be}{\begin{equation}}
\newcommand{\ee}{\end{equation}}
\newcommand{\beq}{\begin{equation*}}
\newcommand{\eeq}{\end{equation*}}

\newcommand{\R}{\mathbb{R}}

\newcommand{\E}{\mathbb{E}}

\newcommand{\eps}{\varepsilon}
\newcommand{\ind}[1]{\mathbb{I}\hspace{-0.1em}\left[\vphantom{\sum}#1\right]}

\newcommand{\AutoAdjust}[3]{\mathchoice{ \left #1 #2  \right #3}{#1 #2 #3}{#1 #2 #3}{#1 #2 #3} }
\newcommand{\Xcomment}[1]{{}}

\newcommand{\InBrackets}[1]{\AutoAdjust{[}{#1}{]}}% {\left[{#1}\right]}
\newcommand{\abs}[1]{\AutoAdjust{|}{#1}{|}}%
\newcommand{\Ex}[2][]{\operatorname{\mathbf E}_{#1}\InBrackets{#2}}

\newcommand{\eqdef}{\overset{\mathrm{def}}{=\mathrel{\mkern-3mu}=}}
\newcommand{\vect}[1]{\ensuremath{\mathbf{#1}}}

\newcommand\restr[2]{{% we make the whole thing an ordinary symbol
  \left.\kern-\nulldelimiterspace % automatically resize the bar with \right
  #1 % the function
  \vphantom{\big|} % pretend it's a little taller at normal size
  \right|_{#2} % this is the delimiter
  }}
\newcommand{\dd}{\mathrm{d}}

\newcommand{\costi}[1][i]{\texttt{cost}_{#1}}
\newcommand{\SC}{\texttt{SC}}

\newcommand{\mech}{\mathcal{M}}

\newcommand{\loc}{\vect{p}}

\newcommand{\f}{\mathbf{f}}

\newcommand{\sign}{\texttt{sign}}

\title{Product Gap Mechanisms for Multi-Facility Location}
\author{Jianhao Jia}
\date{}

\begin{document}

\maketitle

    \begin{abstract}
We study randomized strategyproof mechanisms for locating multiple
facilities on the real line. We introduce the \emph{Product-Gap
mechanism}, which selects $k$ reported locations with probability
proportional to the product of the consecutive gaps between them and
opens facilities at the selected locations. We prove that, for every
$k\geq 2$, the mechanism achieves a tight approximation ratio of $2k$
for social cost. We then study its incentive properties and show that
it is strategyproof in expectation for $k=2$ and $k=3$, but is not
strategyproof for $k\geq 4$. In particular, the mechanism gives
strategyproof $4$- and $6$-approximations for two and three facilities,
respectively. Finally, for two facilities, we combine Product-Gap with
the Proportional mechanism of Lu et al. We show that an optimized
report-independent mixture is strategyproof and has a tight
approximation ratio of
$(74+4\sqrt{3})/23\approx 3.519$ on the line, improving upon the
previous factor of $4$.
\end{abstract}

\section{Introduction}
\label{sec:intro}
Facility location is one of the canonical problems in mechanism design
without money. A planner wishes to locate public facilities---such as
schools, hospitals, or service centers---based on locations reported by
the agents who will use them. Each agent prefers a facility to be close
to her true location and may therefore benefit from misreporting. The
planner must simultaneously guarantee that truthful reporting is
optimal for every agent and that the resulting facilities serve the
population efficiently.

We consider the standard $k$-facility location problem on the real
line. There are $n$ agents, and a mechanism opens $k$ facilities as a
function of their reported locations. The cost of an agent is her
distance to the nearest open facility, and the social cost is the sum
of the agents' costs. Our objective is to design a randomized
strategyproof mechanism whose expected social cost is within a small
factor of the optimal $k$-facility social cost. Since the underlying
one-dimensional $k$-median problem can be solved optimally, the loss
measured by this approximation ratio arises from the incentive
constraint rather than from computational hardness. Facility location
has consequently served as a central test case for approximate
mechanism design without money
\cite{procaccia2013approximate,chan2021mechanism}.

For a single facility, the interaction between incentives and
efficiency is particularly well understood. On the line, opening the
facility at a median report is strategyproof, indeed
group-strategyproof, and exactly minimizes the social cost. This
observation is rooted in the classical theory of single-peaked
preferences initiated by Black and in Moulin's characterization of
strategyproof rules on the line
\cite{Black48,moulin1980strategy}. The one-facility problem has also
been studied extensively beyond the line. Recent work has refined the
approximation guarantees of the coordinate-wise median in
higher-dimensional and general normed spaces and has demonstrated
further improvements through randomization
\cite{jia1,barak26,ChanLW26,hastings2026}. Thus, although important
questions remain in richer domains, the basic structure of
strategyproof single-facility location is by now understood rather
well.

The picture changes substantially when two facilities must be opened.
Simple deterministic mechanisms have approximation ratios that grow
linearly with the number of agents. In particular, the Two-Extremes
mechanism, which opens facilities at the leftmost and rightmost
reports, has approximation ratio $n-2$ on the line, and this dependence
on $n$ is unavoidable for deterministic strategyproof mechanisms
\cite{LuSWZ10,FotakisT14}. Lu, Sun, Wang, and Zhu showed that
randomization overcomes this linear barrier. Their Proportional
mechanism first selects one agent uniformly at random and then selects
a second agent with probability proportional to her distance from the
first. The mechanism is strategyproof in expectation and achieves
approximation ratio $4$ in every metric space
\cite{LuSWZ10}. This result has been the basic positive benchmark for
multiple-facility location.

For more than two facilities, however, the standard model has resisted
constant-factor strategyproof approximation. Fotakis and Tzamos proved
that, for every $k\geq 3$, no deterministic anonymous strategyproof
mechanism has a bounded approximation ratio on the line, even on
instances with only $k+1$ agents
\cite{FotakisT14}. They also obtained positive results under the
winner-imposing relaxation. In that model, an agent selected to host a
facility is required to use that particular facility rather than her
nearest open facility. In particular, they gave an
$(n-1)$-approximate imposing strategyproof mechanism for three
facilities on the line, and more generally a $4k$-approximation in the
winner-imposing model
\cite{FotakisT10,FotakisT14}. These mechanisms demonstrate that
multiple facilities can be handled after modifying the agents' service
rule, but they do not resolve the standard model considered here.

In the standard nearest-facility model, the previously known
strategyproof guarantees for $k>2$ depended on the number of agents.
For three facilities on the line, Lu et al.\ described a randomized
mechanism with approximation ratio $n-1$
\cite{LuSWZ10}. Fotakis and Tzamos later proposed the Equal-Cost
mechanism, which applies to every $k$ and $n$ and has social-cost
approximation ratio at most $n$
\cite{FotakisT16}. Constant guarantees were also known in special
regimes, such as when $n=k+1$
\cite{EscoffierGTPS11,FotakisT16}. Nevertheless, for the unrestricted
standard model, it remained open whether a randomized strategyproof
mechanism could achieve an approximation factor independent of $n$ for
any fixed $k>2$.

This leads to the central open question answered in this paper:

\begin{quote}
\emph{Can one design a randomized strategyproof mechanism with a
constant approximation ratio for more than two facilities on the
line?}
\end{quote}

We answer this question through a simple randomized mechanism that
extends the distance-based idea underlying two-facility mechanisms.
For a set $S$ of $k$ agents, let
$p_1(S)\leq\cdots\leq p_k(S)$ be their ordered reported locations. The
\emph{Product-Gap mechanism} assigns $S$ the weight

$$
W_k(S)
=
\prod_{j=1}^{k-1}
\bigl(p_{j+1}(S)-p_j(S)\bigr),
$$

selects a set with probability proportional to this weight, and opens
facilities at the selected reports. For $k=2$, the weight is simply
the distance between the two selected agents. For larger $k$, the
product rewards sets that are separated across every consecutive gap.
The definition is uniform in $k$, but, perhaps surprisingly, its
approximation and incentive properties exhibit very different
boundaries.

\paragraph{Our results.}
Our first result is an exact approximation analysis that applies to
every number of facilities. For every $k\geq2$, the Product-Gap
mechanism has expected social cost at most
$2k\cdot\mathrm{OPT}_k$. Moreover, this guarantee is tight for the
mechanism: for every fixed $k$, there are instances whose approximation
ratio approaches $2k$. This welfare analysis does not rely on
strategyproofness and therefore continues to apply for values of $k$
for which the mechanism is manipulable.

Our second result determines the incentive boundary of the mechanism.
We prove that the Product-Gap mechanism is strategyproof in expectation
for $k=3$. Together with the approximation theorem, this gives a
strategyproof $6$-approximation for three-facility location on the
line. In contrast, we construct a profitable deviation for $k=4$ and
extend the construction to every $k\geq4$. Thus, although the
Product-Gap mechanism is defined and achieves a $2k$-approximation for
every $k$, three is the largest number of facilities for which the
mechanism remains strategyproof.

Finally, we revisit the two-facility problem. For $k=2$, Product-Gap
selects an unordered pair of agents with probability proportional to
the distance between their reports; this is the mechanism called
\emph{Global Pair} by Ma and Peng \cite{ma2026}. We combine Global
Pair with the Proportional mechanism of Lu et al.\ \cite{LuSWZ10}.
For every fixed mixing probability $\lambda$, we determine the exact
worst-case approximation ratio of the resulting mixture on the line.
Optimizing over $\lambda$ gives the tight ratio

$$
\frac{74+4\sqrt{3}}{23}
\approx 3.519.
$$

The tightness statement is with respect to the class of fixed,
report-independent mixtures of Proportional and Global Pair. The proof
exploits the order structure of the line and the complementary
worst-case behavior of the two component mechanisms.

\paragraph{Concurrent and independent work.}
This work was developed independently and concurrently with the work of
Ma and Peng \cite{ma2026}, whose preprint appeared while this
manuscript was in preparation. Ma and Peng were the first to make
public a strategyproof two-facility mechanism with an approximation
ratio strictly below $4$. They independently introduced the mechanism
they call \emph{Global Pair}, which coincides with the $k=2$
specialization of Product-Gap, and independently proposed mixing it
with the Proportional mechanism of Lu et al.\ \cite{LuSWZ10}. They
prove an approximation ratio of $11/3$ for their mixture on the
broader class of Ptolemaic metric spaces.

Before becoming aware of their work, we had independently discovered
the Product-Gap mechanism, its Global Pair specialization, and the idea
of mixing this specialization with the Proportional mechanism, together
with an analysis showing that the mixture beats the factor-$4$
barrier on the line. We therefore regard the discovery of the
mechanism and the mixture idea as independent and concurrent, while
recognizing Ma and Peng's priority as the first public account of the
below-$4$ result.

Our two-facility contribution is a sharper analysis specialized to the
line. We obtain the exact ratio
$(74+4\sqrt{3})/23$ among all fixed mixtures of Proportional and
Global Pair. Our proof of the Global Pair dispersion bound uses the
same basic triple-decomposition argument as Ma and Peng. We retain
their attribution and defer a self-contained line-specific proof to
the appendix. The main focus of our paper is complementary to their
work: we define Product-Gap for arbitrary $k$, prove its tight
$2k$ approximation guarantee, establish strategyproofness for three
facilities, and show that strategyproofness fails for every
$k\geq4$.

\paragraph{Organization.}
Section~\ref{sec:model-mechanism} introduces the model and formally
defines the Product-Gap mechanism. Section~\ref{sec:approximation}
proves its tight $2k$ approximation guarantee for every $k$.
Section~\ref{sec:strategyproofness} proves strategyproofness for
$k=3$ and gives profitable deviations for every $k\geq4$.
Section~\ref{sec:two-facility-mixture} studies the two-facility
specialization and proves the exact approximation ratio of the optimal
fixed mixture of Global Pair and the Proportional mechanism of Lu et
al. A self-contained proof of the Global Pair dispersion bound on the
line is provided in the appendix.

\paragraph{Use of artificial intelligence.}
Generative AI tools were used in the preparation of this manuscript to
assist with writing and with formalizing parts of the proofs. In
particular, these tools helped improve the exposition, organize proof
arguments, and expand proof sketches into explicit intermediate claims
and calculations. AI-generated material was treated as draft material
rather than as mathematical evidence. The research questions,
mechanisms, and underlying proof ideas were developed by the authors.

\section{Other related work}
Facility location lies at the intersection of social choice and
mechanism design without money; see the survey of Chan et al.\
\cite{chan2021mechanism} and the framework of Procaccia and
Tennenholtz \cite{procaccia2013approximate}. Its social-choice
foundations go back to Black's median-voter analysis for single-peaked
preferences \cite{Black48}. Moulin subsequently characterized a broad
class of strategyproof rules on single-peaked domains
\cite{moulin1980strategy}. Related characterizations investigate the
roles of unanimity, continuity, anonymity, range restrictions, and
phantom voters
\cite{border1983straightforward,peters1993range,
barbera1993generalized,ching1997strategy,barbera1998strategy}.
For a broader overview of strategyproof social choice, see
\cite{barbera2011strategyproof}. These results help explain why the
single-facility problem on a line admits a particularly clean
description in terms of median and generalized-median mechanisms.

A substantial literature studies how this structure changes with the
geometry or topology of the location space. Schummer and Vohra
consider strategyproof facility location on networks
\cite{schummer2002strategy}, while Dokow et al.\ study discrete lines
and cycles \cite{DokowFMN12}. Meir investigates the three-agent
facility-location problem on a circle
\cite{meir2019strategyproof}. In multidimensional domains, Walsh
studies strategyproof mechanisms in Euclidean and Manhattan spaces
\cite{walsh2020strategy}, and Tang et al.\ characterize
group-strategyproof mechanisms in strictly convex spaces
\cite{TangYZ20}. Goel and Hann-Caruthers analyze the optimality of the
coordinate-wise median in two dimensions
\cite{goel2023optimality}, while El-Mhamdi et al.\ investigate the
strategyproofness of the geometric median
\cite{el2023strategyproofness}. More recently, Gravin and Jia obtain
approximation guarantees for the median mechanism in $\R^d$
\cite{jia1}. Recent work has also studied the value of randomization
and strategyproof approximation under general $L_p$-norm social-cost
objectives \cite{barak26,ChanLW26,hastings2026}.

The social cost considered in this paper is only one of several
objectives studied in strategic facility location. Procaccia and
Tennenholtz study approximate mechanism design for both utilitarian and
minimax objectives \cite{procaccia2013approximate}. Alon et al.\
consider strategyproof approximation of the minimax objective on
networks \cite{AlonFPT10}, whereas Feldman and Wilf study the
least-squares objective \cite{FeldmanW13}. Feigenbaum et al.\ consider
the more general objective given by the $L_p$ norm of the agents'
individual costs \cite{FeigenbaumSY17}. Procaccia et al.\ study the
tradeoff between approximation quality and the variance of randomized
facility-location mechanisms \cite{procaccia2018approximation}.
Aziz et al.\ instead focus on proportional fairness and characterize a
randomized facility-location rule satisfying both fairness and
strategyproofness \cite{AzizLSW22}. Facility-location mechanisms are
also closely related to voting mechanisms, particularly when the
facility must be chosen from a restricted set of alternatives
\cite{Feldman16}.

Beyond the standard homogeneous nearest-facility model, several works
consider different facility types, agent preferences, or service
rules. Fotakis and Tzamos study winner-imposing mechanisms for multiple
facilities \cite{FotakisT10}, deterministic mechanisms for
multi-facility location \cite{FotakisT14}, and multi-facility location games
with concave cost functions \cite{FotakisT16}. Escoffier et al.\
consider settings with many facilities, including regimes in which the
number of facilities is close to the number of agents
\cite{EscoffierGTPS11}. Serafino and Ventre study heterogeneous
facilities, where facilities may provide different services
\cite{SerafinoV16}. Other variants allow agents to have dual or
distinct preferences over facility locations \cite{zou15,mei2019}.
Obnoxious facility-location problems, in which agents prefer a facility
to be far away, have also been studied on networks and in
prediction-augmented models
\cite{cheng2013strategy,istrate2022mechanism}.

Another line of work strengthens the manipulation model. Todo et al.\
introduce false-name-proof mechanism design without money, allowing an
agent to participate under multiple identities \cite{todo11}.
False-name and identity manipulations have subsequently been studied
for two-facility location and for discrete facility-location problems
with optional preferences
\cite{sonoda2016false,ono2017rename}. Related extensions include
distributed facility location, in which information or decisions are
aggregated across multiple districts \cite{filoratsikas2021}, and
dynamic facility reallocation, in which already opened facilities must
move as the instance changes over time
\cite{fotakis2019,dekeijzer2022}. These settings impose requirements
that are absent from our static model, but they illustrate the variety
of strategic constraints that arise once multiple facilities are
present.

Finally, a recent literature augments facility-location mechanisms
with predictions or advice. Agrawal et al.\ initiate the systematic
study of learning-augmented mechanism design for facility location
\cite{agrawal2022learning}, while Xu and Lu develop a more general
framework for mechanism design with predictions \cite{XuL22}. Barak
et al.\ study advice about agents' locations
\cite{barak2024mac}, whereas Christodoulou et al.\ consider mechanisms
augmented with advice about the desired output
\cite{christodoulou2024mechanismdesignaugmentedoutput}. Balkanski et
al.\ further investigate randomized strategic facility location with
predictions
\cite{balkanski2024randomizedstrategicfacilitylocation}. These works
seek mechanisms that exploit accurate advice while retaining
worst-case guarantees when the advice is inaccurate. In contrast, the
Product-Gap mechanism uses no predictions or auxiliary information and
is evaluated solely through worst-case strategyproofness and
approximation guarantees.

% \subsection{Related Work:sample}
% \input{related_sample}

% \section{Preliminaries}
% \label{sec:prelim}
% In the preamble
% \usepackage{amsthm}

% \theoremstyle{definition}
% \newtheorem{definition}{Definition}[section]

\section{Model and the Product-Gap Mechanism}
\label{sec:model-mechanism}
We consider the $k$-facility location problem on the real line. There are $n$ agents, indexed by $N=[n]$, and the mechanism
locates $k$ facilities, where $n\geq k\geq 2$. Agent $i$ has her private location $\ell_i$ and reports a
location $p_i\in\R$. We denote the reported location profile by
$\loc=(p_1,\ldots,p_n)\in\R^n$.
An outcome is a vector of facility locations
$\f=(f_1,\ldots,f_k)\in\R^k$. We order the facilities from left to
right, so $f_1\leq f_2\leq\cdots\leq f_k$. The cost of agent $i$ is
her distance to the nearest facility, namely
$\costi(p_i,\f)=\min_{j\in[k]}\abs{p_i-f_j}$. The social cost of
$\f$ on profile $\loc$ is
$\SC(\loc,\f)=\sum_{i\in N}\costi(p_i,\f)$.
The optimal social cost for profile $\loc$ is
$\mathrm{OPT}_k(\loc)
\eqdef
\min_{\f\in\R^k}\SC(\loc,\f)$.

A randomized mechanism $\mech$ maps each reported profile $\loc$ to
a distribution over facility locations. We say that $\mech$ is an
$\alpha$-approximation if, for every profile $\loc$, $\Ex[\f\sim\mech(\loc)]{\SC(\loc,\f)}
\leq
\alpha\cdot\mathrm{OPT}_k(\loc)$.

We use strategyproofness to mean strategyproofness in expectation.
To define it, fix a true location profile. Under truthful
reporting, agent $i$ reports $\ell_i$. If she deviates to
$p_i\in\R$, the reported profile is
$(p_i,\loc_{-i})$, while her cost continues to be evaluated at
her true location $\ell_i$. A randomized mechanism $\mech$ is
\emph{strategyproof} if, for every agent $i$ with true location $\ell_i$
and every alternative report $p_i\in\R$, $\Ex[\f\sim\mech(\loc)]
    {\costi(\ell_i,\f)}
\leq
\Ex[\f\sim\mech(p_i,\loc_{-i})]
    {\costi(\ell_i,\f)}$. Under truthful reporting, we can assume without loss of generality that $\ell_i=p_i$

\paragraph{The Product-Gap Mechanism}
\label{subsec:product-gap}

We now introduce the Product-Gap mechanism. For every set
$S\in\binom{N}{k}$, let
$p_1(S)\leq p_2(S)\leq\cdots\leq p_k(S)$ denote the ordered reported
locations of the agents in $S$. The weight of $S$ is the product of
the consecutive gaps between these locations:

$$
W_k(S)
\eqdef
\prod_{j=1}^{k-1}
\bigl(p_{j+1}(S)-p_j(S)\bigr).
$$

Let $Z_k(\loc)$ denote the total weight of all sets of $k$ agents:

$$
Z_k(\loc)
\eqdef
\sum_{S\in\binom{N}{k}}W_k(S).
$$

\begin{definition}[Product-Gap mechanism]
\label{def:product-gap}
Given a reported location profile $\loc$, the Product-Gap mechanism
$\mech_{\mathrm{PG},k}$ operates as follows.

If $Z_k(\loc)>0$, the mechanism selects a set
$S\in\binom{N}{k}$ with probability $W_k(S)/Z_k(\loc)$. For the
selected set $S$, it locates the $j$-th facility at
$f_j=p_j(S)$ for every $j\in[k]$.

If $Z_k(\loc)=0$, the mechanism returns an arbitrary outcome $\f$
satisfying $\SC(\loc,\f)=0$, according to a fixed rule.
\end{definition}

Thus, the probability of selecting a set of $k$ reported locations is
proportional to the product of the $k-1$ consecutive gaps between
them. When $k$ is clear from the context, we write
$\mech_{\mathrm{PG}}$, $W(S)$, and $Z(\loc)$ instead of
$\mech_{\mathrm{PG},k}$, $W_k(S)$, and $Z_k(\loc)$, respectively.

% $\bullet$ Convention: bold face script for vectors \\
% $\bullet$ Facility location: mechanism $M$ outputs location $F_M(L)\in \reals^d$\\
% $\bullet$ $L_q-norm$: $\qnorm{x},\qnorm{X_i^j}$ $\Vert \cdot \Vert_q$, $L_\infty-norm$: $\Vert \cdot \Vert_\infty$ \\
% $\bullet$ median point: $\m = (m_1,m_2,\cdots,m_d)\in \reals^d$ \\
% $\bullet$ optimal location: $\f= (f_1,f_2,\cdots,f_d)\in \reals^d$ \\
% $\bullet$ Preferred location of agent $i$: $\loc_i=(l_1,l_2,\cdots,l_d)\in \mathbf{R}^d$ \\
% $\bullet$ Upper and lower bound of approximation ratio depending on q and d: $UB(q,d), LB(q,d)$. $UB(q)=UB(q,+\infty)$

\section{Approximation Guarantees}
\label{sec:approximation}
% In the preamble, if these environments are not already defined
% \theoremstyle{plain}
% \newtheorem{theorem}{Theorem}[section]
% \newtheorem{lemma}[theorem]{Lemma}
\newtheorem{corollary}[theorem]{Corollary}

% \section{Approximation}
% \label{sec:approximation}

In this section, we analyze the Product-Gap mechanism under truthful
reporting. We show that, for every $k\geq 2$, its approximation ratio
is $2k$, and that this guarantee is tight.

For the analysis, relabel the agents so that
$p_1\leq p_2\leq\cdots\leq p_n$, and define the elementary gaps
$g_q=p_{q+1}-p_q$ for every $q\in[n-1]$. We write $\E[\SC]$ for the
expected social cost of $\mech_{\mathrm{PG}}$ on the profile $\loc$.

\begin{theorem}[Approximation guarantee]
\label{thm:2k-approximation}
For every $k\geq 2$, the Product-Gap mechanism is a
$2k$-approximation for social cost. That is, for every truthful
profile $\loc$,
$
\Ex[\f\sim\mech_{\mathrm{PG}}(\loc)]{\SC(\loc,\f)}
\leq
2k\cdot\mathrm{OPT}_k(\loc).
$
\end{theorem}

If $Z_k(\loc)=0$, the mechanism returns an outcome of social cost zero,
and hence Theorem~\ref{thm:2k-approximation} is immediate. We therefore
assume $Z_k(\loc)>0$ throughout the upper-bound proof.

\paragraph{A random-cut representation.}

Consider a selected set $S=\{i_1,\ldots,i_k\}$ with
$i_1<\cdots<i_k$. For every $j\in[k-1]$,
$
p_{i_{j+1}}-p_{i_j}
=
\sum_{q=i_j}^{i_{j+1}-1}g_q.
$
Expanding the product of these $k-1$ sums gives
$
W_k(S)
=
\sum_{\substack{i_j\leq c_j<i_{j+1}\\
j=1,\ldots,k-1}}
\prod_{j=1}^{k-1}g_{c_j}.
$
Every vector appearing in this expansion satisfies
$1\leq c_1<\cdots<c_{k-1}\leq n-1$. Let $\mathcal C_{k-1}$ be the
set of all such cut vectors
$\vect{c}=(c_1,\ldots,c_{k-1})$. For
$\vect{c}\in\mathcal C_{k-1}$, set $c_0=0$ and $c_k=n$, and define

$$
I_j(\vect{c})
=
\{c_{j-1}+1,\ldots,c_j\},
\qquad
m_j(\vect{c})
=
c_j-c_{j-1}
$$

for every $j\in[k]$. We further define

$$
A_{\vect{c}}
=
\prod_{j=1}^k m_j(\vect{c}),
\qquad
G_{\vect{c}}
=
\prod_{j=1}^{k-1}g_{c_j}.
$$

\begin{lemma}[Normalization by cuts]
\label{lem:cut-normalization}
The normalization of the Product-Gap mechanism satisfies

$$
Z_k(\loc)
=
\sum_{\vect{c}\in\mathcal C_{k-1}}
A_{\vect{c}}G_{\vect{c}}.
$$
\end{lemma}

\begin{proof}
Fix a cut vector $\vect{c}$. The monomial $G_{\vect{c}}$ appears in
the expansion of $W_k(S)$ exactly when $S$ contains one agent from
each block $I_j(\vect{c})$. There are

$$
\prod_{j=1}^k m_j(\vect{c})
=
A_{\vect{c}}
$$

such sets. Summing first over selected sets and then over cut vectors
proves the identity.
\end{proof}

\begin{corollary}[Random-cut representation]
\label{cor:random-cut}
The Product-Gap mechanism can be sampled as follows. First choose a
cut vector $\vect{c}\in\mathcal C_{k-1}$ with probability
$A_{\vect{c}}G_{\vect{c}}/Z_k(\loc)$. Then, independently in each
block $I_j(\vect{c})$, choose one agent uniformly and open a facility
at each chosen report.
\end{corollary}

\begin{proof}
Conditioned on $\vect{c}$, every set containing one agent from each
block has probability $1/A_{\vect{c}}$. Its joint probability with
$\vect{c}$ is therefore $G_{\vect{c}}/Z_k(\loc)$. For a fixed set
$S$, summing over all cut vectors compatible with $S$ gives
$W_k(S)/Z_k(\loc)$ by the product expansion above.
\end{proof}

\paragraph{The cost for a fixed cut vector.}

We first bound the expected cost generated inside one block.

\begin{lemma}[Uniform block representative]
\label{lem:block-representative}
Let $I=\{s,s+1,\ldots,t\}$ be a nonempty block of size
$m=t-s+1$, and let $R$ be chosen uniformly from $I$. Then

\[
\begin{aligned}
\Ex[R]{\sum_{i=s}^t\abs{p_i-p_R}}
&=
\frac{1}{m}
\sum_{r=s}^t\sum_{i=s}^t\abs{p_i-p_r} \\
&=
\frac{2}{m}
\sum_{q=s}^{t-1}(q-s+1)(t-q)g_q.
\end{aligned}
\]
\end{lemma}

\begin{proof}
The gap $g_q$ separates $q-s+1$ agents on its left from $t-q$
agents on its right. Each ordered pair separated by this gap
contributes $g_q$, and every unordered pair is counted in both
directions.
\end{proof}

Conditioned on a cut vector $\vect{c}$, every agent can be served by
the representative chosen from her own block. Hence,
Lemma~\ref{lem:block-representative} gives

\[
\E[\SC\mid\vect{c}]
\leq
\sum_{j=1}^k
\frac{2}{m_j(\vect{c})}
\sum_{r=c_{j-1}+1}^{c_j-1}
(r-c_{j-1})(c_j-r)g_r.
\]

\paragraph{Expansion over service gaps.}

We next combine the block cost with the probability of the cut vector.

\begin{lemma}[Service-gap expansion]
\label{lem:service-gap-expansion}
For $1\leq q_1<\cdots<q_k\leq n-1$, set $q_0=0$,
$q_{k+1}=n$, and define

$$
Q(q_1,\ldots,q_k)
\eqdef
\prod_{h=0}^k(q_{h+1}-q_h).
$$

Then

\[
Z_k(\loc)\E[\SC]
\leq
2k
\sum_{1\leq q_1<\cdots<q_k\leq n-1}
Q(q_1,\ldots,q_k)
\prod_{h=1}^k g_{q_h}.
\]
\end{lemma}

\begin{proof}
Multiply the conditional-cost bound by the unnormalized cut mass
$A_{\vect{c}}G_{\vect{c}}$ and sum over $\vect{c}$. Each resulting
term is indexed by a triple $(\vect{c},j,r)$ satisfying
$c_{j-1}<r<c_j$.

Adjoin the service gap $r$ to the $k-1$ cut gaps and write the
resulting ordered tuple as $q_1<\cdots<q_k$. Suppose that $r=q_h$.
Deleting $q_h$ merges the two adjacent intervals of sizes
$q_h-q_{h-1}$ and $q_{h+1}-q_h$. Since
$A_{\vect{c}}/m_j(\vect{c})$ is the product of the remaining block
sizes,

$$
\frac{A_{\vect{c}}}{m_j(\vect{c})}
(r-c_{j-1})(c_j-r)
=
Q(q_1,\ldots,q_k).
$$

Conversely, every ordered $k$-tuple $(q_1,\ldots,q_k)$ and every
choice $h\in[k]$ uniquely determine the triple obtained by taking
$r=q_h$ and deleting $q_h$ to form the cut vector. Thus every ordered
$k$-tuple occurs exactly $k$ times, once for each choice of its
service gap. This proves the bound.
\end{proof}

\paragraph{Decomposing the optimal cost.}

The optimal social cost has a convenient representation in terms of
the same elementary gaps.

\begin{lemma}[Contiguous optimal clustering]
\label{lem:contiguous-optimum}
There is an optimal $k$-facility solution whose nonempty service
clusters are contiguous index intervals. Moreover, the minimum
one-facility cost of a cluster $\{a,a+1,\ldots,b\}$ is

$$
\sum_{q=a}^{b-1}
\min\{q-a+1,b-q\}g_q.
$$
\end{lemma}

\begin{proof}
Order the facilities from left to right and assign every agent to a
nearest facility, breaking ties consistently from left to right. The
nonempty service regions are then contiguous intervals.

For a fixed interval $\{a,\ldots,b\}$, a median minimizes the sum of
distances. A gap $g_q$ to the left of a chosen median is crossed by
the $q-a+1$ agents on its left, while a gap to its right is crossed
by the $b-q$ agents on its right. At a median, the relevant
coefficient is the smaller of these two quantities.
\end{proof}

Fix an optimal contiguous clustering. For each $q\in[n-1]$, set
$\mu_q=0$ if $q$ lies between two optimal clusters. If $q$ lies
inside the optimal cluster $\{a,\ldots,b\}$, set
$\mu_q=\min\{q-a+1,b-q\}$. Lemma~\ref{lem:contiguous-optimum}
gives
$
\mathrm{OPT}_k(\loc)
=
\sum_{q=1}^{n-1}\mu_qg_q.
$

\paragraph{Charging to the optimum.}

We now compare each term in
Lemma~\ref{lem:service-gap-expansion} with the optimal cost.

\begin{lemma}[Charging inequality]
\label{lem:charging}
For every $1\leq q_1<\cdots<q_k\leq n-1$,

$$
Q(q_1,\ldots,q_k)
\leq
\sum_{j=1}^k
\mu_{q_j}A_{(q_1,\ldots,q_k)\setminus q_j},
$$

where, writing $a_h=q_{h+1}-q_h$ for $h=0,\ldots,k$,

$$
A_{(q_1,\ldots,q_k)\setminus q_j}
\eqdef
(a_{j-1}+a_j)
\prod_{h\notin\{j-1,j\}}a_h.
$$
\end{lemma}

\begin{proof}
The cuts $q_1,\ldots,q_k$ partition the agents into $k+1$ intervals

$$
J_h
=
\{q_h+1,\ldots,q_{h+1}\},
\qquad
h=0,\ldots,k.
$$

Since $\abs{J_h}=a_h$, the quantity
$
Q(q_1,\ldots,q_k)
=
\prod_{h=0}^k a_h
$
counts the tuples obtained by choosing one agent from each $J_h$.
Color every agent by her cluster in the fixed optimal solution. These
colors form at most $k$ contiguous runs. Hence every tuple of $k+1$
chosen agents contains two agents from adjacent intervals
$J_{j-1}$ and $J_j$ with the same color. Charge the tuple to the
smallest such $j$.
Fix $j$. If $q_j$ is a boundary between two optimal clusters, no tuple
is charged to $j$, and $\mu_{q_j}=0$. Otherwise, let $L$ and $R$ be
the numbers of agents in the optimal cluster crossing $q_j$ that lie
weakly to the left and strictly to the right of $q_j$, respectively.
Then $\mu_{q_j}=\min\{L,R\}$.
Let $\ell$ and $r$ be the numbers of agents of this cluster in
$J_{j-1}$ and $J_j$. We have
$\ell\leq\min\{L,a_{j-1}\}$ and
$r\leq\min\{R,a_j\}$. If $L\leq R$, then
$\ell r\leq La_j=\mu_{q_j}a_j$. If $R\leq L$, then
$\ell r\leq Ra_{j-1}=\mu_{q_j}a_{j-1}$. Thus,
$
\ell r
\leq
\mu_{q_j}(a_{j-1}+a_j).
$
The choices in all other intervals contribute
$\prod_{h\notin\{j-1,j\}}a_h$. Hence the number of tuples charged to
$j$ is at most
$\mu_{q_j}A_{(q_1,\ldots,q_k)\setminus q_j}$. Summing over $j$
proves the claim.
\end{proof}

Multiplying the charging inequality by $\prod_{j=1}^k g_{q_j}$ and
summing over all ordered $k$-tuples yields

\[
\begin{aligned}
&\sum_{1\leq q_1<\cdots<q_k\leq n-1}
Q(q_1,\ldots,q_k)
\prod_{j=1}^k g_{q_j} \\
&\leq
\sum_{q=1}^{n-1}\mu_qg_q
\sum_{\substack{\vect{c}\in\mathcal C_{k-1}\\
q\notin\{c_1,\ldots,c_{k-1}\}}}
A_{\vect{c}}G_{\vect{c}} \\
&\leq
\left(\sum_{q=1}^{n-1}\mu_qg_q\right)
\left(
\sum_{\vect{c}\in\mathcal C_{k-1}}
A_{\vect{c}}G_{\vect{c}}
\right) \\
&=
\mathrm{OPT}_k(\loc)Z_k(\loc).
\end{aligned}
\]

\begin{proof}[Proof of Theorem~\ref{thm:2k-approximation}]
If $Z_k(\loc)=0$, the mechanism has social cost zero. Otherwise,
Lemma~\ref{lem:service-gap-expansion} and the charging bound above give
$
Z_k(\loc)\E[\SC]
\leq
2k\cdot\mathrm{OPT}_k(\loc)Z_k(\loc).
$
Dividing by $Z_k(\loc)>0$ proves the theorem.
\end{proof}

\subsection{Tightness}
\label{subsec:2k-tightness}

The factor $2k$ cannot be improved for the Product-Gap mechanism.

\begin{theorem}[Tightness]
\label{thm:2k-tightness}
For every $k\geq 2$ and every $\eps>0$, there is a location profile
$\loc$ for which

$$
\frac{
\Ex[\f\sim\mech_{\mathrm{PG}}(\loc)]{\SC(\loc,\f)}
}{
\mathrm{OPT}_k(\loc)
}
>
2k-\eps.
$$
\end{theorem}

\begin{proof}
Fix an integer $M\geq 1$. Place one agent at $0$ and $M$ agents at
each location $1,2,\ldots,k$. Every positive-weight selected set uses
$k$ of these $k+1$ locations and therefore omits exactly one location.

If location $0$ is omitted, there are $M^k$ choices, the Product-Gap
weight is $1$, and the social cost is $1$. If location $k$ is omitted,
there are $M^{k-1}$ choices, the weight is $1$, and the social cost is
$M$. If an interior location $r\in\{1,\ldots,k-1\}$ is omitted, there
are $M^{k-1}$ choices, the weight is $2$, and the social cost is $M$.
Consequently,

\[
\begin{aligned}
Z_k(\loc)
&=
M^k+M^{k-1}+2(k-1)M^{k-1} \\
&=
M^{k-1}(M+2k-1), \\
Z_k(\loc)\E[\SC]
&=
M^k+M^k+2(k-1)M^k \\
&=
2kM^k.
\end{aligned}
\]

Therefore,

$$
\E[\SC]
=
\frac{2kM}{M+2k-1}.
$$

Opening facilities at $1,2,\ldots,k$ has social cost $1$. Conversely,
consider any solution with $k$ facilities and assign each of the
$k+1$ occupied locations to a nearest facility. Two occupied locations
must be assigned to the same facility. The combined cost of one agent
at each of these locations is at least the distance between them,
which is at least $1$. Hence $\mathrm{OPT}_k(\loc)=1$.

Letting $M$ tend to infinity makes the approximation ratio tend to
$2k$, proving the theorem.
\end{proof}

\newtheorem{proposition}[theorem]{Proposition}

\section{Strategyproofness for $k=2,3$ and Failure for $k\geq4$}
\label{sec:strategyproofness}

We now determine the exact incentive boundary of the Product-Gap
mechanism on the line.

\begin{theorem}[Strategyproofness classification]
\label{thm:sp-classification}
The Product-Gap mechanism is strategyproof in expectation for
$k=2$ and $k=3$. For every $k\geq4$, the mechanism is not
strategyproof.
\end{theorem}

The proof has three parts. We first establish strategyproofness for
three facilities through a closest-point deletion argument. We then
derive the two-facility case from a far-anchor limit. Finally, we give
an exact manipulation for four facilities and lift it to every larger
number of facilities.

\subsection{Strategyproofness for Three Facilities}
\label{subsec:sp-k3}

\begin{theorem}[Three-facility strategyproofness]
\label{thm:sp-k3}
The Product-Gap mechanism is strategyproof in expectation for $k=3$.
\end{theorem}

Fix a strategic agent with true location $y$, an alternative report
$r$, and the collection $O$ of reports of the other agents. For a
selected set $S$, write

$$
d(y,S)=\min_{p\in S}\abs{y-p}.
$$

We first dispose of the zero-normalization cases. If the truthful
three-set normalization is zero, the mechanism opens at every occupied
truthful location, including $y$, and the strategic agent has cost
zero. Suppose instead that the truthful normalization is positive but
the normalization after reporting $r$ is zero. Then $O\cup\{r\}$
occupies at most two distinct locations, whereas $O\cup\{y\}$ occupies
at least three. Hence $O$ occupies exactly two distinct locations,
$r$ is already occupied, and $y$ is a third location. Every
positive-weight truthful triple must use all three locations and
therefore contains the strategic report at $y$. The truthful cost is
again zero.

It remains to consider the case in which both normalizations are
positive. For a pair $A\in\binom{O}{2}$ and a report
$a\in\{y,r\}$, define

$$
W_a(A)=W_3(A\cup\{a\}),
\qquad
c_A=d(y,A\cup\{r\}).
$$

For a triple $F\in\binom{O}{3}$, let $d_F=d(y,F)$. We use the
following four aggregate quantities:

$$
\begin{aligned}
Z(O)
&=
\sum_{F\in\binom{O}{3}}W_3(F),
&&
\text{the mass of triples using only ordinary reports},\\
B(O)
&=
\sum_{F\in\binom{O}{3}}W_3(F)d_F,
&&
\text{their cost-weighted mass},\\
S_a(O)
&=
\sum_{A\in\binom{O}{2}}W_a(A),
&&
\text{the mass of triples using report $a$},\\
T_r(O)
&=
\sum_{A\in\binom{O}{2}}W_r(A)c_A,
&&
\text{the cost-weighted mass of triples using $r$}.
\end{aligned}
$$

A truthful selected triple containing the strategic report at $y$
gives the agent cost zero. Her truthful and deviating expected costs
are therefore

$$
C_y(y\mid O)
=
\frac{B(O)}{Z(O)+S_y(O)}
$$

and

$$
C_y(r\mid O)
=
\frac{B(O)+T_r(O)}{Z(O)+S_r(O)}.
$$

Cross-multiplication gives

\[
C_y(r\mid O)-C_y(y\mid O)
=
\frac{
Z(O)T_r(O)
+S_y(O)T_r(O)
-B(O)\bigl(S_r(O)-S_y(O)\bigr)
}{
\bigl(Z(O)+S_r(O)\bigr)
\bigl(Z(O)+S_y(O)\bigr)
}.
\]

Since $S_y(O)T_r(O)\geq0$, it is enough to prove

$$
\Psi(O)
\eqdef
Z(O)T_r(O)
-
B(O)\bigl(S_r(O)-S_y(O)\bigr)
\geq0.
$$

The proof is based on deleting the ordinary report closest to the
strategic agent.

\begin{lemma}[Closest-point deletion]
\label{lem:k3-closest-point-deletion}
Assume first that the profile is in generic position: no ordinary
report equals $y$, the ordinary report closest to $y$ is unique, and
all coordinate and radial comparisons used in the proof are strict.
Let $p\in O$ be closest to $y$, and let $Q=O\setminus\{p\}$. Then

$$
\Psi(O)\geq\Psi(Q).
$$
\end{lemma}

The proof of Lemma~\ref{lem:k3-closest-point-deletion} is given in
Appendix~\ref{app:k3-proof}. The proof first
decomposes the deletion difference into reduced atoms involving at most
four witness reports. Deviations toward the deleted report are
termwise nonnegative. Deviations through the true location are handled
by a radial certificate, and deviations outside the closest-point
interval are reduced to its boundary by an exterior-clipping
certificate.

\begin{proof}[Proof of Theorem~\ref{thm:sp-k3}]
Consider first a generic profile. Repeatedly delete an ordinary report
closest to $y$. By Lemma~\ref{lem:k3-closest-point-deletion}, the
potential cannot decrease when the deleted reports are restored.

Once fewer than three ordinary reports remain, there are no ordinary
triples. Hence $Z(O)=B(O)=0$ and therefore $\Psi(O)=0$. Reversing
the sequence of deletions gives

$$
\Psi(O)\geq0
$$

for every generic profile. The expected-cost identity above then
implies

$$
C_y(r\mid O)\geq C_y(y\mid O).
$$

For fixed labels, every order statistic, consecutive-gap product, and
nearest-distance term is continuous in all reported coordinates.
Consequently, $Z(O)$, $B(O)$, $S_a(O)$, $T_r(O)$, and
$\Psi(O)$ are continuous. Approximating an arbitrary profile by
generic profiles and taking a limit handles repeated reports,
nonunique closest reports, coordinate ties, and radial ties. Since the
agent, her true location, her deviation, and the other reports were
arbitrary, the mechanism is strategyproof in expectation.
\end{proof}

\subsection{Strategyproofness for Two Facilities}
\label{subsec:sp-k2}

The two-facility result follows from the three-facility result through
a general descent property.

For a strategic agent with true location $y$, report $a$, and other
reports $O$, let $C_k(y;a\mid O)$ denote her expected cost under the
$k$-facility Product-Gap mechanism.

\begin{lemma}[Far-anchor descent]
\label{lem:far-anchor-descent}
Fix $k\geq2$. If the $(k+1)$-facility Product-Gap mechanism is
strategyproof on every finite profile, then the $k$-facility
Product-Gap mechanism is also strategyproof.
\end{lemma}

\begin{proof}
Fix the other reports $O$, a true location $y$, and an alternative
report $r$. We first assume that the $k$-facility normalizations are
positive under both reports.

Choose $L$ larger than every point in $O\cup\{y,r\}$ and add one
ordinary agent at $L$. For $a\in\{y,r\}$, let

$$
D_k(a)
=
\sum_{S\in\binom{O\cup\{a\}}{k}}W_k(S)
$$

be the $k$-facility normalization, and let

$$
N_k(a)
=
\sum_{S\in\binom{O\cup\{a\}}{k}}
W_k(S)d(y,S)
$$

be the corresponding cost numerator.

Every selected $(k+1)$-set containing the remote anchor is uniquely
of the form $S\cup\{L\}$, where
$S\in\binom{O\cup\{a\}}{k}$. Its weight is

$$
W_{k+1}(S\cup\{L\})
=
W_k(S)\bigl(L-\max S\bigr).
$$

The total contribution of selected sets not containing $L$ is
independent of $L$. Moreover, for all sufficiently large $L$, the
remote anchor is never the closest selected facility to $y$ whenever
another facility is present. Therefore, the normalization
$D_{k+1}^{L}(a)$ and cost numerator $N_{k+1}^{L}(a)$ on the augmented
profile satisfy

$$
D_{k+1}^{L}(a)
=
L D_k(a)+O(1)
$$

and

$$
N_{k+1}^{L}(a)
=
L N_k(a)+O(1).
$$

It follows that

$$
\lim_{L\to\infty}
C_{k+1}(y;a\mid O\cup\{L\})
=
\frac{N_k(a)}{D_k(a)}
=
C_k(y;a\mid O).
$$

By strategyproofness of the $(k+1)$-facility mechanism, for every
sufficiently large $L$,

\[
C_{k+1}(y;r\mid O\cup\{L\})
\geq
C_{k+1}(y;y\mid O\cup\{L\}).
\]

Taking $L\to\infty$ proves

$$
C_k(y;r\mid O)\geq C_k(y;y\mid O).
$$

The zero-normalization cases are handled separately. If the truthful
normalization is zero, the truthful cost is zero. If the deviating
normalization is zero while the truthful normalization is positive,
the ordinary reports occupy exactly $k-1$ distinct locations, the
deviating report is already occupied, and the true location is not.
Every positive-weight truthful $k$-set must therefore contain the
truthful strategic report, so the truthful cost is again zero.
\end{proof}

\begin{corollary}[Two-facility strategyproofness]
\label{cor:sp-k2}
The Product-Gap mechanism is strategyproof in expectation for $k=2$.
\end{corollary}

\begin{proof}
Apply Lemma~\ref{lem:far-anchor-descent} with $k=2$ and use
Theorem~\ref{thm:sp-k3}.
\end{proof}

\subsection{A Profitable Deviation for Four Facilities}
\label{subsec:not-sp-k4}

We next show that strategyproofness already fails for four facilities.
Although the manipulation is small, its sign is certified by exact
rational arithmetic.

\begin{proposition}[Exact four-facility manipulation]
\label{prop:k4-manipulation}
Consider a strategic agent with true location $0$ and alternative
report
$
h=\frac{1}{1000}$.
The other reports consist of:
\begin{itemize}
    \item $100$ agents at $0$;
    \item one agent at $-1$ and one agent at $1$;
    \item for every $j\in\{1,\ldots,100\}$, exactly $2^j$ agents at
    the location
    $-1-\frac{1}{100\cdot2^j}$.
\end{itemize}

On this profile, reporting $h$ gives the strategic agent strictly
smaller expected cost than reporting $0$.
\end{proposition}

\begin{proof}
For $a\in\{0,h\}$, define the four-facility normalization and cost
numerator by

$$
U_4(a)
=
\sum_{S\in\binom{O\cup\{a\}}{4}}W_4(S)
$$

and

$$
M_4(a)
=
\sum_{S\in\binom{O\cup\{a\}}{4}}
W_4(S)d(0,S).
$$

Both normalizations are positive. Truthful reporting is weakly better
than reporting $h$ if and only if

$$
\Phi_4(O;h)
\eqdef
U_4(0)M_4(h)-U_4(h)M_4(0)
\geq0.
$$

Appendix~\ref{app:k4-certificate} evaluates the four quantities by an
exact support recurrence and obtains

$$
\Phi_4(O;h)
=
-\frac{N_\Phi}{D_\Phi}
<0,
$$

where $N_\Phi$ and $D_\Phi$ are explicit positive integers listed
there. Thus,

$$
\frac{M_4(h)}{U_4(h)}
<
\frac{M_4(0)}{U_4(0)}.
$$

For orientation, the two expected costs are approximately

$$
\frac{M_4(0)}{U_4(0)}
\approx
0.00724556377223489236
$$

and

$$
\frac{M_4(h)}{U_4(h)}
\approx
0.00724554178194980469.
$$

The exact negative determinant, rather than these decimal
approximations, certifies the strict manipulation.
\end{proof}

\subsection{Failure for Every $k\geq4$}
\label{subsec:not-sp-kge4}

The four-facility counterexample can be lifted to any larger number of
facilities by adding increasingly remote anchors.

\begin{lemma}[Far-anchor padding]
\label{lem:far-anchor-padding}
Suppose a finite profile admits a strict profitable deviation under
the $k$-facility Product-Gap mechanism. Then, for every integer
$s\geq1$, some finite profile admits a strict profitable deviation
under the $(k+s)$-facility Product-Gap mechanism.
\end{lemma}

\begin{proof}
Let $O$, $y$, and $r$ describe a strict manipulation:

$$
C_k(y;r\mid O)
<
C_k(y;y\mid O).
$$

Choose a positive integer $R$ larger than every absolute coordinate in
the original profile and add one ordinary agent at each of

$$
R,R^2,\ldots,R^s.
$$

Let

$$
\sigma_s
=
1+2+\cdots+s
=
\frac{s(s+1)}{2}.
$$

For $a\in\{y,r\}$, consider a selected $(k+s)$-set containing all
$s$ remote anchors. It is uniquely of the form

$$
F\cup\{R,R^2,\ldots,R^s\},
$$

where $F\in\binom{O\cup\{a\}}{k}$, and its weight is

\[
W_{k+s}\bigl(F\cup\{R,R^2,\ldots,R^s\}\bigr)
=
W_k(F)
\bigl(R-\max F\bigr)
\prod_{j=2}^s
\bigl(R^j-R^{j-1}\bigr).
\]

After division by $R^{\sigma_s}$, the factor multiplying $W_k(F)$
converges to $1$, uniformly over the finitely many original sets $F$.
For sufficiently large $R$, none of the remote anchors is closest to
$y$, so the selected-set cost remains $d(y,F)$.

A selected $(k+s)$-set that omits at least one remote anchor has total
anchor degree at most $\sigma_s-1$. Since all original coordinates
remain bounded, the total normalization and cost-numerator
contributions of all such sets are $O(R^{\sigma_s-1})$. Consequently,
for each $a\in\{y,r\}$,

$$
R^{-\sigma_s}D_{k+s}^{R}(a)
\longrightarrow
D_k(a)
$$

and

$$
R^{-\sigma_s}N_{k+s}^{R}(a)
\longrightarrow
N_k(a).
$$

Therefore,

$$
C_{k+s}^{R}(y;a)
\longrightarrow
C_k(y;a\mid O).
$$

The limiting cost under report $r$ is strictly smaller than the
limiting truthful cost. Hence the same strict inequality holds for
every sufficiently large finite integer $R$.
\end{proof}

\begin{corollary}
\label{cor:not-sp-kge4}
For every $k\geq4$, the Product-Gap mechanism is not strategyproof.
\end{corollary}

\begin{proof}
Proposition~\ref{prop:k4-manipulation} establishes the result for
$k=4$. For $k=4+s$, apply
Lemma~\ref{lem:far-anchor-padding} with any $s\geq1$.
\end{proof}

\begin{proof}[Proof of Theorem~\ref{thm:sp-classification}]
Theorem~\ref{thm:sp-k3} proves strategyproofness for $k=3$, and
Corollary~\ref{cor:sp-k2} proves it for $k=2$.
Corollary~\ref{cor:not-sp-kge4} proves failure for every $k\geq4$.
\end{proof}

\section{An Improved Approximation for Two Facilities}
\label{sec:two-facility-mixture}

We now specialize the Product-Gap mechanism to two facilities and
combine it with the Proportional mechanism of Lu et al.\
\cite{LuSWZ10}.

For $k=2$, the Product-Gap mechanism chooses an unordered pair of
agents $\{i,j\}$ with probability proportional to
$\abs{p_i-p_j}$ and opens facilities at $p_i$ and $p_j$. Following
\cite{ma2026}, we refer to this specialization as the
\emph{Global Pair} mechanism and denote it by
$\mech_{\mathrm{G}}$.

The \emph{Proportional mechanism}, denoted by
$\mech_{\mathrm{P}}$, first chooses an anchor $i$ uniformly from
$N$. It then chooses a second agent $j$ with probability
$
\frac{\abs{p_i-p_j}}
{\sum_{v\in N}\abs{p_i-p_v}},
$
and opens facilities at $p_i$ and $p_j$.

For $\lambda\in[0,1]$, let $\mech_\lambda$ be the
report-independent mixture that runs $\mech_{\mathrm{P}}$ with
probability $\lambda$ and $\mech_{\mathrm{G}}$ with probability
$1-\lambda$. Both component mechanisms are strategyproof on the line,
and hence so is every fixed mixture
\cite{LuSWZ10,ma2026}.

For a randomized mechanism $\mech$, write
$\SC_{\mech}(\loc)$ for its expected social cost on the truthful
profile $\loc$. Define $
\kappa\eqdef 4\sqrt{3}-6.
$

\begin{theorem}
\label{thm:two-facility-exact-mixture}
For every $\lambda\in[0,1]$, the worst-case approximation ratio of
$\mech_\lambda$ on the real line is
$
\max\{3+\lambda,\,4-\kappa\lambda\}.
$
Consequently, the optimal fixed mixture uses
$
\lambda^\star
=
\frac{1}{1+\kappa}
=
\frac{5+4\sqrt{3}}{23}
$
and has approximation ratio
$
3+\lambda^\star
=
\frac{74+4\sqrt{3}}{23}
=
3.5186175318\ldots.
$
\end{theorem}

\subsection{Two Complementary Bounds}
\label{subsec:two-complementary-bounds}

Fix a profile $\loc$ with positive optimal social cost. Let
$o_L<o_R$ be an optimal pair of facilities. Assign each agent to a
nearest optimal facility, breaking midpoint ties arbitrarily, and let
$L$ and $R$ be the resulting left and right clusters.

Let $\delta=o_R-o_L$. For an agent $i$ in cluster $X\in\{L,R\}$,
write $
r_i=\abs{p_i-o_X}.
$
Thus, $r_i$ is agent $i$'s contribution to the optimal social cost.
Let
$
C_X=\sum_{i\in X}r_i
$
be the optimal cost contributed by cluster $X$, and put
$
C=C_L+C_R=\mathrm{OPT}_2(\loc).
$
We also associate with each cluster its \emph{truncated spread}
$
T_X
=
\sum_{i\in X}r_i\min\{r_i,\delta\}.
$
The dispersion parameter of the profile is
$
h(\loc)
=
\frac{T_L+T_R}{\delta C}.
$
Since $r_i\min\{r_i,\delta\}\leq\delta r_i$, we have
$0\leq h(\loc)\leq1$.

The proof of the following bound uses the same basic
triple-decomposition argument as Ma and Peng~\cite{ma2026}. For
completeness, we defer our self-contained specialization to the line
to Appendix.
\begin{lemma}
\label{lem:global-pair-dispersion}
For every profile on the line, $
\SC_{\mech_{\mathrm{G}}}(\loc)
\leq
\bigl(3+h(\loc)\bigr)C.
$
\end{lemma}

Our main task is to prove the complementary inequality for the
Proportional mechanism.

\begin{lemma}[Proportional dispersion bound]
\label{lem:proportional-dispersion}
For every profile on the line,

$$
\SC_{\mech_{\mathrm{P}}}(\loc)
\leq
\bigl(4-\kappa h(\loc)\bigr)C.
$$
\end{lemma}

\subsection{A Weighted Service Bound}
\label{subsec:weighted-service-bound}

For points $x,y\in\R$, write $d(x,y)=\abs{x-y}$, and write
$d(x,\{y,z\})=\min\{d(x,y),d(x,z)\}$.

We first record a simple bound for serving one cluster when one
facility is already fixed.

\begin{lemma}[Weighted service of one cluster]
\label{lem:weighted-cluster-service}
Let $K$ be a nonempty set of agents, let $o$ be a median of their
locations, and let
$
C_K=\sum_{i\in K}d(p_i,o).
$
For every point $z\in\R$,

\[
\sum_{j\in K}d(z,p_j)
\sum_{i\in K}d(p_i,\{z,p_j\})
\leq
3C_K\sum_{j\in K}d(z,p_j).
\]
\end{lemma}

\begin{proof}
Consider an unordered pair $\{i,j\}\subseteq K$. Put
$x=d(z,p_i)$, $y=d(z,p_j)$, and
$u=d(p_i,p_j)$, and assume $x\leq y$. The contribution of this pair
to the left-hand side is
$
x\min\{y,u\}+y\min\{x,u\}.
$
A direct check of the three cases $u\leq x$, $x<u\leq y$, and
$y<u$ gives
$
x\min\{y,u\}+y\min\{x,u\}
\leq
2xu+\min\{x,u\}^2.
$
Let $r_i=d(p_i,o)$ and $r_j=d(p_j,o)$. Since
$u\leq r_i+r_j$ and $x\leq y$,
$
xu\leq yr_i+xr_j.
$
Summing this inequality over all unordered pairs gives at most
$
C_K\sum_{j\in K}d(z,p_j).
$
Moreover, $\min\{x,u\}^2\leq xu$. Thus, the first term contributes
at most two copies of this quantity and the second term contributes at
most one further copy.
\end{proof}

\subsection{The Gain from a Fixed Anchor}
\label{subsec:fixed-anchor-gain}

We analyze anchors in the right optimal cluster $R$; the argument for
$L$ is obtained by reflection. Let $m=\abs{R}$ and
$\ell=\abs{L}$.

For an anchor $i\in R$, define
$
D_i
=
\sum_{j\in R}d(p_i,p_j)
$
and
$
E_i
=
\sum_{h\in L}d(p_i,p_h).
$
Thus, $D_i$ is the total distance from the anchor to its own optimal
cluster, $E_i$ is its total distance to the opposite cluster, and
$
S_i=D_i+E_i
$
is the normalization used by the Proportional mechanism after choosing
$i$ as the anchor. Let $F_i$ denote the expected social cost
conditioned on choosing $i$ first.

When the second facility is also chosen in $R$, moving that facility
away from the anchor may reduce the service cost of the agents in
$L$. We call this reduction the \emph{routing gain}:

\[
\operatorname{Route}(i)
=
\sum_{j\in R}d(p_i,p_j)
\sum_{h\in L}
\bigl[d(p_i,p_h)-d(p_j,p_h)\bigr]_+,
\]

where $[x]_+=\max\{x,0\}$.

We also define the following guaranteed gain from serving the anchor's
own cluster:

\[
\operatorname{Base}(i)
=
\max\left\{
D_i(D_i+3C_L),
\;
D_i\bigl[2D_i-3(C_R-C_L)\bigr]_+
\right\}.
\]

The first expression comes from the fact that the anchor alone can
serve its own cluster at cost $D_i$. The second comes from applying
Lemma~\ref{lem:weighted-cluster-service} to that cluster.

\begin{lemma}[Fixed-anchor bound]
\label{lem:fixed-anchor-bound}
For every anchor $i\in R$,

$$
F_i
\leq
2D_i+3C_L
-
\frac{\operatorname{Base}(i)+\operatorname{Route}(i)}
{S_i}.
$$
\end{lemma}

\begin{proof}
Suppose first that the second selected agent $j$ also belongs to
$R$. By Lemma~\ref{lem:weighted-cluster-service}, the weighted
service cost of $R$ is at most $3C_RD_i$. It is also at most
$D_i^2$, since for every choice of $j$, all agents in $R$ can be
served from the anchor $p_i$ alone.

The weighted service cost of $L$ is exactly
$
D_iE_i-\operatorname{Route}(i).
$
If the second selected agent belongs to $L$, the weighted service cost
of $L$ is at most $3C_LE_i$ by
Lemma~\ref{lem:weighted-cluster-service}. The weighted service cost of
$R$ is at most $D_iE_i$, because the agents in $R$ may all be served
from the anchor.
Consequently,

\[
S_iF_i
\leq
\min\{3C_RD_i,D_i^2\}
+3C_LE_i
+2D_iE_i
-\operatorname{Route}(i).
\]

Compare this bound with $S_i(2D_i+3C_L)$. If we use the bound
$3C_RD_i$ for the first term, the difference is
$
D_i\bigl(2D_i-3(C_R-C_L)\bigr)
+\operatorname{Route}(i).
$
If we instead use $D_i^2$, the difference is
$
D_i(D_i+3C_L)+\operatorname{Route}(i).
$
Taking the better of these two estimates proves the lemma.
\end{proof}

\subsection{Why Outward Anchors Are Better}
\label{subsec:outward-anchors}

Because $o_R$ is a median of the agents in $R$, at most half of the
agents in $R$ lie strictly to its right. We call such agents
\emph{outward}, since they lie on the side of $o_R$ away from the
other optimal cluster. The remaining agents are called
\emph{nonoutward}.

The next lemma is the point where the order structure of the line is
used.

\begin{lemma}[Squared routing gain]
\label{lem:squared-routing}
If $i\in R$ is outward, then
$
\operatorname{Route}(i)
\geq
\ell r_i(D_i-C_R).
$
\end{lemma}

\begin{proof}
Translate the line so that $o_R=0$. Since $i$ is outward, we may
write $p_i=r_i>0$. Every agent of $L$ lies to the left of every
agent of $R$. Therefore,
$
\operatorname{Route}(i)
=
\ell
\sum_{\substack{j\in R\\p_j<p_i}}
(p_i-p_j)^2.
$
For every $j\in R$, let $r_j=\abs{p_j}$. We claim that
$
r_i\bigl(d(p_i,p_j)-r_j\bigr)
\leq
\ind{p_j<p_i}\,d(p_i,p_j)^2.
$

There are three cases. If $p_j\geq p_i$, the left-hand side is
$-r_i^2\leq0$. If $0\leq p_j<p_i$, the difference between the
right-hand side and the left-hand side is $p_j^2$. Finally, if
$p_j<0$, the left-hand side is $r_i^2$, while the right-hand side is
$(r_i+r_j)^2$.
Summing over $j\in R$ gives

$$
\frac{\operatorname{Route}(i)}{\ell}
\geq
r_i\sum_{j\in R}\bigl(d(p_i,p_j)-r_j\bigr)
=
r_i(D_i-C_R).
$$
\end{proof}

\subsection{From Anchor Bounds to a Cluster Bound}
\label{subsec:cluster-bound}

Define the total gain credited to anchors in $R$ by

\[
\operatorname{Gain}(R)
=
\sum_{i\in R}
\left(
4C_R-2D_i+
\frac{\operatorname{Base}(i)}{S_i}
\right)
+
\sum_{\substack{i\in R\\i\text{ outward}}}
\frac{\operatorname{Route}(i)}{S_i}.
\]

The analogous quantity $\operatorname{Gain}(L)$ is defined after
reflecting the line.

\begin{lemma}[Global decomposition]
\label{lem:global-proportional-decomposition}
The expected cost of the Proportional mechanism satisfies

\[
n\bigl(4C-\SC_{\mech_{\mathrm{P}}}(\loc)\bigr)
\geq
\ell C_R+\operatorname{Gain}(R)
+
mC_L+\operatorname{Gain}(L).
\]
\end{lemma}

\begin{proof}
Summing Lemma~\ref{lem:fixed-anchor-bound} over anchors in $R$ and
retaining the routing gain only for outward anchors gives

\[
\sum_{i\in R}F_i
\leq
2\sum_{i\in R}D_i
+3mC_L
-
\sum_{i\in R}\frac{\operatorname{Base}(i)}{S_i}
-
\sum_{\substack{i\in R\\i\text{ outward}}}
\frac{\operatorname{Route}(i)}{S_i}.
\]

By the definition of $\operatorname{Gain}(R)$, the right-hand side
equals
$
4mC_R+3mC_L-\operatorname{Gain}(R).
$

The reflected argument for anchors in $L$ gives
$
\sum_{i\in L}F_i
\leq
4\ell C_L+3\ell C_R-\operatorname{Gain}(L).
$

Adding these two inequalities and using
$n\SC_{\mech_{\mathrm{P}}}(\loc)=\sum_{i\in N}F_i$ proves the
claim.
\end{proof}

It therefore remains to show that the gain associated with each
optimal cluster pays for its truncated spread.

\begin{lemma}[Cluster improvement]
\label{lem:cluster-improvement}
For the right cluster,
$
\ell C_R+\operatorname{Gain}(R)
\geq
n\kappa\frac{T_R}{\delta}.
$
The analogous inequality holds for $L$.
\end{lemma}

\begin{proof}
By scaling, assume $\delta=1$. If $C_R=0$, then $T_R=0$ and the
claim is immediate. We therefore assume $C_R>0$.
Let
$
\bar r=\frac{C_R}{m}
$
be the average distance of an agent in $R$ from $o_R$, and let
$
\beta=\frac{\ell}{m}.
$
For every $i\in R$, define
$
t_i=\frac{r_i}{\bar r}$,
and
$u_i=\frac{D_i}{C_R}.
$

Thus, $t_i$ is the radius of agent $i$ measured in units of the
average radius, while $u_i$ is the total distance from $i$ to its
cluster measured in units of the cluster's optimal cost. We have
$
\sum_{i\in R}t_i=m.
$
Since $o_R$ is a median of $R$, placing a facility at $p_i$ cannot
reduce the one-facility cost of the cluster, and hence $D_i\geq C_R$.
The triangle inequality gives
$
D_i
\leq
\sum_{j\in R}(r_i+r_j)
=
mr_i+C_R.
$
Therefore,
$
1\leq u_i\leq1+t_i.
$
Translate and scale so that $o_L=0$ and $o_R=1$. If $i$ is
outward, then $p_i=1+r_i$; if it is nonoutward, then
$p_i=1-r_i$. Since every agent in $L$ lies to the left of every
agent in $R$,

$$
E_i\leq\ell(1+r_i)+C_L
\qquad
(i\text{ outward}),
$$

and
$$
E_i\leq\ell(1-r_i)+C_L
\qquad
(i\text{ nonoutward}).
$$

For a nonoutward point, the nearest-center assignment also gives
$r_i\leq1/2$.

The following piecewise function records which of the two terms in
$\operatorname{Base}(i)$ is useful after normalization:

$$
b(u)
=
\begin{cases}
u^2, & 1\leq u\leq3,\\
u(2u-3), & u\geq3.
\end{cases}
$$

We claim that an outward anchor satisfies

\[
\frac{\operatorname{Base}(i)+\operatorname{Route}(i)}
{S_iC_R}
\geq
\frac{
b(u_i)+\beta t_i(u_i-1)
}{
u_i+\beta(t_i+1/\bar r)
},
\]

whereas a nonoutward anchor satisfies

$$
\frac{\operatorname{Base}(i)}{S_iC_R}
\geq
\frac{
b(u_i)
}{
u_i+\beta(1/\bar r-t_i)
}.
$$

To verify these inequalities, temporarily write
$\eta=C_L/C_R$. For an outward point, let
$w=\beta(t+1/\bar r)$ and $b_0=\beta t$. For a nonoutward point,
let $w=\beta(1/\bar r-t)$ and $b_0=0$. The bounds on $E_i$ imply
$
\frac{S_i}{C_R}\leq u+w+\eta.
$
Lemma~\ref{lem:squared-routing} contributes
$b_0(u-1)$ in the outward case. If $1\leq u\leq3$, then

\[
\frac{u(u+3\eta)+b_0(u-1)}{u+w+\eta}
-
\frac{u^2+b_0(u-1)}{u+w}
=
\frac{
\eta\bigl(2u^2+3uw-b_0(u-1)\bigr)
}{
(u+w)(u+w+\eta)
}
\geq0.
\]

If $u\geq3$, then

\[
\frac{u(2u-3+3\eta)+b_0(u-1)}{u+w+\eta}
-
\frac{u(2u-3)+b_0(u-1)}{u+w}
=
\frac{
\eta\bigl(u^2+3u+3uw-b_0(u-1)\bigr)
}{
(u+w)(u+w+\eta)
}
\geq0.
\]

In both cases $0\leq b_0\leq w$, so the numerators are
nonnegative. This proves the two normalized anchor bounds.

For fixed $t$, the expression
$
4-2u+\frac{b(u)+b_0(u-1)}{u+w}
$
is nonincreasing in $u$. Indeed, for $1\leq u\leq3$, its
derivative is
$
-
\frac{
u^2+2uw+2w^2-b_0(w+1)
}{
(u+w)^2
}
\leq0,
$
and for $u\geq3$, its derivative is
$
-
\frac{
2w^2+3w-b_0(w+1)
}{
(u+w)^2
}
\leq0.
$

We may therefore use the upper bound $u_i\leq1+t_i$.
Define
$$
A(t)
=
\begin{cases}
(1+t)^2, & 0\leq t\leq2,\\
(1+t)(2t-1), & t\geq2.
\end{cases}
$$

The contribution retained from an outward anchor of normalized radius
$t$ is

$$
\operatorname{Out}(t)
=
\frac{
A(t)+\beta t^2
}{
1+(1+\beta)t+\beta/\bar r
},
$$

whereas the contribution retained from a nonoutward anchor is

$$
\operatorname{In}(t)
=
\frac{
A(t)
}{
1+(1-\beta)t+\beta/\bar r
}.
$$

Using $\sum_{i\in R}t_i=m$, the terms
$2-2t_i$ cancel when summed over the cluster. Hence

\[
\frac{\operatorname{Gain}(R)}{C_R}
\geq
\sum_{\substack{i\in R\\i\text{ outward}}}
\operatorname{Out}(t_i)
+
\sum_{\substack{i\in R\\i\text{ nonoutward}}}
\operatorname{In}(t_i).
\]

It remains to use the median property. Define the common reserve
$
c_0=\frac{1}{1+\beta/\bar r}.
$

We need two scalar inequalities.
First, every nonoutward value $t$ satisfies

\[
\beta t+\operatorname{In}(t)
-(1+\beta)\kappa\bar r t^2
\geq
c_0.
\]

To prove this, put $z=\bar r t$. Since $t$ is nonoutward,
$0\leq z\leq1/2$. Also, $\kappa<1$ and
$A(t)\geq(1+t)^2$. Thus, it is enough to replace $\kappa$ by $1$
and $A(t)$ by $(1+t)^2$. After substituting
$\bar r=z/t$ and clearing positive denominators, the difference
between the resulting left-hand side and $c_0$ is

\[
\frac{t}{
(z+\beta t)\bigl(z+t(z+\beta(1-z))\bigr)
}
\left[
\begin{aligned}
&\beta^3t^2(1-z)^2
+\beta^2tz(1-z)(2-z)\\
&+\beta\bigl(
t^2z(1-z)+tz(2-z)+z^2(2-z)
\bigr)\\
&+z^2(1-z)(t+1)
\end{aligned}
\right].
\]

Every term is nonnegative.

Second, every outward value $t>0$ satisfies
$\beta t+\operatorname{Out}(t)
-(1+\beta)\kappa t\min\{\bar r t,1\}
\geq
-c_0$.
Put $z=\bar r t$ and move $c_0$ to the left. The resulting
expression is

\[
H(z)
=
\beta t
+
\frac{
\bigl(A(t)+\beta t^2\bigr)z
}{
\bigl(1+(1+\beta)t\bigr)z+\beta t
}
+
\frac{z}{z+\beta t}
-
(1+\beta)\kappa t\min\{z,1\}.
\]

On $0<z\leq1$, the two rational terms are concave in $z$, while
the last term is linear. Hence $H$ is concave and its minimum occurs
at an endpoint. Moreover,
$
\lim_{z\downarrow0}H(z)=\beta t>0.
$
For $z\geq1$, the final term is constant and the two rational terms
are nondecreasing. It therefore suffices to prove $H(1)\geq0$.

First use the expression $(1+t)(2t-1)$ for $A(t)$. After
multiplying $H(1)$ by the positive denominator
$
(1+\beta t)\bigl(1+(1+2\beta)t\bigr),
$
the resulting numerator is

\[
\begin{aligned}
&(14-8\sqrt{3})\beta(\beta-1-\sqrt{3})^2t^3\\
&+\bigl(
(21-12\sqrt{3})\beta^2
+(27-16\sqrt{3})\beta
+(8-4\sqrt{3})
\bigr)t^2\\
&+(8-4\sqrt{3})(\beta+1)t.
\end{aligned}
\]

The coefficient of $t^3$ and the coefficient of $t$ are
nonnegative. The quadratic polynomial in $\beta$ multiplying $t^2$
has positive leading coefficient and discriminant
$
249-144\sqrt{3}<0.
$
It is therefore positive for every $\beta\geq0$. When $t\leq2$,
the actual value of $A(t)$ is $(1+t)^2$, which adds the
nonnegative term
$
(2-t)(1+t)(1+\beta t)
$
to the cleared numerator. Thus, $H(1)\geq0$.

Because $o_R$ is a median of $R$, the number of outward agents is at
most the number of nonoutward agents. Pair each outward agent with a
different nonoutward agent. The outward agent can lose at most $c_0$
in the scalar inequality above, while the paired nonoutward agent
contributes at least $c_0$. Any unpaired nonoutward agent contributes
a nonnegative surplus. Therefore,

\[
\beta m
+
\sum_{i\text{ outward}}\operatorname{Out}(t_i)
+
\sum_{i\text{ nonoutward}}\operatorname{In}(t_i)
\geq
(1+\beta)\kappa
\sum_{i\in R}
t_i\min\{\bar r t_i,1\}.
\]

Since $\beta m=\ell$ and
$
T_R
=
\bar r
\sum_{i\in R}
t_i\min\{\bar r t_i,1\},
$
we obtain
$
\frac{\ell C_R+\operatorname{Gain}(R)}{C_R}
\geq
(1+\beta)\kappa\frac{T_R}{\bar r}$.
Finally, $C_R=m\bar r$ and
$m(1+\beta)=m+\ell=n$, so
$
\ell C_R+\operatorname{Gain}(R)
\geq
n\kappa T_R$.
This proves the result under the normalization $\delta=1$. Rescaling
the line gives
$
\ell C_R+\operatorname{Gain}(R)
\geq
n\kappa\frac{T_R}{\delta}.
$
\end{proof}

\begin{proof}[Proof of Lemma~\ref{lem:proportional-dispersion}]
Apply Lemma~\ref{lem:cluster-improvement} to both optimal clusters
and substitute the resulting inequalities into
Lemma~\ref{lem:global-proportional-decomposition}. We obtain
$
n\bigl(4C-\SC_{\mech_{\mathrm{P}}}(\loc)\bigr)
\geq
n\kappa\frac{T_L+T_R}{\delta}.
$
Dividing by $n$ and using the definition of $h(\loc)$ gives
$
\SC_{\mech_{\mathrm{P}}}(\loc)
\leq
\bigl(4-\kappa h(\loc)\bigr)C.
$
\end{proof}

\subsection{The Mixture Upper Bound}
\label{subsec:mixture-upper-bound}

By linearity of expectation and
Lemmas~\ref{lem:global-pair-dispersion} and
\ref{lem:proportional-dispersion},

\[
\begin{aligned}
\SC_{\mech_\lambda}(\loc)
&=
\lambda\SC_{\mech_{\mathrm{P}}}(\loc)
+(1-\lambda)\SC_{\mech_{\mathrm{G}}}(\loc)\\
&\leq
\lambda(4-\kappa h(\loc))C
+(1-\lambda)(3+h(\loc))C\\
&=
\bigl(
3+\lambda+
(1-\lambda-\kappa\lambda)h(\loc)
\bigr)C.
\end{aligned}
\]

Since $0\leq h(\loc)\leq1$, this is at most
$
\max\{3+\lambda,\,4-\kappa\lambda\}C$. At $\lambda=\lambda^\star=1/(1+\kappa)$, the coefficient of
$h(\loc)$ vanishes, and the upper bound becomes 
$
3+\lambda^\star
=
\frac{74+4\sqrt{3}}{23}.
$

\subsection{Matching Lower Bounds}
\label{subsec:mixture-lower-bounds}

We finish by showing that both affine branches of the upper bound are
necessary.

\paragraph{The branch $3+\lambda$.}

For an integer $m\geq2$, let $\eps_m=m^{-2}$ and consider the
profile
$
\loc_m=\langle 0^m,\eps_m,1\rangle,
$
where $0^m$ denotes $m$ agents at $0$. The optimal two-facility cost
is $\eps_m$.
Let $F_0,F_{\eps_m},F_1$ be the expected costs of Proportional
conditioned on choosing an anchor at $0,\eps_m,1$, respectively.
A direct calculation gives
$
F_0
=
\frac{\eps_m(2-\eps_m)}{1+\eps_m},
$,
$
F_{\eps_m}
=
\frac{2m\eps_m(1-\eps_m)}
{m\eps_m+1-\eps_m},
$
and
$
F_1
=
\frac{m\eps_m(2-\eps_m)}
{m+1-\eps_m}.
$
Therefore,

\[
\frac{\SC_{\mech_{\mathrm{P}}}(\loc_m)}
{\mathrm{OPT}_2(\loc_m)}
=
\frac{m}{m+2}
\left(
\frac{2-\eps_m}{1+\eps_m}
+
\frac{2(1-\eps_m)}{m\eps_m+1-\eps_m}
+
\frac{2-\eps_m}{m+1-\eps_m}
\right)
\longrightarrow4.
\]

For Global Pair, the total pair weight is
$m(1+\eps_m)+1-\eps_m$, and the weighted social-cost numerator is
$m\eps_m(3-2\eps_m)$. Hence,

$$
\frac{\SC_{\mech_{\mathrm{G}}}(\loc_m)}
{\mathrm{OPT}_2(\loc_m)}
=
\frac{m(3-2\eps_m)}
{m(1+\eps_m)+1-\eps_m}
\longrightarrow3.
$$

It follows that
$
\lim_{m\to\infty}
\frac{\SC_{\mech_\lambda}(\loc_m)}
{\mathrm{OPT}_2(\loc_m)}
=
3+\lambda.
$

\paragraph{The branch $4-\kappa\lambda$.}

For integers $A,B\geq1$, consider
$
\loc_{A,B}=\langle0^A,1^B,2\rangle.
$
The optimal social cost is $1$. For the Proportional mechanism, the
conditional costs for anchors at $0,1,2$ are
$
F_0=\frac{3B}{B+2}$,
$F_1=\frac{2A}{A+1}$,
and
$F_2=\frac{3AB}{2A+B}$.
Thus,

$
\SC_{\mech_{\mathrm{P}}}(\loc_{A,B})
=
\frac{AF_0+BF_1+F_2}{A+B+1}.
$
If $A/B\to\gamma>0$, then
$
\SC_{\mech_{\mathrm{P}}}(\loc_{A,B})
\longrightarrow
\frac{2(3\gamma^2+5\gamma+1)}
{(\gamma+1)(2\gamma+1)}.
$
For Global Pair,
$
\SC_{\mech_{\mathrm{G}}}(\loc_{A,B})
=
\frac{4AB}{AB+2A+B}
\longrightarrow4.
$
The amount by which the Proportional limit is below $4$ is
$
g(\gamma)
=
\frac{2(\gamma^2+\gamma+1)}
{(\gamma+1)(2\gamma+1)}.
$
The exact factorization
$
g(\gamma)-\kappa
=
\frac{
(7-4\sqrt{3})(\gamma-1-\sqrt{3})^2
}{
(\gamma+1/2)(\gamma+1)
}
\geq0
$
shows that $g$ is minimized at $\gamma=1+\sqrt{3}$, where its value
is $\kappa$. Choosing $A/B\to1+\sqrt{3}$ gives
$
\lim
\frac{\SC_{\mech_\lambda}(\loc_{A,B})}
{\mathrm{OPT}_2(\loc_{A,B})}
=
4-\kappa\lambda.
$
The two profile families prove that the worst-case ratio of
$\mech_\lambda$ is at least
$
\max\{3+\lambda,\,4-\kappa\lambda\}.
$
Together with the upper bound, this proves
Theorem~\ref{thm:two-facility-exact-mixture}.

% \subsection{Lower Bounds}
% \label{sec:lower}
\section{Discussion and open questions.}
Our results leave several natural directions for future work. The most
immediate question is whether, for every fixed $k\geq 4$, there exists
a randomized strategyproof mechanism with an approximation ratio
independent of the number of agents on the line. The Product-Gap
mechanism gives a tight $2k$ approximation for every $k$, but its
strategyproofness fails from $k=4$ onward. It would therefore be
interesting either to modify Product-Gap so as to restore
strategyproofness, or to identify a fundamentally different mechanism
for larger $k$. Another question is whether the factors $6$ 
obtained for $k=3$ is optimal. The improvement for two
facilities suggests that mixtures of mechanisms with complementary
worst-case instances may also yield better guarantees for larger
values of $k$.

More broadly, one may ask what is possible in higher-dimensional
Euclidean spaces. The definition and analysis of Product-Gap rely
strongly on the order structure of the line, and there is no immediate
analogue of consecutive gaps in $\R^d$. A natural direction is to
investigate geometric weights based on distances, simplex volumes, or
other measures of dispersion, and to determine whether they can
simultaneously provide strategyproofness and approximation guarantees.
In particular, it remains open to understand whether, for fixed $k$
and $d$, there are randomized strategyproof mechanisms for
$k$-facility location in $\R^d$ whose approximation ratios depend only
on $k$ and $d$, rather than on the number of agents.

\bibliographystyle{ACM-Reference-Format}
\bibliography{facility}

\appendix

\label{sec:appendix}

\section{Proof of Lemma~\ref{lem:global-pair-dispersion}}
\label{app:global-pair-dispersion}

\begin{proof}
Recall that $o_L<o_R$ is an optimal pair of facilities,
$\delta=o_R-o_L$, and
$
C=\mathrm{OPT}_2(\loc)
$
is the optimal social cost. For every agent $i$, let
$
r_i
=
\min\{\abs{p_i-o_L},\abs{p_i-o_R}\}.
$
Thus, $C=\sum_{i\in N}r_i$. We must prove
$
\SC_{\mech_{\mathrm{G}}}(\loc)
\leq
3C+
\frac{1}{\delta}
\sum_{i\in N}r_i\min\{r_i,\delta\}.
$
If $C=0$, every agent is located at one of the two optimal facility
locations. Every pair with positive Global Pair weight then opens both
occupied locations, and hence the mechanism has social cost zero.
Therefore, suppose that $C>0$.

Translate and scale the profile so that $o_L=0$ and $o_R=1$. We first
prove the result under this normalization, for which $\delta=1$. We
continue to write $p_i$ for the normalized locations and define
$
H(t)
=
3t+t\min\{t,1\}.
$
For agents $i$ and $j$, let $d_{ij}=\abs{p_i-p_j}$, and let
$\SC_{ij}$ be the social cost obtained by opening facilities at
$p_i$ and $p_j$. The normalizing weight of the Global Pair mechanism
is
$
W
=
\sum_{i<j}d_{ij}.
$
Therefore,
$
W\SC_{\mech_{\mathrm{G}}}(\loc)
=
\sum_{i<j}d_{ij}\SC_{ij}.
$
The two selected agents incur zero cost. Consequently, every term on
the right-hand side is determined by a selected pair and one third
agent. For three distinct agents $i,j,k$, define

\[
\begin{aligned}
\Phi_{ijk}
={}&
d_{ij}\min\{d_{ik},d_{jk}\}
+d_{ik}\min\{d_{ij},d_{jk}\}\\
&+
d_{jk}\min\{d_{ij},d_{ik}\}.
\end{aligned}
\]

Rearranging the sum by unordered triples gives the exact identity
$
W\SC_{\mech_{\mathrm{G}}}(\loc)
=
\sum_{i<j<k}\Phi_{ijk}.
$
We prove the following local inequality for every triple:

$$
\Phi_{ijk}
\leq
d_{jk}H(r_i)+d_{ik}H(r_j)+d_{ij}H(r_k).
$$

Order the three reported locations as $x\leq y\leq z$, and write

$
s=y-x$
and 
$u=z-y$.
Their pairwise distances are $s$, $u$, and $s+u$. If the selected pair
is $\{x,y\}$, the remaining point contributes $su$. The pair
$\{y,z\}$ contributes another $su$, while the pair $\{x,z\}$
contributes $(s+u)\min\{s,u\}$. Hence
$
\Phi
=
2su+(s+u)\min\{s,u\}.
$
Assign each of the three points to a nearest optimal facility,
breaking midpoint ties consistently. Since the optimal facilities are
at $0$ and $1$, these assignments are monotone from left to right.
Up to reflection, there are only two cases.

\paragraph{All three points have the same optimal facility.}

Let $a$, $b$, and $c$ be the distances of $x$, $y$, and $z$,
respectively, from their common optimal facility. By the triangle
inequality,

$$
s\leq a+b
\qquad\text{and}\qquad
u\leq b+c.
$$

The part of the desired right-hand side arising from the linear term
$3t$ in $H(t)$ is

\[
3ua+3(s+u)b+3sc
=
3u(a+b)+3s(b+c)\geq
6su.
\]

On the other hand,
$(s+u)\min\{s,u\}\leq2su$,
and therefore $\Phi\leq4su\leq6su$. Thus, the linear part of
$H$ alone proves the local inequality in this case.

\paragraph{The points $x,y$ are assigned to $0$, and $z$ is assigned
to $1$.}

Write

$$
a=r_y,
\qquad
b=r_x,
\qquad
c=r_z,
$$

and let $q(t)=t\min\{t,1\}$, so that $H(t)=3t+q(t)$. The geometry
gives
$
s\leq a+b$
and 
$u\geq1-a-c$.
The first inequality follows by routing the distance between $x$ and
$y$ through $0$. The second follows from
$
1
\leq
\abs{y-0}+\abs{z-y}+\abs{1-z}
=
a+u+c.
$

We distinguish the relative sizes of the two consecutive gaps.

\emph{First suppose that $s\leq u$.} In this case,
$
\Phi=s^2+3su.
$
Let $D$ be the desired right-hand side minus $\Phi$. Expanding gives

\[
D
=
3u(a+b-s)+3s(a+c)-s^2+
uq(b)+(s+u)q(a)+sq(c).
\]

Since $a+b-s\geq0$ and all $q$-terms are nonnegative,
$
D
\geq
uq(b)+3s(a+c)-s^2.
$

If $b\geq1$, then $q(b)=b$. Since $u\geq s$ and
$s\leq a+b$,

\[
\begin{aligned}
D
&\geq
sb+3s(a+c)-s^2\\
&=
s\bigl(b+3a+3c-s\bigr)\\
&\geq
s(2a+3c)
\geq0.
\end{aligned}
\]

Now suppose that $b<1$. If $2a+3c\geq b$, then

\[
\begin{aligned}
D
&\geq
3s(a+c)-s^2\\
&\geq
s\bigl(3a+3c-a-b\bigr)\\
&=
s(2a+3c-b)
\geq0.
\end{aligned}
\]

It remains to consider $b<1$ and $2a+3c<b$. Since
$q(b)=b^2$ and $s\leq a+b$,

$$
D
\geq
ub^2+s(2a+3c-b).
$$

The coefficient $2a+3c-b$ is negative. Using
$u\geq1-a-c$ and $s\leq a+b$, we obtain

\[
\begin{aligned}
D
&\geq
(1-a-c)b^2
+(a+b)(2a+3c-b)\\
&=
a(2a+b-b^2)
+c(3a+3b-b^2).
\end{aligned}
\]

Because $0\leq b<1$,

$$
2a+b-b^2
\geq
b(1-b)
\geq0
$$

and

$$
3a+3b-b^2
\geq
b(3-b)
\geq0.
$$

Thus, $D\geq0$ whenever $s\leq u$.

\emph{Now suppose that $u\leq s$.} In this case,

$$
\Phi=u^2+3su.
$$

The same expansion gives

\[
\begin{aligned}
D
={}&
3u(a+b-s)+3s(a+c)-u^2\\
&+
uq(b)+(s+u)q(a)+sq(c).
\end{aligned}
\]

Using $s\leq a+b$, $s\geq u$, and the nonnegativity of the omitted
terms,

$$
D
\geq
u\bigl(q(b)+3(a+c)-u\bigr).
$$

If $b\geq1$, then $q(b)=b$, and $u\leq s\leq a+b$ implies

$$
q(b)+3(a+c)-u
\geq
b+3a+3c-(a+b)
=
2a+3c
\geq0.
$$

Finally, suppose that $b<1$. From

$$
1-a-c
\leq
u
\leq
s
\leq
a+b
$$

we obtain $1-b\leq2a+c$. Hence,

$$
b-b^2
=
b(1-b)
\leq
1-b
\leq
2a+c
\leq
2a+3c.
$$

Using again $u\leq a+b$,

\[
\begin{aligned}
q(b)+3(a+c)-u
&\geq
b^2+3a+3c-(a+b)\\
&=
2a+3c-(b-b^2)
\geq0.
\end{aligned}
\]

Thus, $D\geq0$ also when $u\leq s$. The remaining split, in which
$x$ is assigned to $0$ and $y,z$ are assigned to $1$, follows by
reflection. This completes the proof of the local triple inequality.

Summing the local inequality over all unordered triples gives

\[
\begin{aligned}
W\SC_{\mech_{\mathrm{G}}}(\loc)
&\leq
\sum_{i\in N}
H(r_i)
\sum_{\substack{\{j,k\}\subseteq N\setminus\{i\}}}
d_{jk}\\
&\leq
W\sum_{i\in N}H(r_i).
\end{aligned}
\]

The second inequality holds because the total distance weight of all
pairs not containing $i$ is at most $W$. Dividing by $W>0$ yields,
under the normalization $\delta=1$,

$$
\SC_{\mech_{\mathrm{G}}}(\loc)
\leq
3C+
\sum_{i\in N}r_i\min\{r_i,1\}.
$$

We now restore the original scale. Apply the normalized inequality to
the locations $(p_i-o_L)/\delta$. Distances, social costs, and optimal
costs are divided by $\delta$, so

$$
\frac{\SC_{\mech_{\mathrm{G}}}(\loc)}{\delta}
\leq
3\frac{C}{\delta}
+
\sum_{i\in N}
\frac{r_i}{\delta}
\min\left\{\frac{r_i}{\delta},1\right\}.
$$

Multiplying by $\delta$ gives

$$
\SC_{\mech_{\mathrm{G}}}(\loc)
\leq
3C+
\frac{1}{\delta}
\sum_{i\in N}r_i\min\{r_i,\delta\}.
$$

Since

$$
T_L+T_R
=
\sum_{i\in N}r_i\min\{r_i,\delta\}
$$

and $h(\loc)=(T_L+T_R)/(\delta C)$, we conclude that

$$
\SC_{\mech_{\mathrm{G}}}(\loc)
\leq
\bigl(3+h(\loc)\bigr)C.
$$
\end{proof}

\section{Complete Proof of Strategyproofness for Three Facilities}
\label{app:k3-proof}

Fix a strategic agent with true location $y$, an alternative report
$r$, and the collection $O$ of reports of the other agents. For a
selected set $S$, write

$$
d(y,S)=\min_{p\in S}\abs{y-p}.
$$

\subsection{Expected-cost potential}

We first handle zero normalizations. If the truthful three-set
normalization is zero, the mechanism opens at every occupied truthful
location, including $y$, and the strategic agent has cost zero. If the
truthful normalization is positive but the normalization after
reporting $r$ is zero, then $O$ occupies exactly two distinct
locations, $r$ is already occupied, and $y$ is a third location. Every
positive-weight truthful triple therefore contains the strategic report
at $y$, so the truthful cost is again zero.

Assume from now on that both normalizations are positive. For
$A\in\binom{O}{2}$ and $a\in\{y,r\}$, define
$W_a(A)=W_3(A\cup\{a\})$. Let
$c_A=d(y,A\cup\{r\})$. For $F\in\binom{O}{3}$, let
$d_F=d(y,F)$. Define

\[
\begin{aligned}
Z(O)
&=\sum_{F\in\binom{O}{3}}W_3(F),
&
B(O)
&=\sum_{F\in\binom{O}{3}}W_3(F)d_F,\\
S_a(O)
&=\sum_{A\in\binom{O}{2}}W_a(A),
&
T_r(O)
&=\sum_{A\in\binom{O}{2}}W_r(A)c_A.
\end{aligned}
\]

Here, $Z(O)$ is the mass of triples containing only ordinary reports,
$B(O)$ is their cost-weighted mass, $S_a(O)$ is the mass of triples
containing the strategic report $a$, and $T_r(O)$ is the corresponding
cost-weighted mass under the deviation.

A truthful triple containing $y$ gives the strategic agent cost zero.
Thus,

$$
C_y(y\mid O)=\frac{B(O)}{Z(O)+S_y(O)}
$$

and

$$
C_y(r\mid O)=\frac{B(O)+T_r(O)}{Z(O)+S_r(O)}.
$$

Cross-multiplication gives

\[
C_y(r\mid O)-C_y(y\mid O)
=
\frac{
Z(O)T_r(O)+S_y(O)T_r(O)
-B(O)\bigl(S_r(O)-S_y(O)\bigr)
}{
\bigl(Z(O)+S_r(O)\bigr)
\bigl(Z(O)+S_y(O)\bigr)
}.
\]

Since $S_y(O)T_r(O)\geq0$, it is enough to prove

$$
\Psi(O)
\eqdef
Z(O)T_r(O)-B(O)\bigl(S_r(O)-S_y(O)\bigr)
\geq0.
$$

\subsection{A local pair inequality}

We first work in generic position: no ordinary report equals $y$, the
closest ordinary report is unique, and all coordinate and radial
comparisons are strict. The general case follows by continuity at the
end.

For $A\in\binom{O}{2}$, let $\delta_A=d(y,A)$.

\begin{lemma}[Local pair inequality]
\label{lem:app-k3-local-pair}
For every pair $A\in\binom{O}{2}$,

$$
W_r(A)-W_y(A)
\leq
W_r(A)\frac{c_A}{\delta_A}.
$$
\end{lemma}

\begin{proof}
Translate so that $y=0$ and reflect if necessary so that $r=h>0$.
If $h\geq\delta_A$, then $c_A=\delta_A$, and the claim follows from
$W_y(A)\geq0$. Suppose that $0<h<\delta_A$. Then $0$ and $h$ lie in
the same component of $\R\setminus A$.

If both points of $A$ lie to the left of $0$, and the nearer point is
$-a$, all gap factors not adjacent to the moving report form a
nonnegative factor $Q$, and $W_t(A)=Q(a+t)$ for $0\leq t\leq h$.
Hence

$$
W_h(A)-W_0(A)
=Qh
\leq
Q(a+h)\frac{h}{a}
=W_h(A)\frac{c_A}{\delta_A}.
$$

If both points lie to the right of $h$, then $W_t(A)$ is nonincreasing
on $[0,h]$, so the left-hand side is nonpositive.

The remaining case is $A=\{-a,b\}$ with $a,b>h$. Then
$W_0(A)=ab$, $W_h(A)=(a+h)(b-h)$, $c_A=h$, and
$\delta_A=\min\{a,b\}$. If $a\leq b$, multiplying the desired
inequality by $a$ leaves the nonnegative difference

$$
h\bigl(a^2+h(b-h)\bigr).
$$

The case $b\leq a$ is symmetric.
\end{proof}

\subsection{Closest-point deletion and reduced atoms}

Let $p\in O$ be the ordinary report closest to $y$, let
$Q=O\setminus\{p\}$, and put $\rho=\abs{p-y}$. For the smaller
profile $Q$, abbreviate $Z=Z(Q)$, $B=B(Q)$,
$S_r=S_r(Q)$, $S_y=S_y(Q)$, and $T=T_r(Q)$. Define

\[
\begin{aligned}
Z_p
&=\sum_{E\in\binom{Q}{2}}W_3(E\cup\{p\}),\\
\Phi_r
&=\sum_{q\in Q}W_3(\{p,q,r\}),
&
\Phi_y
&=\sum_{q\in Q}W_3(\{p,q,y\}),\\
c_p
&=d(y,\{p,r\}).
\end{aligned}
\]

Every ordinary triple containing $p$ is at distance exactly $\rho$
from $y$. Therefore,

\[
\begin{aligned}
Z(O)&=Z+Z_p,
&
B(O)&=B+\rho Z_p,\\
S_r(O)&=S_r+\Phi_r,
&
S_y(O)&=S_y+\Phi_y,\\
T_r(O)&=T+c_p\Phi_r.
\end{aligned}
\]

For $A\in\binom{O}{2}$, define the local surplus

$$
H_A
\eqdef
W_r(A)c_A-\rho\bigl(W_r(A)-W_y(A)\bigr).
$$

Because $p$ is closest to $y$, $\rho\leq\delta_A$. The local pair
inequality implies $H_A\geq0$. Substitution into the definition of
$\Psi$ gives the exact deletion identity

$$
\Psi(O)-\Psi(Q)
=
Z_p\sum_{A\in\binom{O}{2}}H_A
+Zc_p\Phi_r-B(\Phi_r-\Phi_y).
$$

The last two terms need not be nonnegative separately. To expose their
cancellation, group the expansion by the exact set of ordinary witness
labels used by a term. For $V\subseteq Q$ with
$2\leq\abs{V}\leq4$, define

\[
\begin{aligned}
\Theta_V
={}&
\sum_{\substack{
E\in\binom{V}{2},\;
A\in\binom{V\cup\{p\}}{2}\\
E\cup(A\setminus\{p\})=V
}}
W_3(E\cup\{p\})H_A\\
&+
\sum_{\substack{
F\in\binom{V}{3},\;q\in V\\
F\cup\{q\}=V
}}
W_3(F)
\Bigl[
 c_pW_3(\{p,q,r\})
-d_F\bigl(W_3(\{p,q,r\})-W_3(\{p,q,y\})\bigr)
\Bigr].
\end{aligned}
\]

The union conditions make $V$ the exact witness support of the
corresponding product term. Expanding the deletion identity and grouping
by this support gives

$$
\Psi(O)-\Psi(Q)
=
\sum_{\substack{V\subseteq Q\\2\leq\abs{V}\leq4}}\Theta_V.
$$

It remains to prove $\Theta_V\geq0$ for every reduced atom.

\subsection{Same-side interior deviations}

Translate $y$ to $0$ and reflect so that $p=\rho>0$. Every witness in
a generic reduced atom lies outside $[-\rho,\rho]$. Suppose first that
$0\leq r\leq\rho$, and write $r=t$. The first sum in $\Theta_V$ is
nonnegative because $H_A\geq0$.

For a witness $q=\rho+x$ with $x>0$,

$$
W_3(\{p,q,t\})-W_3(\{p,q,0\})=-tx\leq0.
$$

For a witness $q=-\rho-x$ with $x>0$,

$$
W_3(\{p,q,t\})-W_3(\{p,q,0\})=-t(x+t)\leq0.
$$

Since $c_p=t\geq0$, every bracket in the second sum defining
$\Theta_V$ is nonnegative. Hence $\Theta_V\geq0$ for
$0\leq r\leq\rho$.

\subsection{Opposite-side interior deviations}
\label{app:k3-interior}

Assume now that $-\rho\leq r\leq0$. The endpoints follow from the
same-side chamber and continuity, so suppose

$$
r=-h<0<p=\rho,
\qquad
\rho=h+\eta,
\qquad
h,\eta>0.
$$

Every witness lies outside $[-\rho,\rho]$.

\subsubsection{Local surplus and two witnesses}

Since $r=-h$ is closer to $0$ than every report in a pair $A$,
$c_A=h$. Hence

$$
H_A
=
\rho W_3(A\cup\{0\})
-(\rho-h)W_3(A\cup\{-h\}).
$$

There are three side patterns. If $A=\{u,v\}$ with
$\rho<u<v$, then

$$
H_A=h(u-\rho+h)(v-u)\geq0.
$$

If $A=\{-v,-u\}$ with $\rho<u<v$, then

$$
H_A=h(\rho+u-h)(v-u)\geq0.
$$

Finally, if $A=\{-u,v\}$ with $u,v>\rho$, write
$u=\rho+x$ and $v=\rho+z$. Then

\[
H_A
=
h\bigl(
 h^2+3h\eta+hx+hz+\eta^2+2\eta z+xz
\bigr)
\geq0.
\]

Thus every two-witness atom is nonnegative.

\subsubsection{Three witnesses}

Let $V=\{v_1,v_2,v_3\}$, ordered by increasing distance from $0$, and
write

$$
\abs{v_1}=\rho+a,
\qquad
\abs{v_2}=\rho+a+b,
\qquad
\abs{v_3}=\rho+a+b+c,
$$

where $a,b,c>0$. A sign word $\sigma\in\{-,+\}^3$ records the side
of each witness in this radial order. Define the common core

\[
\begin{aligned}
\mathcal C_3
=
h\Bigl(
&2h\bigl[
 a^2(b+c)+a(2b^2+6bc+c^2)+bc(2b+c)
\bigr]\\
&+ab(2ab+3ac+2bc+c^2)
\Bigr).
\end{aligned}
\]

For each side pattern, direct substitution into the reduced atom gives
an exact identity

$$
\Theta_\sigma
=
\mathcal C_3+h\mathcal R^{(3)}_\sigma,
$$

where $\mathcal R^{(3)}_\sigma$ is a polynomial in
$h,\eta,a,b,c$ with nonnegative coefficients. The eight exact Horner
forms are listed in the table file included below. Hence every
three-witness atom is nonnegative.

\subsubsection{Four witnesses}

Let $V=\{v_1,v_2,v_3,v_4\}$, ordered radially, and write

\[
\begin{aligned}
\abs{v_1}&=\rho+a,\\
\abs{v_2}&=\rho+a+b,\\
\abs{v_3}&=\rho+a+b+c,\\
\abs{v_4}&=\rho+a+b+c+d,
\end{aligned}
\qquad
a,b,c,d>0.
\]

For a sign word $\sigma\in\{-,+\}^4$, let $\Theta_\sigma(d)$ be the
corresponding reduced atom. Within a fixed side pattern, this atom is
affine in the outer radial gap $d$. Define

$$
\mathcal D_4
=
h\bigl(
6hab+6hac+4hb^2+6hbc+ab^2+2abc
\bigr)
$$

and

$$
\mathcal B_4
=
2hc\bigl(
3hab+2hac+2hb^2+2hbc+a^2b+ab^2+abc
\bigr).
$$

For every side pattern, the exact identities in the table file have
the form

$$
\frac{\partial\Theta_\sigma}{\partial d}
=
\mathcal D_4+\mathcal R^{(4,\partial)}_\sigma
$$

and

$$
\Theta_\sigma(0)
=
\mathcal B_4+\mathcal R^{(4,0)}_\sigma,
$$

where both remainders have nonnegative coefficients. Therefore,

$$
\frac{\partial\Theta_\sigma}{\partial d}\geq0
\qquad\text{and}\qquad
\Theta_\sigma(0)\geq0.
$$

It follows that

$$
\Theta_\sigma(d)
=
\Theta_\sigma(0)
+
\int_0^d
\frac{\partial\Theta_\sigma(s)}{\partial s}\,\dd s
\geq0.
$$

Thus every reduced atom is nonnegative for $-\rho\leq r\leq0$.

\begingroup
\small
\setlength{\jot}{3pt}
\allowdisplaybreaks
\subsubsection{Exact Horner-Form Coefficient Identities}
\label{app:k3-radial-tables}

All variables in this subsection are nonnegative.  The displayed
identities are exact; their nested form makes coefficientwise
nonnegativity visible without printing an unstructured monomial list.

\paragraph{Three-witness residuals.}

\[
\begin{aligned}
R^{(3)}_{---}={}&\left(\begin{aligned}&2 a^{2} c^{2}\\ &+h \left(\begin{aligned}&a \left(b \left(4 b + 15 c\right) + 4 c^{2}\right) + b \left(4 b c + 2 c^{2}\right)\\ &+h \left(a \left(6 b + 6 c\right) + b \left(8 b + 20 c\right) + 4 c^{2} + h \left(4 b + 4 c\right)\right)\end{aligned}\right)\end{aligned}\right)\\ &+\eta\left(\begin{aligned}&a \left(b \left(8 b + 24 c\right) + 8 c^{2}\right) + b \left(8 b c + 4 c^{2}\right)\\ &+h \left(a \left(10 b + 10 c\right) + b \left(16 b + 52 c\right) + 12 c^{2} + h \left(16 b + 16 c\right)\right)\end{aligned}\right)\\ &+\eta^{2}\left(\begin{aligned}&a \left(2 b + 2 c\right) + b \left(8 b + 28 c\right) + 8 c^{2} + h \left(16 b + 16 c\right)\end{aligned}\right)\\ &+\eta^{3}\left(\begin{aligned}&4 b + 4 c\end{aligned}\right).
\end{aligned}
\]

\[
\begin{aligned}
R^{(3)}_{--+}={}&\left(\begin{aligned}&a \left(\begin{aligned}&a \left(a \left(2 a + 4 b + 4 c\right) + b \left(3 b + 4 c\right) + 2 c^{2}\right)\\ &+b^{2} \left(2 b + c\right)\end{aligned}\right)\\ &+h \left(\begin{aligned}&a \left(a \left(10 a + 40 b + 14 c\right) + b \left(34 b + 28 c\right) + 4 c^{2}\right)\\ &+b \left(b \left(8 b + 8 c\right) + 2 c^{2}\right)\\ &+h \left(\begin{aligned}&a \left(16 a + 56 b + 20 c\right) + b \left(28 b + 28 c\right) + 4 c^{2}\\ &+h \left(8 a + 16 b + 8 c\right)\end{aligned}\right)\end{aligned}\right)\end{aligned}\right)\\ &+\eta\left(\begin{aligned}&a \left(a \left(8 a + 30 b + 16 c\right) + b \left(32 b + 38 c\right) + 8 c^{2}\right)\\ &+b \left(b \left(8 b + 12 c\right) + 4 c^{2}\right)\\ &+h \left(\begin{aligned}&a \left(32 a + 112 b + 44 c\right) + b \left(60 b + 68 c\right) + 12 c^{2}\\ &+h \left(32 a + 64 b + 32 c\right)\end{aligned}\right)\end{aligned}\right)\\ &+\eta^{2}\left(\begin{aligned}&a \left(12 a + 44 b + 20 c\right) + b \left(28 b + 36 c\right) + 8 c^{2}\\ &+h \left(32 a + 64 b + 32 c\right)\end{aligned}\right)\\ &+\eta^{3}\left(\begin{aligned}&8 a + 16 b + 8 c\end{aligned}\right).
\end{aligned}
\]

\[
\begin{aligned}
R^{(3)}_{-+-}={}&\left(\begin{aligned}&a \left(a \left(a \left(2 a + 4 b\right) + 3 b^{2}\right) + b^{2} \left(2 b + c\right)\right)\\ &+h \left(\begin{aligned}&a \left(a \left(10 a + 40 b + 24 c\right) + b \left(34 b + 24 c\right) + 2 c^{2}\right)\\ &+b \left(b \left(8 b + 8 c\right) + 2 c^{2}\right)\\ &+h \left(\begin{aligned}&a \left(16 a + 56 b + 36 c\right) + b \left(28 b + 28 c\right) + 4 c^{2}\\ &+h \left(8 a + 16 b + 8 c\right)\end{aligned}\right)\end{aligned}\right)\end{aligned}\right)\\ &+\eta\left(\begin{aligned}&a \left(a \left(8 a + 30 b + 14 c\right) + b \left(32 b + 26 c\right) + 2 c^{2}\right)\\ &+b \left(b \left(8 b + 12 c\right) + 4 c^{2}\right)\\ &+h \left(\begin{aligned}&a \left(32 a + 112 b + 68 c\right) + b \left(60 b + 52 c\right) + 4 c^{2}\\ &+h \left(32 a + 64 b + 32 c\right)\end{aligned}\right)\end{aligned}\right)\\ &+\eta^{2}\left(\begin{aligned}&a \left(12 a + 44 b + 24 c\right) + b \left(28 b + 20 c\right)\\ &+h \left(32 a + 64 b + 32 c\right)\end{aligned}\right)\\ &+\eta^{3}\left(\begin{aligned}&8 a + 16 b + 8 c\end{aligned}\right).
\end{aligned}
\]

\[
\begin{aligned}
R^{(3)}_{-++}={}&\left(\begin{aligned}&a^{3} \left(2 a + 4 b\right)\\ &+h \left(\begin{aligned}&a \left(a \left(10 a + 16 b + 20 c\right) + b \left(4 b + 17 c\right) + 2 c^{2}\right)\\ &+b \left(4 b c + 2 c^{2}\right)\\ &+h \left(\begin{aligned}&a \left(14 a + 22 b + 28 c\right) + b \left(8 b + 20 c\right) + 4 c^{2}\\ &+h \left(4 a + 4 b + 4 c\right)\end{aligned}\right)\end{aligned}\right)\end{aligned}\right)\\ &+\eta\left(\begin{aligned}&a \left(a \left(8 a + 16 b + 10 c\right) + b \left(8 b + 20 c\right) + 2 c^{2}\right)\\ &+b \left(8 b c + 4 c^{2}\right)\\ &+h \left(\begin{aligned}&a \left(26 a + 42 b + 48 c\right) + b \left(16 b + 36 c\right) + 4 c^{2}\\ &+h \left(16 a + 16 b + 16 c\right)\end{aligned}\right)\end{aligned}\right)\\ &+\eta^{2}\left(\begin{aligned}&a \left(10 a + 18 b + 14 c\right) + b \left(8 b + 12 c\right)\\ &+h \left(16 a + 16 b + 16 c\right)\end{aligned}\right)\\ &+\eta^{3}\left(\begin{aligned}&4 a + 4 b + 4 c\end{aligned}\right).
\end{aligned}
\]

\[
\begin{aligned}
R^{(3)}_{+--}={}&\left(\begin{aligned}&a^{3} \left(2 a + 4 b\right)\\ &+h \left(\begin{aligned}&a \left(a \left(10 a + 12 b + 24 c\right) + 14 b c + 2 c^{2}\right)\\ &+h \left(a \left(16 a + 12 b + 36 c\right) + 12 b c + 4 c^{2} + h \left(8 a + 8 c\right)\right)\end{aligned}\right)\end{aligned}\right)\\ &+\eta\left(\begin{aligned}&a \left(a \left(8 a + 8 b + 14 c\right) + 12 b c + 2 c^{2}\right)\\ &+h \left(a \left(32 a + 20 b + 68 c\right) + 20 b c + 4 c^{2} + h \left(32 a + 32 c\right)\right)\end{aligned}\right)\\ &+\eta^{2}\left(\begin{aligned}&a \left(12 a + 4 b + 24 c\right) + 4 b c + h \left(32 a + 32 c\right)\end{aligned}\right)\\ &+\eta^{3}\left(\begin{aligned}&8 a + 8 c\end{aligned}\right).
\end{aligned}
\]

\[
\begin{aligned}
R^{(3)}_{+-+}={}&\left(\begin{aligned}&a \left(a \left(a \left(2 a + 4 b\right) + 3 b^{2}\right) + b^{2} \left(2 b + c\right)\right)\\ &+h \left(\begin{aligned}&a \left(a \left(10 a + 32 b + 20 c\right) + b \left(21 b + 15 c\right) + 2 c^{2}\right)\\ &+b^{2} \left(4 b + 2 c\right)\\ &+h \left(\begin{aligned}&a \left(14 a + 34 b + 28 c\right) + b \left(10 b + 14 c\right) + 4 c^{2}\\ &+h \left(4 a + 4 b + 4 c\right)\end{aligned}\right)\end{aligned}\right)\end{aligned}\right)\\ &+\eta\left(\begin{aligned}&a \left(a \left(8 a + 18 b + 10 c\right) + b \left(8 b + 10 c\right) + 2 c^{2}\right)\\ &+h \left(\begin{aligned}&a \left(26 a + 58 b + 48 c\right) + b \left(14 b + 18 c\right) + 4 c^{2}\\ &+h \left(16 a + 16 b + 16 c\right)\end{aligned}\right)\end{aligned}\right)\\ &+\eta^{2}\left(\begin{aligned}&a \left(10 a + 16 b + 14 c\right) + b \left(2 b + 2 c\right)\\ &+h \left(16 a + 16 b + 16 c\right)\end{aligned}\right)\\ &+\eta^{3}\left(\begin{aligned}&4 a + 4 b + 4 c\end{aligned}\right).
\end{aligned}
\]

\[
\begin{aligned}
R^{(3)}_{++-}={}&\left(\begin{aligned}&a \left(\begin{aligned}&a \left(a \left(2 a + 4 b + 4 c\right) + b \left(3 b + 4 c\right) + 2 c^{2}\right)\\ &+b^{2} \left(2 b + c\right)\end{aligned}\right)\\ &+h \left(\begin{aligned}&a \left(a \left(10 a + 32 b + 10 c\right) + b \left(21 b + 11 c\right)\right)\\ &+b^{2} \left(4 b + 2 c\right)\\ &+h \left(a \left(14 a + 34 b + 6 c\right) + b \left(10 b + 6 c\right) + h \left(4 a + 4 b\right)\right)\end{aligned}\right)\end{aligned}\right)\\ &+\eta\left(\begin{aligned}&a \left(a \left(8 a + 18 b + 8 c\right) + b \left(8 b + 6 c\right)\right)\\ &+h \left(a \left(26 a + 58 b + 10 c\right) + b \left(14 b + 10 c\right) + h \left(16 a + 16 b\right)\right)\end{aligned}\right)\\ &+\eta^{2}\left(\begin{aligned}&a \left(10 a + 16 b + 2 c\right) + b \left(2 b + 2 c\right) + h \left(16 a + 16 b\right)\end{aligned}\right)\\ &+\eta^{3}\left(\begin{aligned}&4 a + 4 b\end{aligned}\right).
\end{aligned}
\]

\[
\begin{aligned}
R^{(3)}_{+++}={}&\left(\begin{aligned}&2 a^{2} c^{2}\end{aligned}\right).
\end{aligned}
\]

\paragraph{Four-witness slopes.}

For $\varepsilon\in\{-,+\}$, the patterns $---\varepsilon$ satisfy

\[
\begin{aligned}
\frac{1}{h}\frac{\partial\Theta_{---\varepsilon}}{\partial d}={}&\left(\begin{aligned}&a \left(a \left(2 b + 2 c\right) + b \left(2 b + 2 c\right)\right)\\ &+h \left(a \left(16 b + 16 c\right) + b \left(8 b + 13 c\right) + h \left(10 b + 10 c\right)\right)\end{aligned}\right)\\ &+\eta\left(\begin{aligned}&a \left(16 b + 16 c\right) + b \left(8 b + 12 c\right) + h \left(30 b + 30 c\right)\end{aligned}\right)\\ &+\eta^{2}\left(\begin{aligned}&18 b + 18 c\end{aligned}\right).
\end{aligned}
\]

For $\varepsilon\in\{-,+\}$, the patterns $--+\varepsilon$ satisfy

\[
\begin{aligned}
\frac{1}{h}\frac{\partial\Theta_{--+\varepsilon}}{\partial d}={}&\left(\begin{aligned}&a \left(2 a c + b \left(b + 2 c\right)\right)\\ &+h \left(\begin{aligned}&a \left(18 a + 31 b + 16 c\right) + b \left(14 b + 13 c\right)\\ &+h \left(26 a + 22 b + 10 c + 8 h\right)\end{aligned}\right)\end{aligned}\right)\\ &+\eta\left(\begin{aligned}&a \left(12 a + 22 b + 16 c\right) + b \left(12 b + 12 c\right)\\ &+h \left(54 a + 46 b + 30 c + 32 h\right)\end{aligned}\right)\\ &+\eta^{2}\left(\begin{aligned}&22 a + 20 b + 18 c + 32 h\end{aligned}\right)\\ &+\eta^{3}\left(\begin{aligned}&8\end{aligned}\right).
\end{aligned}
\]

For $\varepsilon\in\{-,+\}$, the patterns $-+-\varepsilon$ satisfy

\[
\begin{aligned}
\frac{1}{h}\frac{\partial\Theta_{-+-\varepsilon}}{\partial d}={}&\left(\begin{aligned}&a b \left(b + 2 c\right)\\ &+h \left(\begin{aligned}&a \left(18 a + 31 b + 15 c\right) + b \left(14 b + 13 c\right)\\ &+h \left(26 a + 22 b + 10 c + 8 h\right)\end{aligned}\right)\end{aligned}\right)\\ &+\eta\left(\begin{aligned}&a \left(12 a + 22 b + 8 c\right) + b \left(12 b + 12 c\right)\\ &+h \left(54 a + 46 b + 14 c + 32 h\right)\end{aligned}\right)\\ &+\eta^{2}\left(\begin{aligned}&22 a + 20 b + 2 c + 32 h\end{aligned}\right)\\ &+\eta^{3}\left(\begin{aligned}&8\end{aligned}\right).
\end{aligned}
\]

For $\varepsilon\in\{-,+\}$, the patterns $-++\varepsilon$ satisfy

\[
\begin{aligned}
\frac{1}{h}\frac{\partial\Theta_{-++\varepsilon}}{\partial d}={}&\left(\begin{aligned}&a \left(2 a b + b \left(2 b + 2 c\right)\right)\\ &+h \left(\begin{aligned}&a \left(10 a + 18 b + 15 c\right) + b \left(8 b + 13 c\right)\\ &+h \left(10 a + 10 b + 10 c\right)\end{aligned}\right)\end{aligned}\right)\\ &+\eta\left(\begin{aligned}&a \left(4 a + 12 b + 8 c\right) + b \left(8 b + 12 c\right) + h \left(14 a + 14 b + 14 c\right)\end{aligned}\right)\\ &+\eta^{2}\left(\begin{aligned}&2 a + 2 b + 2 c\end{aligned}\right).
\end{aligned}
\]

For $\varepsilon\in\{-,+\}$, the patterns $+--\varepsilon$ satisfy

\[
\begin{aligned}
\frac{1}{h}\frac{\partial\Theta_{+--\varepsilon}}{\partial d}={}&\left(\begin{aligned}&a \left(2 a b + b \left(2 b + 2 c\right)\right)\\ &+h \left(\begin{aligned}&a \left(18 a + 20 b + 15 c\right) + b \left(4 b + 6 c\right)\\ &+h \left(26 a + 12 b + 10 c + 8 h\right)\end{aligned}\right)\end{aligned}\right)\\ &+\eta\left(\begin{aligned}&a \left(12 a + 12 b + 8 c\right) + h \left(54 a + 20 b + 14 c + 32 h\right)\end{aligned}\right)\\ &+\eta^{2}\left(\begin{aligned}&22 a + 4 b + 2 c + 32 h\end{aligned}\right)\\ &+\eta^{3}\left(\begin{aligned}&8\end{aligned}\right).
\end{aligned}
\]

For $\varepsilon\in\{-,+\}$, the patterns $+-+\varepsilon$ satisfy

\[
\begin{aligned}
\frac{1}{h}\frac{\partial\Theta_{+-+\varepsilon}}{\partial d}={}&\left(\begin{aligned}&a b \left(b + 2 c\right)\\ &+h \left(\begin{aligned}&a \left(10 a + 19 b + 15 c\right) + b \left(6 b + 6 c\right)\\ &+h \left(10 a + 10 b + 10 c\right)\end{aligned}\right)\end{aligned}\right)\\ &+\eta\left(\begin{aligned}&a \left(4 a + 6 b + 8 c\right) + h \left(14 a + 14 b + 14 c\right)\end{aligned}\right)\\ &+\eta^{2}\left(\begin{aligned}&2 a + 2 b + 2 c\end{aligned}\right).
\end{aligned}
\]

For $\varepsilon\in\{-,+\}$, the patterns $++-\varepsilon$ satisfy

\[
\begin{aligned}
\frac{1}{h}\frac{\partial\Theta_{++-\varepsilon}}{\partial d}={}&\left(\begin{aligned}&a \left(2 a c + b \left(b + 2 c\right)\right)\\ &+h \left(a \left(10 a + 19 b + 6 c\right) + b \left(6 b + 6 c\right) + h \left(10 a + 10 b\right)\right)\end{aligned}\right)\\ &+\eta\left(\begin{aligned}&a \left(4 a + 6 b\right) + h \left(14 a + 14 b\right)\end{aligned}\right)\\ &+\eta^{2}\left(\begin{aligned}&2 a + 2 b\end{aligned}\right).
\end{aligned}
\]

For $\varepsilon\in\{-,+\}$, the patterns $+++\varepsilon$ satisfy

\[
\begin{aligned}
\frac{1}{h}\frac{\partial\Theta_{+++\varepsilon}}{\partial d}={}&\left(\begin{aligned}&a \left(a \left(2 b + 2 c\right) + b \left(2 b + 2 c\right)\right)\\ &+h \left(a \left(6 b + 6 c\right) + b \left(4 b + 6 c\right)\right)\end{aligned}\right).
\end{aligned}
\]

\paragraph{Four-witness boundary values.}

For $\sigma=----$,

\[
\begin{aligned}
\frac{1}{h}\Theta_{\sigma}(0)={}&\left(\begin{aligned}&a \left(a \left(2 b c + 4 c^{2}\right) + b \left(2 b c + 2 c^{2}\right)\right)\\ &+h \left(\begin{aligned}&a \left(16 b c + 12 c^{2}\right) + b \left(8 b c + 8 c^{2}\right)\\ &+h \left(10 b c + 8 c^{2}\right)\end{aligned}\right)\end{aligned}\right)\\ &+\eta\left(\begin{aligned}&a \left(16 b c + 16 c^{2}\right) + b \left(8 b c + 8 c^{2}\right)\\ &+h \left(30 b c + 24 c^{2}\right)\end{aligned}\right)\\ &+\eta^{2}\left(\begin{aligned}&18 b c + 16 c^{2}\end{aligned}\right).
\end{aligned}
\]

For $\sigma\in\{---+, --+-\}$,

\[
\begin{aligned}
\frac{1}{h}\Theta_{\sigma}(0)={}&\left(\begin{aligned}&a \left(\begin{aligned}&a \left(b \left(2 b + 2 c\right) + 2 c^{2}\right)\\ &+b \left(b \left(2 b + 3 c\right) + 2 c^{2}\right)\end{aligned}\right)\\ &+h \left(\begin{aligned}&a \left(a \left(18 b + 18 c\right) + b \left(24 b + 47 c\right) + 16 c^{2}\right)\\ &+b \left(b \left(8 b + 22 c\right) + 13 c^{2}\right)\\ &+h \left(a \left(26 b + 26 c\right) + b \left(18 b + 32 c\right) + 10 c^{2} + h \left(8 b + 8 c\right)\right)\end{aligned}\right)\end{aligned}\right)\\ &+\eta\left(\begin{aligned}&a \left(a \left(12 b + 12 c\right) + b \left(20 b + 38 c\right) + 16 c^{2}\right)\\ &+b \left(b \left(8 b + 20 c\right) + 12 c^{2}\right)\\ &+h \left(a \left(54 b + 54 c\right) + b \left(38 b + 76 c\right) + 30 c^{2} + h \left(32 b + 32 c\right)\right)\end{aligned}\right)\\ &+\eta^{2}\left(\begin{aligned}&a \left(22 b + 22 c\right) + b \left(18 b + 38 c\right) + 18 c^{2} + h \left(32 b + 32 c\right)\end{aligned}\right)\\ &+\eta^{3}\left(\begin{aligned}&8 b + 8 c\end{aligned}\right).
\end{aligned}
\]

For $\sigma=--++$,

\[
\begin{aligned}
\frac{1}{h}\Theta_{\sigma}(0)={}&\left(\begin{aligned}&a \left(\begin{aligned}&a \left(a \left(4 a + 8 b + 8 c\right) + b \left(6 b + 10 c\right) + 4 c^{2}\right)\\ &+b \left(b \left(2 b + 4 c\right) + 2 c^{2}\right)\end{aligned}\right)\\ &+h \left(\begin{aligned}&a \left(a \left(16 a + 38 b + 28 c\right) + b \left(30 b + 42 c\right) + 12 c^{2}\right)\\ &+b \left(b \left(8 b + 16 c\right) + 8 c^{2}\right)\\ &+h \left(\begin{aligned}&a \left(20 a + 38 b + 28 c\right) + b \left(18 b + 26 c\right) + 8 c^{2}\\ &+h \left(8 a + 8 b + 8 c\right)\end{aligned}\right)\end{aligned}\right)\end{aligned}\right)\\ &+\eta\left(\begin{aligned}&a \left(a \left(16 a + 36 b + 32 c\right) + b \left(28 b + 44 c\right) + 16 c^{2}\right)\\ &+b \left(b \left(8 b + 16 c\right) + 8 c^{2}\right)\\ &+h \left(\begin{aligned}&a \left(44 a + 82 b + 68 c\right) + b \left(38 b + 62 c\right) + 24 c^{2}\\ &+h \left(32 a + 32 b + 32 c\right)\end{aligned}\right)\end{aligned}\right)\\ &+\eta^{2}\left(\begin{aligned}&a \left(20 a + 38 b + 36 c\right) + b \left(18 b + 34 c\right) + 16 c^{2}\\ &+h \left(32 a + 32 b + 32 c\right)\end{aligned}\right)\\ &+\eta^{3}\left(\begin{aligned}&8 a + 8 b + 8 c\end{aligned}\right).
\end{aligned}
\]

For $\sigma=-+--$,

\[
\begin{aligned}
\frac{1}{h}\Theta_{\sigma}(0)={}&\left(\begin{aligned}&a \left(a b \left(2 b + 2 c\right) + b \left(b \left(2 b + 4 c\right) + 2 c^{2}\right)\right)\\ &+h \left(\begin{aligned}&a \left(a \left(18 b + 18 c\right) + b \left(24 b + 32 c\right) + 8 c^{2}\right)\\ &+b \left(b \left(8 b + 16 c\right) + 8 c^{2}\right)\\ &+h \left(a \left(26 b + 26 c\right) + b \left(18 b + 26 c\right) + 8 c^{2} + h \left(8 b + 8 c\right)\right)\end{aligned}\right)\end{aligned}\right)\\ &+\eta\left(\begin{aligned}&a \left(a \left(12 b + 12 c\right) + b \left(20 b + 24 c\right) + 4 c^{2}\right)\\ &+b \left(b \left(8 b + 16 c\right) + 8 c^{2}\right)\\ &+h \left(a \left(54 b + 54 c\right) + b \left(38 b + 46 c\right) + 8 c^{2} + h \left(32 b + 32 c\right)\right)\end{aligned}\right)\\ &+\eta^{2}\left(\begin{aligned}&a \left(22 b + 22 c\right) + b \left(18 b + 18 c\right) + h \left(32 b + 32 c\right)\end{aligned}\right)\\ &+\eta^{3}\left(\begin{aligned}&8 b + 8 c\end{aligned}\right).
\end{aligned}
\]

For $\sigma\in\{-+-+, -++-\}$,

\[
\begin{aligned}
\frac{1}{h}\Theta_{\sigma}(0)={}&\left(\begin{aligned}&a \left(\begin{aligned}&a \left(a \left(4 a + 8 b\right) + b \left(6 b + 2 c\right)\right)\\ &+b \left(b \left(2 b + 3 c\right) + 2 c^{2}\right)\end{aligned}\right)\\ &+h \left(\begin{aligned}&a \left(a \left(16 a + 38 b + 28 c\right) + b \left(30 b + 49 c\right) + 15 c^{2}\right)\\ &+b \left(b \left(8 b + 22 c\right) + 13 c^{2}\right)\\ &+h \left(\begin{aligned}&a \left(20 a + 38 b + 36 c\right) + b \left(18 b + 32 c\right) + 10 c^{2}\\ &+h \left(8 a + 8 b + 8 c\right)\end{aligned}\right)\end{aligned}\right)\end{aligned}\right)\\ &+\eta\left(\begin{aligned}&a \left(a \left(16 a + 36 b + 16 c\right) + b \left(28 b + 34 c\right) + 8 c^{2}\right)\\ &+b \left(b \left(8 b + 20 c\right) + 12 c^{2}\right)\\ &+h \left(\begin{aligned}&a \left(44 a + 82 b + 68 c\right) + b \left(38 b + 60 c\right) + 14 c^{2}\\ &+h \left(32 a + 32 b + 32 c\right)\end{aligned}\right)\end{aligned}\right)\\ &+\eta^{2}\left(\begin{aligned}&a \left(20 a + 38 b + 24 c\right) + b \left(18 b + 22 c\right) + 2 c^{2}\\ &+h \left(32 a + 32 b + 32 c\right)\end{aligned}\right)\\ &+\eta^{3}\left(\begin{aligned}&8 a + 8 b + 8 c\end{aligned}\right).
\end{aligned}
\]

For $\sigma=-+++$,

\[
\begin{aligned}
\frac{1}{h}\Theta_{\sigma}(0)={}&\left(\begin{aligned}&a \left(2 a b c + b \left(2 b c + 2 c^{2}\right)\right)\\ &+h \left(\begin{aligned}&a \left(10 a c + 18 b c + 8 c^{2}\right) + b \left(8 b c + 8 c^{2}\right)\\ &+h \left(10 a c + 10 b c + 8 c^{2}\right)\end{aligned}\right)\end{aligned}\right)\\ &+\eta\left(\begin{aligned}&a \left(4 a c + 12 b c + 4 c^{2}\right) + b \left(8 b c + 8 c^{2}\right)\\ &+h \left(14 a c + 14 b c + 8 c^{2}\right)\end{aligned}\right)\\ &+\eta^{2}\left(\begin{aligned}&2 a c + 2 b c\end{aligned}\right).
\end{aligned}
\]

For $\sigma=+---$,

\[
\begin{aligned}
\frac{1}{h}\Theta_{\sigma}(0)={}&\left(\begin{aligned}&a \left(2 a b c + b \left(2 b c + 2 c^{2}\right)\right)\\ &+h \left(\begin{aligned}&a \left(18 a c + 20 b c + 8 c^{2}\right) + b \left(4 b c + 4 c^{2}\right)\\ &+h \left(26 a c + 12 b c + 8 c^{2} + 8 c h\right)\end{aligned}\right)\end{aligned}\right)\\ &+\eta\left(\begin{aligned}&a \left(12 a c + 12 b c + 4 c^{2}\right) + h \left(54 a c + 20 b c + 8 c^{2} + 32 c h\right)\end{aligned}\right)\\ &+\eta^{2}\left(\begin{aligned}&22 a c + 4 b c + 32 c h\end{aligned}\right)\\ &+\eta^{3}\left(\begin{aligned}&8 c\end{aligned}\right).
\end{aligned}
\]

For $\sigma\in\{+--+, +-+-\}$,

\[
\begin{aligned}
\frac{1}{h}\Theta_{\sigma}(0)={}&\left(\begin{aligned}&a \left(\begin{aligned}&a \left(a \left(4 a + 8 b\right) + b \left(6 b + 2 c\right)\right)\\ &+b \left(b \left(2 b + 3 c\right) + 2 c^{2}\right)\end{aligned}\right)\\ &+h \left(\begin{aligned}&a \left(a \left(16 a + 38 b + 28 c\right) + b \left(24 b + 39 c\right) + 15 c^{2}\right)\\ &+b \left(b \left(4 b + 10 c\right) + 6 c^{2}\right)\\ &+h \left(\begin{aligned}&a \left(20 a + 38 b + 36 c\right) + b \left(12 b + 22 c\right) + 10 c^{2}\\ &+h \left(8 a + 8 b + 8 c\right)\end{aligned}\right)\end{aligned}\right)\end{aligned}\right)\\ &+\eta\left(\begin{aligned}&a \left(a \left(16 a + 28 b + 16 c\right) + b \left(12 b + 18 c\right) + 8 c^{2}\right)\\ &+h \left(\begin{aligned}&a \left(44 a + 74 b + 68 c\right) + b \left(20 b + 34 c\right) + 14 c^{2}\\ &+h \left(32 a + 32 b + 32 c\right)\end{aligned}\right)\end{aligned}\right)\\ &+\eta^{2}\left(\begin{aligned}&a \left(20 a + 26 b + 24 c\right) + b \left(4 b + 6 c\right) + 2 c^{2}\\ &+h \left(32 a + 32 b + 32 c\right)\end{aligned}\right)\\ &+\eta^{3}\left(\begin{aligned}&8 a + 8 b + 8 c\end{aligned}\right).
\end{aligned}
\]

For $\sigma=+-++$,

\[
\begin{aligned}
\frac{1}{h}\Theta_{\sigma}(0)={}&\left(\begin{aligned}&a \left(a b \left(2 b + 2 c\right) + b \left(b \left(2 b + 4 c\right) + 2 c^{2}\right)\right)\\ &+h \left(\begin{aligned}&a \left(a \left(10 b + 10 c\right) + b \left(14 b + 22 c\right) + 8 c^{2}\right)\\ &+b \left(b \left(4 b + 8 c\right) + 4 c^{2}\right)\\ &+h \left(a \left(10 b + 10 c\right) + b \left(8 b + 16 c\right) + 8 c^{2}\right)\end{aligned}\right)\end{aligned}\right)\\ &+\eta\left(\begin{aligned}&a \left(a \left(4 b + 4 c\right) + b \left(4 b + 8 c\right) + 4 c^{2}\right)\\ &+h \left(a \left(14 b + 14 c\right) + b \left(8 b + 16 c\right) + 8 c^{2}\right)\end{aligned}\right)\\ &+\eta^{2}\left(\begin{aligned}&a \left(2 b + 2 c\right)\end{aligned}\right).
\end{aligned}
\]

For $\sigma=++--$,

\[
\begin{aligned}
\frac{1}{h}\Theta_{\sigma}(0)={}&\left(\begin{aligned}&a \left(\begin{aligned}&a \left(a \left(4 a + 8 b + 8 c\right) + b \left(6 b + 10 c\right) + 4 c^{2}\right)\\ &+b \left(b \left(2 b + 4 c\right) + 2 c^{2}\right)\end{aligned}\right)\\ &+h \left(\begin{aligned}&a \left(a \left(16 a + 38 b + 20 c\right) + b \left(24 b + 28 c\right) + 4 c^{2}\right)\\ &+b \left(b \left(4 b + 8 c\right) + 4 c^{2}\right)\\ &+h \left(a \left(20 a + 38 b + 12 c\right) + b \left(12 b + 12 c\right) + h \left(8 a + 8 b\right)\right)\end{aligned}\right)\end{aligned}\right)\\ &+\eta\left(\begin{aligned}&a \left(a \left(16 a + 28 b + 16 c\right) + b \left(12 b + 12 c\right)\right)\\ &+h \left(a \left(44 a + 74 b + 20 c\right) + b \left(20 b + 20 c\right) + h \left(32 a + 32 b\right)\right)\end{aligned}\right)\\ &+\eta^{2}\left(\begin{aligned}&a \left(20 a + 26 b + 4 c\right) + b \left(4 b + 4 c\right) + h \left(32 a + 32 b\right)\end{aligned}\right)\\ &+\eta^{3}\left(\begin{aligned}&8 a + 8 b\end{aligned}\right).
\end{aligned}
\]

For $\sigma\in\{++-+, +++-\}$,

\[
\begin{aligned}
\frac{1}{h}\Theta_{\sigma}(0)={}&\left(\begin{aligned}&a \left(\begin{aligned}&a \left(b \left(2 b + 2 c\right) + 2 c^{2}\right)\\ &+b \left(b \left(2 b + 3 c\right) + 2 c^{2}\right)\end{aligned}\right)\\ &+h \left(\begin{aligned}&a \left(a \left(10 b + 10 c\right) + b \left(14 b + 25 c\right) + 6 c^{2}\right)\\ &+b \left(b \left(4 b + 10 c\right) + 6 c^{2}\right)\\ &+h \left(a \left(10 b + 10 c\right) + b \left(8 b + 10 c\right)\right)\end{aligned}\right)\end{aligned}\right)\\ &+\eta\left(\begin{aligned}&a \left(a \left(4 b + 4 c\right) + b \left(4 b + 6 c\right)\right)\\ &+h \left(a \left(14 b + 14 c\right) + b \left(8 b + 14 c\right)\right)\end{aligned}\right)\\ &+\eta^{2}\left(\begin{aligned}&a \left(2 b + 2 c\right) + 2 b c\end{aligned}\right).
\end{aligned}
\]

For $\sigma=++++$,

\[
\begin{aligned}
\frac{1}{h}\Theta_{\sigma}(0)={}&\left(\begin{aligned}&a \left(a \left(2 b c + 4 c^{2}\right) + b \left(2 b c + 2 c^{2}\right)\right)\\ &+h \left(a \left(6 b c + 4 c^{2}\right) + b \left(4 b c + 4 c^{2}\right)\right)\end{aligned}\right).
\end{aligned}
\]
\endgroup

% VERSION: full-editable-v2, 2026-08-20
\subsection{Exterior deviations}
\label{app:k3-exterior}

Suppose $\abs{r}\geq\rho$, and let

$$
s=\sign(r)\rho
$$

be the nearest endpoint of the closest-point interval. We prove the
atomwise clipping inequality

$$
\Theta_V(r)\geq\Theta_V(s).
$$

By reflection, it is enough to treat an exterior report to the right.
Translate the clipped endpoint to $0$, so the clipped report is $s=0$
and the exterior report is $r=t>0$. The true location is the midpoint
of the closest-point interval. To cover the two possible endpoints at
which the deleted report may lie, introduce nonnegative parameters
$\alpha$ and $\beta$ and write

$$
I=[-(\alpha+\beta),0],
\qquad
p=-\beta.
$$

The two actual configurations are $(\alpha,\beta)=(2\rho,0)$ and
$(\alpha,\beta)=(0,2\rho)$. A witness to the left of $I$ with excess
$x\geq0$ is written as

$$
L(x)=-(\alpha+\beta)-x,
$$

and a witness to the right with excess $x\geq0$ is written as

$$
R(x)=x.
$$

For a witness set $S$, let $e(S)$ be its smallest excess.

\subsubsection{Exact clipping identity}

Let

$$
\Lambda_V
\eqdef
\Theta_V(r)-\Theta_V(s).
$$

For a pair $C\in\binom{V}{2}$, define
$\lambda_C=\min\{t,e(C)\}$, and for $q\in V$ define

$$
\Delta_q
=
W_3(\{p,q,r\})-W_3(\{p,q,s\}).
$$

In either actual endpoint configuration, the true location is the
midpoint of $I$, every witness is at distance $\rho$ plus its excess,
and the strategic report has the same distance $\rho$ from the true
location before and after clipping. Direct subtraction of the two
reduced atoms gives

\[
\begin{aligned}
\Lambda_V
={}&
\sum_{\substack{
E,C\in\binom{V}{2}\\
E\cup C=V
}}
W_3(E\cup\{p\})W_3(C\cup\{r\})\lambda_C\\
&-
\sum_{\substack{
F\in\binom{V}{3},\;q\in V\\
F\cup\{q\}=V
}}
W_3(F)e(F)\Delta_q.
\end{aligned}
\]

If $\abs{V}=2$, the second sum is empty, so $\Lambda_V\geq0$.

\subsubsection{Three shadow inequalities}

The remaining proof uses the following purely algebraic inequalities.
Their ordered-gap expansions are printed in the exterior table file.

\begin{lemma}[Shadow inequalities]
\label{lem:app-k3-shadow}
For a three-element witness set $V$ and arbitrary anchors $p,r$,

\[
\sum_{\substack{
E,C\in\binom{V}{2}\\
E\cup C=V
}}
W_3(E\cup\{p\})W_3(C\cup\{r\})
\geq
W_3(V)\sum_{q\in V}W_3(\{p,q,r\}).
\]

For a four-element witness set $V$,

\[
\sum_{E\in\binom{V}{2}}
W_3(E\cup\{p\})W_3((V\setminus E)\cup\{r\})
\geq
\sum_{q\in V}
W_3(V\setminus\{q\})W_3(\{p,q,r\}).
\]

If $c\in V$ is distinguished, then

\[
\sum_{\substack{E\in\binom{V}{2}\\c\in E}}
W_3(E\cup\{p\})W_3((V\setminus E)\cup\{r\})
\geq
W_3(V\setminus\{c\})W_3(\{p,c,r\}).
\]
\end{lemma}

\begin{proof}
Fix the complete left-to-right order of the anchors and witnesses.
Write every coordinate as a cumulative sum of consecutive gaps. In
each possible rank pattern, the difference between the two sides is an
exact polynomial with nonnegative coefficients. Reflection and the
symmetries of the two anchors reduce the checks to the finite list
printed in the exterior table file.
\end{proof}

\subsubsection{Three witnesses}

Suppose $\abs{V}=3$, and let $e_1=e(V)$. Define

$$
P_3
=
\sum_{\substack{
E,C\in\binom{V}{2}\\
E\cup C=V
}}
W_3(E\cup\{p\})W_3(C\cup\{r\})
$$

and

$$
D_3=\sum_{q\in V}\Delta_q.
$$

The clipping identity gives

$$
\Lambda_V
\geq
\min\{t,e_1\}P_3-e_1W_3(V)D_3.
$$

If $t\geq e_1$, then
$\Delta_q\leq W_3(\{p,q,r\})$, and the three-witness shadow
inequality proves $\Lambda_V\geq0$.

Suppose $t\leq e_1$, and write $e_1=t+u$ with $u\geq0$. Choose a
closest witness $c$. Its side is either $L$ or $R$. Among the other
two witnesses, let $j\in\{0,1,2\}$ be the number lying on the left.
These choices give six exhaustive geometries. The remaining left
witnesses have successive excesses $e_1+\ell_1$ and
$e_1+\ell_1+\ell_2$, as needed, and the remaining right witnesses have
successive excesses $e_1+r_1$ and $e_1+r_1+r_2$.

For each of the six geometries, substitution into the clipping identity
gives an exact Horner polynomial in nonnegative variables. The six
identities are printed in the exterior table file, so
$\Lambda_V\geq0$ for every three-witness atom.

\subsubsection{Four witnesses}

Suppose $\abs{V}=4$. Let $c_1,c_2$ be witnesses with the two smallest
excesses $e_1\leq e_2$. Split the positive part of the clipping
identity according to whether $c_1$ belongs to the pair selected with
$p$:

$$
P_0
=
\sum_{\substack{E\in\binom{V}{2}\\c_1\notin E}}
W_3(E\cup\{p\})W_3((V\setminus E)\cup\{r\})
$$

and

$$
P_1
=
\sum_{\substack{E\in\binom{V}{2}\\c_1\in E}}
W_3(E\cup\{p\})W_3((V\setminus E)\cup\{r\}).
$$

For $q\in V$, put $C_q=W_3(V\setminus\{q\})$. Then

\[
\Lambda_V
\geq
\min\{t,e_1\}P_0
+
\min\{t,e_2\}P_1
-
e_1\sum_{q\neq c_1}C_q\Delta_q
-
e_2C_{c_1}\Delta_{c_1}.
\]

If $t\geq e_2$, write

$$
e_1P_0+e_2P_1
=
e_1(P_0+P_1)+(e_2-e_1)P_1.
$$

Apply the four-witness shadow inequality to the first term and the
distinguished shadow inequality with $c=c_1$ to the second. This gives
$\Lambda_V\geq0$.

Two radial regimes remain. In the first,

$$
t\leq e_1,
\qquad
e_1=t+u,
\qquad
e_2=e_1+v,
$$

with $u,v\geq0$. In the second,

$$
e_1=e,
\qquad
t=e+u,
\qquad
e_2=e+u+v,
$$

with $e,u,v\geq0$.

The sides of $c_1$ and $c_2$ give four possibilities. Let
$j\in\{0,1,2\}$ be the number of left witnesses among the remaining
pair. Thus there are twelve geometries in each radial regime. The
remaining left witnesses have successive excesses
$e_2+\ell_1$ and $e_2+\ell_1+\ell_2$, while the remaining right
witnesses have successive excesses $e_2+r_1$ and
$e_2+r_1+r_2$.

For every geometry and both radial regimes, substitution gives an exact
Horner polynomial with nonnegative coefficients. The resulting
twenty-four identities are printed in the exterior table file. Hence
$\Lambda_V\geq0$ for every four-witness atom.

Combining the two-, three-, and four-witness cases proves

$$
\Theta_V(r)\geq\Theta_V(s)
$$

for every exterior report. If $s=\rho$, the clipped report belongs to
the same-side interior chamber. If $s=-\rho$, it is the boundary of
the opposite-side interior chamber. Therefore every exterior reduced
atom is nonnegative.

\begingroup
\small
\setlength{\jot}{3pt}
\allowdisplaybreaks
% VERSION: compact-editable-v4, 2026-08-20
% This file contains all exterior certificate formulas as native LaTeX.
% It has no PDF inclusion, image inclusion, URL, or external-table dependency.
% Compact full-width exterior certificate for the k=3 closest-point deletion proof.
% Reflowed across the full local text width; formula content is unchanged.
% Local notation is consistent with Appendix~\ref{app:k3-exterior}:
%   source A,B       -> \alpha,\beta;
%   source case k    -> j (to avoid conflict with the facility count);
%   source l_1,l_2   -> \ell_1,\ell_2.

% Preamble requirement for the local full-width table block:
%   \usepackage{changepage}
% The matching certificate file applies \footnotesize and a local
% adjustwidth environment; the rest of the paper keeps its layout.

\subsubsection{Exact exterior coefficient identities}
\label{app:k3-exterior-tables}

All variables in this subsection are nonnegative. The formulas below
are exact identities, not numerical checks. We use $\alpha,\beta$ for
the two endpoint parameters introduced in
Appendix~\ref{app:k3-exterior}, and $j$ for the number of left
witnesses among the undistinguished reports. In Regime~$\mathrm{M}$,
we write $e_1=e$, $t=e+u$, and $e_2=e+u+v$.

\paragraph{Ordered-gap certificates for the shadow inequalities.}
For ordered points $z_0<\cdots<z_{n-1}$, let
$g_i=z_i-z_{i-1}$. Each displayed expression is exactly the left-hand side minus the
right-hand side of the indicated shadow inequality after substituting
the stated ranks. To save space, we do not repeat ``$\geq0$'' after
every row: each displayed right-hand side is a polynomial with
nonnegative coefficients.

\medskip\noindent\textbf{Three-witness shadow.}

\noindent\textit{Anchor ranks $(0,1)$.}

\[
\begin{aligned}
&g_{1} g_{2} \bigl(g_{3} \bigl(2 g_{3} + 3 g_{4}\bigr) + 2 g_{4}^{2}\bigr) + g_{2} \bigl(g_{2} \bigl(g_{3} \bigl(2 g_{3} + 6 g_{4}\bigr) + 2 g_{4}^{2}\bigr) + g_{3} \bigl(4 g_{3} g_{4} + 2 g_{4}^{2}\bigr)\bigr)
\end{aligned}
\]

\noindent\textit{Anchor ranks $(0,2)$.}

\[
\begin{aligned}
&g_{1} g_{2} \bigl(2 g_{2} g_{3} + g_{3} \bigl(2 g_{3} + 3 g_{4}\bigr)\bigr)
\end{aligned}
\]

\noindent\textit{Anchor ranks $(0,3)$.}

\[
\begin{aligned}
&g_{1} \bigl(3 g_{2} g_{3} g_{4} + g_{3} \bigl(2 g_{3} g_{4} + 2 g_{4}^{2}\bigr)\bigr)
\end{aligned}
\]

\noindent\textit{Anchor ranks $(0,4)$.}

\[
\begin{aligned}
&g_{1} \bigl(g_{2} \bigl(2 g_{2} g_{4} + 3 g_{3} g_{4}\bigr) + 2 g_{3}^{2} g_{4}\bigr)
\end{aligned}
\]

\noindent\textit{Anchor ranks $(1,2)$.}

\[
\begin{aligned}
&g_{1} \bigl(g_{1} \bigl(2 g_{2} g_{3} + g_{3} \bigl(2 g_{3} + 2 g_{4}\bigr)\bigr) + g_{2} \bigl(2 g_{2} g_{3} + g_{3} \bigl(2 g_{3} + 3 g_{4}\bigr)\bigr) + g_{3} \bigl(4 g_{3} g_{4} + 2 g_{4}^{2}\bigr)\bigr)
\end{aligned}
\]

\noindent\textit{Anchor ranks $(1,3)$.}

\[
\begin{aligned}
&3 g_{1} g_{2} g_{3} g_{4}
\end{aligned}
\]

\medskip\noindent\textbf{Four-witness shadow.}

\noindent\textit{Anchor ranks $(0,1)$.}

\[
\begin{aligned}
&g_{1} g_{2} \bigl(g_{3} \bigl(2 g_{4} + 2 g_{5}\bigr) + g_{4} \bigl(4 g_{4} + 2 g_{5}\bigr)\bigr) + g_{2} \bigl(g_{2} \bigl(g_{3} \bigl(4 g_{4} + 4 g_{5}\bigr) + g_{4} \bigl(4 g_{4} + 4 g_{5}\bigr)\bigr) + g_{3} \bigl(g_{3} \bigl(4 g_{4} + 4 g_{5}\bigr)
\\[-1pt]
&\quad{}+ g_{4} \bigl(4 g_{4} + 6 g_{5}\bigr)\bigr)\bigr)
\end{aligned}
\]

\noindent\textit{Anchor ranks $(0,2)$.}

\[
\begin{aligned}
&g_{1} g_{2} g_{3} \bigl(2 g_{4} + 2 g_{5}\bigr)
\end{aligned}
\]

\noindent\textit{Anchor ranks $(0,3)$.}

\[
\begin{aligned}
&g_{1} \bigl(2 g_{2} g_{3} g_{4} + g_{3} \bigl(4 g_{3} g_{4} + g_{4} \bigl(4 g_{4} + 2 g_{5}\bigr)\bigr)\bigr)
\end{aligned}
\]

\noindent\textit{Anchor ranks $(0,4)$.}

\[
\begin{aligned}
&g_{1} \bigl(2 g_{2} g_{4} g_{5} + 2 g_{3} g_{4} g_{5}\bigr)
\end{aligned}
\]

\noindent\textit{Anchor ranks $(0,5)$.}

\[
\begin{aligned}
&g_{1} \bigl(g_{2} \bigl(2 g_{3} g_{5} + 2 g_{4} g_{5}\bigr) + g_{3} \bigl(4 g_{3} g_{5} + 2 g_{4} g_{5}\bigr)\bigr)
\end{aligned}
\]

\noindent\textit{Anchor ranks $(1,2)$.}

\[
\begin{aligned}
&g_{1} \bigl(g_{2} g_{3} \bigl(2 g_{4} + 2 g_{5}\bigr) + g_{3} \bigl(g_{3} \bigl(4 g_{4} + 4 g_{5}\bigr) + g_{4} \bigl(4 g_{4} + 6 g_{5}\bigr)\bigr)\bigr)
\end{aligned}
\]

\noindent\textit{Anchor ranks $(1,3)$.}

\[
\begin{aligned}
&2 g_{1} g_{2} g_{3} g_{4}
\end{aligned}
\]

\noindent\textit{Anchor ranks $(1,4)$.}

\[
\begin{aligned}
&2 g_{1} g_{2} g_{4} g_{5}
\end{aligned}
\]

\noindent\textit{Anchor ranks $(2,3)$.}

\[
\begin{aligned}
&g_{1} g_{2} \bigl(2 g_{3} g_{4} + g_{4} \bigl(4 g_{4} + 6 g_{5}\bigr)\bigr) + g_{2} \bigl(g_{2} \bigl(4 g_{3} g_{4} + g_{4} \bigl(4 g_{4} + 4 g_{5}\bigr)\bigr) + g_{3} \bigl(4 g_{3} g_{4} + g_{4} \bigl(4 g_{4} + 2 g_{5}\bigr)\bigr)\bigr)
\end{aligned}
\]

\medskip\noindent\textbf{Distinguished four-witness shadow.}

\noindent\textit{Ranks $(p,c_1;r)=(0,1;2)$.}

\[
\begin{aligned}
&g_{1} \bigl(g_{2} g_{3} \bigl(2 g_{4} + 2 g_{5}\bigr) + g_{3} \bigl(g_{3} \bigl(2 g_{4} + 2 g_{5}\bigr) + g_{4} \bigl(2 g_{4} + 3 g_{5}\bigr)\bigr)\bigr)
\end{aligned}
\]

\noindent\textit{Ranks $(p,c_1;r)=(0,1;3)$.}

\[
\begin{aligned}
&g_{1} \bigl(2 g_{2} g_{3} g_{4} + g_{3} \bigl(2 g_{3} g_{4} + g_{4} \bigl(2 g_{4} + g_{5}\bigr)\bigr)\bigr)
\end{aligned}
\]

\noindent\textit{Ranks $(p,c_1;r)=(0,1;4)$.}

\[
\begin{aligned}
&g_{1} \bigl(2 g_{2} g_{4} g_{5} + g_{3} g_{4} g_{5}\bigr)
\end{aligned}
\]

\noindent\textit{Ranks $(p,c_1;r)=(0,1;5)$.}

\[
\begin{aligned}
&g_{1} \bigl(g_{2} \bigl(2 g_{3} g_{5} + 2 g_{4} g_{5}\bigr) + g_{3} \bigl(2 g_{3} g_{5} + g_{4} g_{5}\bigr)\bigr)
\end{aligned}
\]

\noindent\textit{Ranks $(p,c_1;r)=(0,2;1)$.}

\[
\begin{aligned}
&g_{1} \bigl(g_{2} \bigl(g_{3} \bigl(2 g_{4} + 2 g_{5}\bigr) + g_{4} \bigl(2 g_{4} + g_{5}\bigr)\bigr) + g_{3} \bigl(g_{3} \bigl(2 g_{4} + 2 g_{5}\bigr) + g_{4} \bigl(2 g_{4} + 3 g_{5}\bigr)\bigr)\bigr) + g_{2} \bigl(g_{2} \bigl(g_{3} \bigl(2 g_{4} + 2 g_{5}\bigr)
\\[-1pt]
&\quad{}+ g_{4} \bigl(2 g_{4} + 2 g_{5}\bigr)\bigr) + g_{3} \bigl(g_{3} \bigl(2 g_{4} + 2 g_{5}\bigr) + g_{4} \bigl(2 g_{4} + 3 g_{5}\bigr)\bigr)\bigr)
\end{aligned}
\]

\noindent\textit{Ranks $(p,c_1;r)=(0,2;3)$.}

\[
\begin{aligned}
&g_{1} \bigl(g_{2} \bigl(2 g_{3} g_{4} + g_{4} \bigl(2 g_{4} + 3 g_{5}\bigr)\bigr) + g_{3} \bigl(2 g_{3} g_{4} + g_{4} \bigl(2 g_{4} + g_{5}\bigr)\bigr)\bigr) + g_{2} \bigl(g_{2} \bigl(2 g_{3} g_{4} + g_{4} \bigl(2 g_{4} + 2 g_{5}\bigr)\bigr) + g_{3} \bigl(2 g_{3} g_{4}
\\[-1pt]
&\quad{}+ g_{4} \bigl(2 g_{4} + g_{5}\bigr)\bigr)\bigr)
\end{aligned}
\]

\noindent\textit{Ranks $(p,c_1;r)=(0,2;4)$.}

\[
\begin{aligned}
&g_{1} \bigl(g_{2} g_{4} g_{5} + g_{3} g_{4} g_{5}\bigr) + g_{2} g_{3} g_{4} g_{5}
\end{aligned}
\]

\noindent\textit{Ranks $(p,c_1;r)=(0,2;5)$.}

\[
\begin{aligned}
&g_{1} \bigl(g_{2} \bigl(2 g_{3} g_{5} + g_{4} g_{5}\bigr) + g_{3} \bigl(2 g_{3} g_{5} + g_{4} g_{5}\bigr)\bigr) + g_{2} \bigl(2 g_{2} g_{3} g_{5} + g_{3} \bigl(2 g_{3} g_{5} + g_{4} g_{5}\bigr)\bigr)
\end{aligned}
\]

\noindent\textit{Ranks $(p,c_1;r)=(0,3;1)$.}

\[
\begin{aligned}
&g_{1} g_{2} \bigl(g_{3} \bigl(2 g_{4} + g_{5}\bigr) + g_{4} \bigl(2 g_{4} + g_{5}\bigr)\bigr) + g_{2} \bigl(g_{2} \bigl(g_{3} \bigl(2 g_{4} + 2 g_{5}\bigr) + g_{4} \bigl(2 g_{4} + 2 g_{5}\bigr)\bigr) + g_{3} \bigl(g_{3} \bigl(2 g_{4} + 2 g_{5}\bigr) + g_{4} \bigl(2 g_{4}
\\[-1pt]
&\quad{}+ 3 g_{5}\bigr)\bigr)\bigr)
\end{aligned}
\]

\noindent\textit{Ranks $(p,c_1;r)=(0,3;2)$.}

\[
\begin{aligned}
&g_{1} g_{2} \bigl(g_{3} \bigl(2 g_{4} + g_{5}\bigr) + g_{4} \bigl(2 g_{4} + 3 g_{5}\bigr)\bigr) + g_{2} \bigl(g_{2} \bigl(2 g_{3} g_{4} + g_{4} \bigl(2 g_{4} + 2 g_{5}\bigr)\bigr) + g_{3} \bigl(2 g_{3} g_{4} + g_{4} \bigl(2 g_{4} + g_{5}\bigr)\bigr)\bigr)
\end{aligned}
\]

\noindent\textit{Ranks $(p,c_1;r)=(0,3;4)$.}

\[
\begin{aligned}
&g_{1} \bigl(g_{2} \bigl(3 g_{3} g_{5} + g_{4} g_{5}\bigr) + g_{3} \bigl(2 g_{3} g_{5} + 2 g_{4} g_{5}\bigr)\bigr) + g_{2} \bigl(2 g_{2} g_{3} g_{5} + g_{3} \bigl(2 g_{3} g_{5} + g_{4} g_{5}\bigr)\bigr)
\end{aligned}
\]

\noindent\textit{Ranks $(p,c_1;r)=(0,3;5)$.}

\[
\begin{aligned}
&g_{1} \bigl(g_{2} \bigl(g_{3} g_{5} + g_{4} g_{5}\bigr) + g_{3} \bigl(2 g_{3} g_{5} + 2 g_{4} g_{5}\bigr)\bigr) + g_{2} g_{3} g_{4} g_{5}
\end{aligned}
\]

\noindent\textit{Ranks $(p,c_1;r)=(0,4;1)$.}

\[
\begin{aligned}
&g_{1} g_{2} \bigl(g_{3} \bigl(g_{4} + g_{5}\bigr) + g_{4} \bigl(2 g_{4} + 2 g_{5}\bigr)\bigr) + g_{2} \bigl(g_{2} \bigl(g_{3} \bigl(2 g_{4} + 2 g_{5}\bigr) + g_{4} \bigl(2 g_{4} + 2 g_{5}\bigr)\bigr) + g_{3} \bigl(g_{3} \bigl(2 g_{4} + 2 g_{5}\bigr) + g_{4} \bigl(2 g_{4}
\\[-1pt]
&\quad{}+ 3 g_{5}\bigr)\bigr)\bigr)
\end{aligned}
\]

\noindent\textit{Ranks $(p,c_1;r)=(0,4;2)$.}

\[
\begin{aligned}
&g_{1} g_{2} g_{3} \bigl(g_{4} + g_{5}\bigr) + g_{2} g_{3} g_{4} g_{5}
\end{aligned}
\]

\noindent\textit{Ranks $(p,c_1;r)=(0,4;3)$.}

\[
\begin{aligned}
&g_{1} \bigl(g_{2} g_{3} \bigl(g_{4} + 3 g_{5}\bigr) + g_{3} \bigl(g_{3} \bigl(2 g_{4} + 2 g_{5}\bigr) + g_{4} \bigl(2 g_{4} + 2 g_{5}\bigr)\bigr)\bigr) + g_{2} \bigl(2 g_{2} g_{3} g_{5} + g_{3} \bigl(2 g_{3} g_{5} + g_{4} g_{5}\bigr)\bigr)
\end{aligned}
\]

\noindent\textit{Ranks $(p,c_1;r)=(0,4;5)$.}

\[
\begin{aligned}
&g_{1} \bigl(g_{2} \bigl(g_{3} \bigl(3 g_{4} + g_{5}\bigr) + g_{4} \bigl(2 g_{4} + 2 g_{5}\bigr)\bigr) + g_{3} \bigl(g_{3} \bigl(2 g_{4} + 2 g_{5}\bigr) + g_{4} \bigl(2 g_{4} + 2 g_{5}\bigr)\bigr)\bigr) + g_{2} \bigl(g_{2} \bigl(2 g_{3} g_{4} + g_{4} \bigl(2 g_{4} + 2 g_{5}\bigr)\bigr)
\\[-1pt]
&\quad{}+ g_{3} \bigl(2 g_{3} g_{4} + g_{4} \bigl(2 g_{4} + g_{5}\bigr)\bigr)\bigr)
\end{aligned}
\]

\noindent\textit{Ranks $(p,c_1;r)=(0,5;1)$.}

\[
\begin{aligned}
&g_{1} g_{2} \bigl(g_{3} \bigl(g_{4} + 2 g_{5}\bigr) + g_{4} \bigl(2 g_{4} + 2 g_{5}\bigr)\bigr) + g_{2} \bigl(g_{2} \bigl(g_{3} \bigl(2 g_{4} + 2 g_{5}\bigr) + g_{4} \bigl(2 g_{4} + 2 g_{5}\bigr)\bigr) + g_{3} \bigl(g_{3} \bigl(2 g_{4} + 2 g_{5}\bigr) + g_{4} \bigl(2 g_{4}
\\[-1pt]
&\quad{}+ 3 g_{5}\bigr)\bigr)\bigr)
\end{aligned}
\]

\noindent\textit{Ranks $(p,c_1;r)=(0,5;2)$.}

\[
\begin{aligned}
&g_{1} g_{2} g_{3} \bigl(g_{4} + 2 g_{5}\bigr) + g_{2} \bigl(2 g_{2} g_{3} g_{5} + g_{3} \bigl(2 g_{3} g_{5} + g_{4} g_{5}\bigr)\bigr)
\end{aligned}
\]

\noindent\textit{Ranks $(p,c_1;r)=(1,2;0)$.}

\[
\begin{aligned}
&g_{1} g_{2} \bigl(g_{3} \bigl(2 g_{4} + 2 g_{5}\bigr) + g_{4} \bigl(2 g_{4} + g_{5}\bigr)\bigr) + g_{2} \bigl(g_{2} \bigl(g_{3} \bigl(2 g_{4} + 2 g_{5}\bigr) + g_{4} \bigl(2 g_{4} + 2 g_{5}\bigr)\bigr) + g_{3} \bigl(g_{3} \bigl(2 g_{4} + 2 g_{5}\bigr) + g_{4} \bigl(2 g_{4}
\\[-1pt]
&\quad{}+ 3 g_{5}\bigr)\bigr)\bigr)
\end{aligned}
\]

\noindent\textit{Ranks $(p,c_1;r)=(1,2;3)$.}

\[
\begin{aligned}
&g_{1} g_{2} \bigl(2 g_{3} g_{4} + g_{4} \bigl(2 g_{4} + 3 g_{5}\bigr)\bigr) + g_{2} \bigl(g_{2} \bigl(2 g_{3} g_{4} + g_{4} \bigl(2 g_{4} + 2 g_{5}\bigr)\bigr) + g_{3} \bigl(2 g_{3} g_{4} + g_{4} \bigl(2 g_{4} + g_{5}\bigr)\bigr)\bigr)
\end{aligned}
\]

\noindent\textit{Ranks $(p,c_1;r)=(1,2;4)$.}

\[
\begin{aligned}
&g_{1} g_{2} g_{4} g_{5} + g_{2} g_{3} g_{4} g_{5}
\end{aligned}
\]

\noindent\textit{Ranks $(p,c_1;r)=(1,2;5)$.}

\[
\begin{aligned}
&g_{1} g_{2} \bigl(2 g_{3} g_{5} + g_{4} g_{5}\bigr) + g_{2} \bigl(2 g_{2} g_{3} g_{5} + g_{3} \bigl(2 g_{3} g_{5} + g_{4} g_{5}\bigr)\bigr)
\end{aligned}
\]

\noindent\textit{Ranks $(p,c_1;r)=(1,3;0)$.}

\[
\begin{aligned}
&g_{1} \bigl(g_{2} \bigl(g_{3} \bigl(2 g_{4} + g_{5}\bigr) + g_{4} \bigl(2 g_{4} + g_{5}\bigr)\bigr) + g_{3} \bigl(2 g_{3} g_{4} + g_{4} \bigl(2 g_{4} + g_{5}\bigr)\bigr)\bigr) + g_{2} \bigl(g_{2} \bigl(g_{3} \bigl(2 g_{4} + 2 g_{5}\bigr) + g_{4} \bigl(2 g_{4} + 2 g_{5}\bigr)\bigr)
\\[-1pt]
&\quad{}+ g_{3} \bigl(g_{3} \bigl(2 g_{4} + 2 g_{5}\bigr) + g_{4} \bigl(2 g_{4} + 3 g_{5}\bigr)\bigr)\bigr)
\end{aligned}
\]

\noindent\textit{Ranks $(p,c_1;r)=(1,3;2)$.}

\[
\begin{aligned}
&g_{1} \bigl(g_{2} \bigl(g_{3} \bigl(2 g_{4} + g_{5}\bigr) + g_{4} \bigl(2 g_{4} + 3 g_{5}\bigr)\bigr) + g_{3} \bigl(g_{3} \bigl(2 g_{4} + 2 g_{5}\bigr) + g_{4} \bigl(2 g_{4} + 3 g_{5}\bigr)\bigr)\bigr) + g_{2} \bigl(g_{2} \bigl(2 g_{3} g_{4} + g_{4} \bigl(2 g_{4} + 2 g_{5}\bigr)\bigr)
\\[-1pt]
&\quad{}+ g_{3} \bigl(2 g_{3} g_{4} + g_{4} \bigl(2 g_{4} + g_{5}\bigr)\bigr)\bigr)
\end{aligned}
\]

\noindent\textit{Ranks $(p,c_1;r)=(1,3;4)$.}

\[
\begin{aligned}
&g_{1} \bigl(g_{2} \bigl(3 g_{3} g_{5} + g_{4} g_{5}\bigr) + g_{3} \bigl(2 g_{3} g_{5} + g_{4} g_{5}\bigr)\bigr) + g_{2} \bigl(2 g_{2} g_{3} g_{5} + g_{3} \bigl(2 g_{3} g_{5} + g_{4} g_{5}\bigr)\bigr)
\end{aligned}
\]

\noindent\textit{Ranks $(p,c_1;r)=(1,3;5)$.}

\[
\begin{aligned}
&g_{1} \bigl(g_{2} \bigl(g_{3} g_{5} + g_{4} g_{5}\bigr) + g_{3} g_{4} g_{5}\bigr) + g_{2} g_{3} g_{4} g_{5}
\end{aligned}
\]

\noindent\textit{Ranks $(p,c_1;r)=(1,4;0)$.}

\[
\begin{aligned}
&g_{1} \bigl(g_{2} \bigl(g_{3} \bigl(g_{4} + g_{5}\bigr) + g_{4} \bigl(2 g_{4} + 2 g_{5}\bigr)\bigr) + g_{3} g_{4} g_{5}\bigr) + g_{2} \bigl(g_{2} \bigl(g_{3} \bigl(2 g_{4} + 2 g_{5}\bigr) + g_{4} \bigl(2 g_{4} + 2 g_{5}\bigr)\bigr) + g_{3} \bigl(g_{3} \bigl(2 g_{4} + 2 g_{5}\bigr)
\\[-1pt]
&\quad{}+ g_{4} \bigl(2 g_{4} + 3 g_{5}\bigr)\bigr)\bigr)
\end{aligned}
\]

\noindent\textit{Ranks $(p,c_1;r)=(1,4;2)$.}

\[
\begin{aligned}
&g_{1} \bigl(g_{2} g_{3} \bigl(g_{4} + g_{5}\bigr) + g_{3} \bigl(g_{3} \bigl(2 g_{4} + 2 g_{5}\bigr) + g_{4} \bigl(2 g_{4} + 3 g_{5}\bigr)\bigr)\bigr) + g_{2} g_{3} g_{4} g_{5}
\end{aligned}
\]

\noindent\textit{Ranks $(p,c_1;r)=(2,3;0)$.}

\[
\begin{aligned}
&g_{1} \bigl(g_{2} g_{3} \bigl(2 g_{4} + g_{5}\bigr) + g_{3} \bigl(2 g_{3} g_{4} + g_{4} \bigl(2 g_{4} + g_{5}\bigr)\bigr)\bigr)
\end{aligned}
\]

\noindent\textit{Ranks $(p,c_1;r)=(2,3;1)$.}

\[
\begin{aligned}
&g_{1} \bigl(g_{2} g_{3} \bigl(2 g_{4} + g_{5}\bigr) + g_{3} \bigl(g_{3} \bigl(2 g_{4} + 2 g_{5}\bigr) + g_{4} \bigl(2 g_{4} + 3 g_{5}\bigr)\bigr)\bigr)
\end{aligned}
\]

\paragraph{Six exact three-witness clipping identities.}
For a side $\varepsilon\in\{L,R\}$ and
$j\in\{0,1,2\}$, let
$\mathcal E^{(3)}_{\varepsilon,j}$ denote the right-hand side of the
exact clipping identity after the corresponding coordinate
substitution. Here $e_1=t+u$.

\noindent\textit{Case $(\varepsilon,j)=(L,0)$.}

\[
\begin{aligned}
&\mathcal{E}^{(3)}_{L,0}=\alpha \bigl(\alpha \bigl(\beta t \bigl(2 r_{1} + r_{2} + 2 u\bigr) + t \bigl(r_{1} \bigl(2 r_{1} + 2 r_{2}\bigr) + t \bigl(2 r_{1} + 2 u\bigr) + u \bigl(4 r_{1} + r_{2} + 2 u\bigr)\bigr)\bigr) + \beta \bigl(\beta t \bigl(2 r_{1} + r_{2}
\\[-1pt]
&\quad{}+ 2 u\bigr) + t \bigl(r_{1} \bigl(2 r_{1} + 6 r_{2}\bigr) + r_{2}^{2} + t \bigl(8 r_{1} + 5 r_{2} + 8 u\bigr) + u \bigl(8 r_{1} + 9 r_{2} + 6 u\bigr)\bigr)\bigr) + t \bigl(r_{1} \bigl(4 r_{1} r_{2} + 2 r_{2}^{2}\bigr)
\\[-1pt]
&\quad{}+ t \bigl(r_{1} \bigl(6 r_{1} + 7 r_{2}\bigr) + t \bigl(6 r_{1} + 6 u\bigr) + u \bigl(16 r_{1} + 4 r_{2} + 10 u\bigr)\bigr) + u \bigl(r_{1} \bigl(4 r_{1} + 9 r_{2}\bigr) + r_{2}^{2} + u \bigl(8 r_{1} + 3 r_{2} + 4 u\bigr)\bigr)\bigr)\bigr)
\\[-1pt]
&\quad{}+ \beta \bigl(\beta t \bigl(2 r_{1} r_{2} + r_{2}^{2} + t \bigl(2 r_{1} + 3 r_{2} + 2 u\bigr) + u \bigl(2 r_{1} + 5 r_{2} + 2 u\bigr)\bigr) + t \bigl(r_{1} \bigl(2 r_{1} r_{2} + r_{2}^{2}\bigr) + t \bigl(r_{1} \bigl(2 r_{1} + 10 r_{2}\bigr)
\\[-1pt]
&\quad{}+ 2 r_{2}^{2} + t \bigl(6 r_{1} + 6 r_{2} + 6 u\bigr) + u \bigl(12 r_{1} + 18 r_{2} + 10 u\bigr)\bigr) + u \bigl(r_{1} \bigl(2 r_{1} + 10 r_{2}\bigr) + r_{2}^{2} + u \bigl(6 r_{1} + 10 r_{2} + 4 u\bigr)\bigr)\bigr)\bigr)
\\[-1pt]
&\quad{}+ t \bigl(t \bigl(r_{1} \bigl(4 r_{1} r_{2} + 2 r_{2}^{2}\bigr) + t \bigl(r_{1} \bigl(4 r_{1} + 5 r_{2}\bigr) + t \bigl(4 r_{1} + 4 u\bigr) + u \bigl(14 r_{1} + 3 r_{2} + 10 u\bigr)\bigr) + u \bigl(r_{1} \bigl(6 r_{1} + 10 r_{2}\bigr)
\\[-1pt]
&\quad{}+ u \bigl(14 r_{1} + 3 r_{2} + 8 u\bigr)\bigr)\bigr) + u \bigl(r_{1} \bigl(2 r_{1} r_{2} + r_{2}^{2}\bigr) + u \bigl(r_{1} \bigl(2 r_{1} + 3 r_{2}\bigr) + u \bigl(4 r_{1} + 2 u\bigr)\bigr)\bigr)\bigr)
\end{aligned}
\]

\noindent\textit{Case $(\varepsilon,j)=(L,1)$.}

\[
\begin{aligned}
&\mathcal{E}^{(3)}_{L,1}=\alpha \bigl(\alpha \bigl(\beta t \bigl(2 \ell_{1} + 2 r_{1} + 2 u\bigr) + t \bigl(4 \ell_{1} r_{1} + 2 r_{1}^{2} + t \bigl(2 r_{1} + 2 u\bigr) + u \bigl(2 \ell_{1} + 4 r_{1} + 2 u\bigr)\bigr)\bigr) + \beta \bigl(\beta t \bigl(2 \ell_{1}
\\[-1pt]
&\quad{}+ 2 r_{1} + 2 u\bigr) + t \bigl(\ell_{1} \bigl(\ell_{1} + 6 r_{1}\bigr) + 2 r_{1}^{2} + t \bigl(7 \ell_{1} + 8 r_{1} + 8 u\bigr) + u \bigl(9 \ell_{1} + 8 r_{1} + 6 u\bigr)\bigr)\bigr) + t \bigl(\ell_{1} \bigl(2 \ell_{1} r_{1} + 2 r_{1}^{2}\bigr)
\\[-1pt]
&\quad{}+ t \bigl(11 \ell_{1} r_{1} + 6 r_{1}^{2} + t \bigl(6 r_{1} + 6 u\bigr) + u \bigl(5 \ell_{1} + 16 r_{1} + 10 u\bigr)\bigr) + u \bigl(\ell_{1} \bigl(\ell_{1} + 9 r_{1}\bigr) + 4 r_{1}^{2} + u \bigl(3 \ell_{1} + 8 r_{1} + 4 u\bigr)\bigr)\bigr)\bigr)
\\[-1pt]
&\quad{}+ \beta \bigl(\beta t \bigl(\ell_{1} \bigl(\ell_{1} + r_{1}\bigr) + t \bigl(3 \ell_{1} + 2 r_{1} + 2 u\bigr) + u \bigl(4 \ell_{1} + 2 r_{1} + 2 u\bigr)\bigr) + t \bigl(\ell_{1} \bigl(\ell_{1} r_{1} + r_{1}^{2}\bigr) + t \bigl(\ell_{1} \bigl(2 \ell_{1} + 8 r_{1}\bigr)
\\[-1pt]
&\quad{}+ 2 r_{1}^{2} + t \bigl(6 \ell_{1} + 6 r_{1} + 6 u\bigr) + u \bigl(15 \ell_{1} + 12 r_{1} + 10 u\bigr)\bigr) + u \bigl(\ell_{1} \bigl(\ell_{1} + 8 r_{1}\bigr) + 2 r_{1}^{2} + u \bigl(8 \ell_{1} + 6 r_{1} + 4 u\bigr)\bigr)\bigr)\bigr)
\\[-1pt]
&\quad{}+ t \bigl(t \bigl(\ell_{1} \bigl(2 \ell_{1} r_{1} + 2 r_{1}^{2}\bigr) + t \bigl(7 \ell_{1} r_{1} + 4 r_{1}^{2} + t \bigl(4 r_{1} + 4 u\bigr) + u \bigl(3 \ell_{1} + 14 r_{1} + 10 u\bigr)\bigr) + u \bigl(11 \ell_{1} r_{1} + 6 r_{1}^{2}
\\[-1pt]
&\quad{}+ u \bigl(3 \ell_{1} + 14 r_{1} + 8 u\bigr)\bigr)\bigr) + u \bigl(\ell_{1} \bigl(\ell_{1} r_{1} + r_{1}^{2}\bigr) + u \bigl(3 \ell_{1} r_{1} + 2 r_{1}^{2} + u \bigl(4 r_{1} + 2 u\bigr)\bigr)\bigr)\bigr)
\end{aligned}
\]

\noindent\textit{Case $(\varepsilon,j)=(L,2)$.}

\[
\begin{aligned}
&\mathcal{E}^{(3)}_{L,2}=\alpha \bigl(\alpha t \bigl(\ell_{1} \bigl(2 \ell_{1} + 6 \ell_{2}\bigr) + 2 \ell_{2}^{2}\bigr) + \beta t \bigl(\ell_{1} \bigl(2 \ell_{1} + 6 \ell_{2}\bigr) + 2 \ell_{2}^{2}\bigr) + t \bigl(\ell_{1} \bigl(4 \ell_{1} \ell_{2} + 2 \ell_{2}^{2}\bigr) + t \bigl(\ell_{1} \bigl(6 \ell_{1} + 15 \ell_{2}\bigr)
\\[-1pt]
&\quad{}+ 6 \ell_{2}^{2}\bigr) + u \bigl(\ell_{1} \bigl(4 \ell_{1} + 9 \ell_{2}\bigr) + 4 \ell_{2}^{2}\bigr)\bigr)\bigr) + \beta t \bigl(\ell_{1} \bigl(2 \ell_{1} \ell_{2} + \ell_{2}^{2}\bigr) + t \bigl(\ell_{1} \bigl(2 \ell_{1} + 6 \ell_{2}\bigr) + 2 \ell_{2}^{2}\bigr) + u \bigl(\ell_{1} \bigl(2 \ell_{1} + 6 \ell_{2}\bigr)
\\[-1pt]
&\quad{}+ 2 \ell_{2}^{2}\bigr)\bigr) + t \bigl(t \bigl(\ell_{1} \bigl(4 \ell_{1} \ell_{2} + 2 \ell_{2}^{2}\bigr) + t \bigl(\ell_{1} \bigl(4 \ell_{1} + 9 \ell_{2}\bigr) + 4 \ell_{2}^{2}\bigr) + u \bigl(\ell_{1} \bigl(6 \ell_{1} + 12 \ell_{2}\bigr) + 6 \ell_{2}^{2}\bigr)\bigr) + u \bigl(\ell_{1} \bigl(2 \ell_{1} \ell_{2}
\\[-1pt]
&\quad{}+ \ell_{2}^{2}\bigr) + u \bigl(\ell_{1} \bigl(2 \ell_{1} + 3 \ell_{2}\bigr) + 2 \ell_{2}^{2}\bigr)\bigr)\bigr)
\end{aligned}
\]

\noindent\textit{Case $(\varepsilon,j)=(R,0)$.}

\[
\begin{aligned}
&\mathcal{E}^{(3)}_{R,0}=\beta t \bigl(3 r_{1} r_{2} t + r_{1} \bigl(2 r_{1} r_{2} + r_{2}^{2}\bigr) + u \bigl(r_{1} \bigl(2 r_{1} + 9 r_{2}\bigr) + 2 r_{2}^{2}\bigr)\bigr) + t \bigl(t u \bigl(r_{1} \bigl(2 r_{1} + 3 r_{2}\bigr) + 2 r_{2}^{2}\bigr)
\\[-1pt]
&\quad{}+ u \bigl(r_{1} \bigl(2 r_{1} r_{2} + r_{2}^{2}\bigr) + u \bigl(r_{1} \bigl(2 r_{1} + 3 r_{2}\bigr) + 2 r_{2}^{2}\bigr)\bigr)\bigr)
\end{aligned}
\]

\noindent\textit{Case $(\varepsilon,j)=(R,1)$.}

\[
\begin{aligned}
&\mathcal{E}^{(3)}_{R,1}=\alpha \bigl(\alpha \bigl(\beta t \bigl(r_{1} + 2 u\bigr) + t \bigl(2 t u + u \bigl(r_{1} + 2 u\bigr)\bigr)\bigr) + \beta \bigl(\beta t \bigl(r_{1} + 2 u\bigr) + t \bigl(2 \ell_{1} r_{1} + r_{1}^{2} + t \bigl(5 r_{1} + 8 u\bigr)
\\[-1pt]
&\quad{}+ u \bigl(4 \ell_{1} + 9 r_{1} + 6 u\bigr)\bigr)\bigr) + t \bigl(t \bigl(6 t u + u \bigl(4 \ell_{1} + 4 r_{1} + 10 u\bigr)\bigr) + u \bigl(2 \ell_{1} r_{1} + r_{1}^{2} + u \bigl(4 \ell_{1} + 3 r_{1} + 4 u\bigr)\bigr)\bigr)\bigr)
\\[-1pt]
&\quad{}+ \beta \bigl(\beta t \bigl(\ell_{1} r_{1} + r_{1}^{2} + t \bigl(3 r_{1} + 2 u\bigr) + u \bigl(2 \ell_{1} + 5 r_{1} + 2 u\bigr)\bigr) + t \bigl(\ell_{1} \bigl(\ell_{1} r_{1} + r_{1}^{2}\bigr) + t \bigl(5 \ell_{1} r_{1} + 2 r_{1}^{2} + t \bigl(6 r_{1} + 6 u\bigr)
\\[-1pt]
&\quad{}+ u \bigl(8 \ell_{1} + 18 r_{1} + 10 u\bigr)\bigr) + u \bigl(\ell_{1} \bigl(2 \ell_{1} + 9 r_{1}\bigr) + r_{1}^{2} + u \bigl(6 \ell_{1} + 10 r_{1} + 4 u\bigr)\bigr)\bigr)\bigr) + t \bigl(t \bigl(t \bigl(4 t u + u \bigl(6 \ell_{1} + 3 r_{1} + 10 u\bigr)\bigr)
\\[-1pt]
&\quad{}+ u \bigl(\ell_{1} \bigl(2 \ell_{1} + 4 r_{1}\bigr) + u \bigl(10 \ell_{1} + 3 r_{1} + 8 u\bigr)\bigr)\bigr) + u \bigl(\ell_{1} \bigl(\ell_{1} r_{1} + r_{1}^{2}\bigr) + u \bigl(\ell_{1} \bigl(2 \ell_{1} + 3 r_{1}\bigr) + u \bigl(4 \ell_{1} + 2 u\bigr)\bigr)\bigr)\bigr)
\end{aligned}
\]

\noindent\textit{Case $(\varepsilon,j)=(R,2)$.}

\[
\begin{aligned}
&\mathcal{E}^{(3)}_{R,2}=\alpha \bigl(\alpha \bigl(\beta t \bigl(2 \ell_{2} + 2 u\bigr) + t \bigl(2 t u + u \bigl(2 \ell_{2} + 2 u\bigr)\bigr)\bigr) + \beta \bigl(\beta t \bigl(2 \ell_{2} + 2 u\bigr) + t \bigl(4 \ell_{1} \ell_{2} + \ell_{2}^{2} + t \bigl(7 \ell_{2} + 8 u\bigr)
\\[-1pt]
&\quad{}+ u \bigl(4 \ell_{1} + 9 \ell_{2} + 6 u\bigr)\bigr)\bigr) + t \bigl(t \bigl(6 t u + u \bigl(4 \ell_{1} + 5 \ell_{2} + 10 u\bigr)\bigr) + u \bigl(4 \ell_{1} \ell_{2} + \ell_{2}^{2} + u \bigl(4 \ell_{1} + 3 \ell_{2} + 4 u\bigr)\bigr)\bigr)\bigr)
\\[-1pt]
&\quad{}+ \beta \bigl(\beta t \bigl(2 \ell_{1} \ell_{2} + \ell_{2}^{2} + t \bigl(3 \ell_{2} + 2 u\bigr) + u \bigl(2 \ell_{1} + 4 \ell_{2} + 2 u\bigr)\bigr) + t \bigl(\ell_{1} \bigl(2 \ell_{1} \ell_{2} + \ell_{2}^{2}\bigr) + t \bigl(7 \ell_{1} \ell_{2} + 2 \ell_{2}^{2} + t \bigl(6 \ell_{2}
\\[-1pt]
&\quad{}+ 6 u\bigr) + u \bigl(8 \ell_{1} + 15 \ell_{2} + 10 u\bigr)\bigr) + u \bigl(\ell_{1} \bigl(2 \ell_{1} + 9 \ell_{2}\bigr) + \ell_{2}^{2} + u \bigl(6 \ell_{1} + 8 \ell_{2} + 4 u\bigr)\bigr)\bigr)\bigr) + t \bigl(t \bigl(t \bigl(4 t u + u \bigl(6 \ell_{1} + 3 \ell_{2}
\\[-1pt]
&\quad{}+ 10 u\bigr)\bigr) + u \bigl(\ell_{1} \bigl(2 \ell_{1} + 5 \ell_{2}\bigr) + u \bigl(10 \ell_{1} + 3 \ell_{2} + 8 u\bigr)\bigr)\bigr) + u \bigl(\ell_{1} \bigl(2 \ell_{1} \ell_{2} + \ell_{2}^{2}\bigr) + u \bigl(\ell_{1} \bigl(2 \ell_{1} + 3 \ell_{2}\bigr) + u \bigl(4 \ell_{1}
\\[-1pt]
&\quad{}+ 2 u\bigr)\bigr)\bigr)\bigr)
\end{aligned}
\]

\paragraph{Four-witness clipping identities.}
For $\varepsilon_1,\varepsilon_2\in\{L,R\}$ and
$j\in\{0,1,2\}$, let
$\mathcal E^{\mathrm I}_{\varepsilon_1,\varepsilon_2,j}$ and
$\mathcal E^{\mathrm M}_{\varepsilon_1,\varepsilon_2,j}$ denote the
right-hand side of the clipping identity in Regimes~$\mathrm I$ and
$\mathrm M$, respectively. In Regime~$\mathrm I$,
$e_1=t+u$ and $e_2=t+u+v$. In Regime~$\mathrm M$,
$e_1=e$, $t=e+u$, and $e_2=e+u+v$.

\noindent\textit{Geometry $(\varepsilon_1,\varepsilon_2,j)=(L,L,0)$.}

\noindent Regime $\mathrm I$:

\[
\begin{aligned}
&\mathcal{E}^{\mathrm{I}}_{L,L,0}=\alpha \bigl(\alpha \bigl(\beta t \bigl(4 r_{1} + 2 r_{2} + 4 u + 4 v\bigr) + t \bigl(r_{1} \bigl(4 r_{1} + 4 r_{2}\bigr) + t \bigl(4 r_{1} + 4 u + 4 v\bigr) + u \bigl(8 r_{1} + 2 r_{2} + 4 u + 8 v\bigr)
\\[-1pt]
&\quad{}+ v \bigl(8 r_{1} + 3 r_{2} + 4 v\bigr)\bigr)\bigr) + \beta \bigl(\beta t \bigl(4 r_{1} + 2 r_{2} + 4 u + 4 v\bigr) + t \bigl(r_{1} \bigl(4 r_{1} + 4 r_{2}\bigr) + t \bigl(16 r_{1} + 6 r_{2} + 16 u + 18 v\bigr)
\\[-1pt]
&\quad{}+ u \bigl(16 r_{1} + 6 r_{2} + 12 u + 22 v\bigr) + v \bigl(12 r_{1} + 6 r_{2} + 8 v\bigr)\bigr)\bigr) + t \bigl(t \bigl(r_{1} \bigl(12 r_{1} + 10 r_{2}\bigr) + t \bigl(12 r_{1} + 12 u + 12 v\bigr)
\\[-1pt]
&\quad{}+ u \bigl(32 r_{1} + 4 r_{2} + 20 u + 34 v\bigr) + v \bigl(26 r_{1} + 7 r_{2} + 14 v\bigr)\bigr) + u \bigl(r_{1} \bigl(8 r_{1} + 6 r_{2}\bigr) + u \bigl(16 r_{1} + 2 r_{2} + 8 u + 18 v\bigr)
\\[-1pt]
&\quad{}+ v \bigl(22 r_{1} + 6 r_{2} + 14 v\bigr)\bigr) + v \bigl(r_{1} \bigl(4 r_{1} + 5 r_{2}\bigr) + v \bigl(8 r_{1} + 4 r_{2} + 4 v\bigr)\bigr)\bigr)\bigr) + \beta \bigl(\beta t \bigl(t \bigl(4 r_{1} + 2 r_{2} + 4 u + 6 v\bigr)
\\[-1pt]
&\quad{}+ u \bigl(4 r_{1} + 2 r_{2} + 4 u + 8 v\bigr) + v \bigl(2 r_{1} + 2 r_{2} + 2 v\bigr)\bigr) + t \bigl(t \bigl(r_{1} \bigl(4 r_{1} + 4 r_{2}\bigr) + t \bigl(12 r_{1} + 4 r_{2} + 12 u + 16 v\bigr)
\\[-1pt]
&\quad{}+ u \bigl(24 r_{1} + 8 r_{2} + 20 u + 38 v\bigr) + v \bigl(16 r_{1} + 9 r_{2} + 12 v\bigr)\bigr) + u \bigl(r_{1} \bigl(4 r_{1} + 4 r_{2}\bigr) + u \bigl(12 r_{1} + 4 r_{2} + 8 u + 20 v\bigr)
\\[-1pt]
&\quad{}+ v \bigl(16 r_{1} + 8 r_{2} + 12 v\bigr)\bigr) + v \bigl(r_{1} \bigl(2 r_{1} + 3 r_{2}\bigr) + v \bigl(4 r_{1} + 3 r_{2} + 2 v\bigr)\bigr)\bigr)\bigr) + t \bigl(t \bigl(t \bigl(r_{1} \bigl(8 r_{1} + 6 r_{2}\bigr) + t \bigl(8 r_{1} + 8 u
\\[-1pt]
&\quad{}+ 8 v\bigr) + u \bigl(28 r_{1} + 2 r_{2} + 20 u + 30 v\bigr) + v \bigl(18 r_{1} + 4 r_{2} + 10 v\bigr)\bigr) + u \bigl(r_{1} \bigl(12 r_{1} + 8 r_{2}\bigr) + u \bigl(28 r_{1} + 2 r_{2} + 16 u + 30 v\bigr)
\\[-1pt]
&\quad{}+ v \bigl(30 r_{1} + 5 r_{2} + 18 v\bigr)\bigr) + v \bigl(r_{1} \bigl(4 r_{1} + 5 r_{2}\bigr) + v \bigl(8 r_{1} + 4 r_{2} + 4 v\bigr)\bigr)\bigr) + u \bigl(u \bigl(r_{1} \bigl(4 r_{1} + 2 r_{2}\bigr) + u \bigl(8 r_{1} + 4 u + 8 v\bigr)
\\[-1pt]
&\quad{}+ v \bigl(10 r_{1} + 6 v\bigr)\bigr) + v \bigl(r_{1} \bigl(2 r_{1} + 2 r_{2}\bigr) + v \bigl(4 r_{1} + r_{2} + 2 v\bigr)\bigr)\bigr)\bigr)
\end{aligned}
\]

\noindent Regime $\mathrm M$:

\[
\begin{aligned}
&\mathcal{E}^{\mathrm{M}}_{L,L,0}=\alpha \bigl(\alpha \bigl(\beta \bigl(e \bigl(4 r_{1} + 2 r_{2} + 4 v\bigr) + u \bigl(2 r_{1} + r_{2} + 2 v\bigr)\bigr) + e \bigl(e \bigl(4 r_{1} + 4 v\bigr) + r_{1} \bigl(4 r_{1} + 4 r_{2}\bigr) + u \bigl(6 r_{1} + 6 v\bigr)
\\[-1pt]
&\quad{}+ v \bigl(8 r_{1} + 3 r_{2} + 4 v\bigr)\bigr) + u \bigl(r_{1} \bigl(2 r_{1} + 2 r_{2}\bigr) + u \bigl(2 r_{1} + 2 v\bigr) + v \bigl(4 r_{1} + r_{2} + 2 v\bigr)\bigr)\bigr) + \beta \bigl(\beta \bigl(e \bigl(4 r_{1} + 2 r_{2} + 4 v\bigr)
\\[-1pt]
&\quad{}+ u \bigl(2 r_{1} + r_{2} + 2 v\bigr)\bigr) + e \bigl(e \bigl(16 r_{1} + 6 r_{2} + 2 u + 18 v\bigr) + r_{1} \bigl(4 r_{1} + 4 r_{2}\bigr) + u \bigl(20 r_{1} + 8 r_{2} + 2 u + 22 v\bigr) + v \bigl(12 r_{1}
\\[-1pt]
&\quad{}+ 6 r_{2} + 8 v\bigr)\bigr) + u \bigl(r_{1} \bigl(2 r_{1} + 2 r_{2}\bigr) + u \bigl(6 r_{1} + 2 r_{2} + 6 v\bigr) + v \bigl(6 r_{1} + 2 r_{2} + 4 v\bigr)\bigr)\bigr) + e \bigl(e \bigl(e \bigl(12 r_{1} + 12 v\bigr)
\\[-1pt]
&\quad{}+ r_{1} \bigl(12 r_{1} + 10 r_{2}\bigr) + u \bigl(24 r_{1} + 24 v\bigr) + v \bigl(26 r_{1} + 7 r_{2} + 14 v\bigr)\bigr) + u \bigl(r_{1} \bigl(14 r_{1} + 13 r_{2}\bigr) + u \bigl(16 r_{1} + 16 v\bigr) + v \bigl(32 r_{1}
\\[-1pt]
&\quad{}+ 8 r_{2} + 18 v\bigr)\bigr) + v \bigl(r_{1} \bigl(4 r_{1} + 5 r_{2}\bigr) + v \bigl(8 r_{1} + 4 r_{2} + 4 v\bigr)\bigr)\bigr) + u \bigl(u \bigl(r_{1} \bigl(4 r_{1} + 4 r_{2}\bigr) + u \bigl(4 r_{1} + 4 v\bigr) + v \bigl(10 r_{1}
\\[-1pt]
&\quad{}+ 2 r_{2} + 6 v\bigr)\bigr) + v \bigl(r_{1} \bigl(2 r_{1} + 2 r_{2}\bigr) + v \bigl(4 r_{1} + r_{2} + 2 v\bigr)\bigr)\bigr)\bigr) + \beta \bigl(\beta e \bigl(e \bigl(4 r_{1} + 2 r_{2} + 2 u + 6 v\bigr) + u \bigl(4 r_{1} + 3 r_{2}
\\[-1pt]
&\quad{}+ 2 u + 6 v\bigr) + v \bigl(2 r_{1} + 2 r_{2} + 2 v\bigr)\bigr) + e \bigl(e \bigl(e \bigl(12 r_{1} + 4 r_{2} + 4 u + 16 v\bigr) + r_{1} \bigl(4 r_{1} + 4 r_{2}\bigr) + u \bigl(22 r_{1} + 10 r_{2} + 6 u
\\[-1pt]
&\quad{}+ 28 v\bigr) + v \bigl(16 r_{1} + 9 r_{2} + 12 v\bigr)\bigr) + u \bigl(r_{1} \bigl(4 r_{1} + 5 r_{2}\bigr) + u \bigl(10 r_{1} + 5 r_{2} + 2 u + 12 v\bigr) + v \bigl(14 r_{1} + 8 r_{2} + 10 v\bigr)\bigr)
\\[-1pt]
&\quad{}+ v \bigl(r_{1} \bigl(2 r_{1} + 3 r_{2}\bigr) + v \bigl(4 r_{1} + 3 r_{2} + 2 v\bigr)\bigr)\bigr)\bigr) + e \bigl(e \bigl(e \bigl(e \bigl(8 r_{1} + 8 v\bigr) + r_{1} \bigl(8 r_{1} + 6 r_{2}\bigr) + u \bigl(18 r_{1} + 18 v\bigr)
\\[-1pt]
&\quad{}+ v \bigl(18 r_{1} + 4 r_{2} + 10 v\bigr)\bigr) + u \bigl(r_{1} \bigl(12 r_{1} + 11 r_{2}\bigr) + u \bigl(14 r_{1} + 14 v\bigr) + v \bigl(28 r_{1} + 7 r_{2} + 16 v\bigr)\bigr) + v \bigl(r_{1} \bigl(4 r_{1} + 5 r_{2}\bigr)
\\[-1pt]
&\quad{}+ v \bigl(8 r_{1} + 4 r_{2} + 4 v\bigr)\bigr)\bigr) + u \bigl(u \bigl(r_{1} \bigl(4 r_{1} + 4 r_{2}\bigr) + u \bigl(4 r_{1} + 4 v\bigr) + v \bigl(10 r_{1} + 2 r_{2} + 6 v\bigr)\bigr) + v \bigl(r_{1} \bigl(2 r_{1} + 2 r_{2}\bigr)
\\[-1pt]
&\quad{}+ v \bigl(4 r_{1} + r_{2} + 2 v\bigr)\bigr)\bigr)\bigr)
\end{aligned}
\]

\noindent\textit{Geometry $(\varepsilon_1,\varepsilon_2,j)=(L,L,1)$.}

\noindent Regime $\mathrm I$:

\[
\begin{aligned}
&\mathcal{E}^{\mathrm{I}}_{L,L,1}=\alpha \bigl(\alpha \bigl(\beta t \bigl(2 \ell_{1} + 2 v\bigr) + t \bigl(4 \ell_{1} r_{1} + u \bigl(2 \ell_{1} + 2 v\bigr) + v \bigl(3 \ell_{1} + 4 r_{1} + 4 v\bigr)\bigr)\bigr) + \beta \bigl(\beta t \bigl(2 \ell_{1} + 2 v\bigr) + t \bigl(4 \ell_{1} r_{1}
\\[-1pt]
&\quad{}+ t \bigl(6 \ell_{1} + 6 v\bigr) + u \bigl(6 \ell_{1} + 6 v\bigr) + v \bigl(6 \ell_{1} + 4 r_{1} + 6 v\bigr)\bigr)\bigr) + t \bigl(t \bigl(10 \ell_{1} r_{1} + u \bigl(4 \ell_{1} + 4 v\bigr) + v \bigl(7 \ell_{1} + 10 r_{1} + 10 v\bigr)\bigr)
\\[-1pt]
&\quad{}+ u \bigl(6 \ell_{1} r_{1} + u \bigl(2 \ell_{1} + 2 v\bigr) + v \bigl(6 \ell_{1} + 6 r_{1} + 8 v\bigr)\bigr) + v \bigl(5 \ell_{1} r_{1} + v \bigl(4 \ell_{1} + 4 r_{1} + 4 v\bigr)\bigr)\bigr)\bigr) + \beta \bigl(\beta t \bigl(t \bigl(2 \ell_{1} + 2 v\bigr)
\\[-1pt]
&\quad{}+ u \bigl(2 \ell_{1} + 2 v\bigr) + v \bigl(2 \ell_{1} + 2 v\bigr)\bigr) + t \bigl(t \bigl(4 \ell_{1} r_{1} + t \bigl(4 \ell_{1} + 4 v\bigr) + u \bigl(8 \ell_{1} + 8 v\bigr) + v \bigl(9 \ell_{1} + 4 r_{1} + 8 v\bigr)\bigr) + u \bigl(4 \ell_{1} r_{1}
\\[-1pt]
&\quad{}+ u \bigl(4 \ell_{1} + 4 v\bigr) + v \bigl(8 \ell_{1} + 4 r_{1} + 6 v\bigr)\bigr) + v \bigl(3 \ell_{1} r_{1} + v \bigl(3 \ell_{1} + 2 r_{1} + 2 v\bigr)\bigr)\bigr)\bigr) + t \bigl(t \bigl(t \bigl(6 \ell_{1} r_{1} + u \bigl(2 \ell_{1} + 2 v\bigr)
\\[-1pt]
&\quad{}+ v \bigl(4 \ell_{1} + 6 r_{1} + 6 v\bigr)\bigr) + u \bigl(8 \ell_{1} r_{1} + u \bigl(2 \ell_{1} + 2 v\bigr) + v \bigl(5 \ell_{1} + 8 r_{1} + 8 v\bigr)\bigr) + v \bigl(5 \ell_{1} r_{1} + v \bigl(4 \ell_{1} + 4 r_{1} + 4 v\bigr)\bigr)\bigr)
\\[-1pt]
&\quad{}+ u \bigl(u \bigl(2 \ell_{1} r_{1} + v \bigl(2 r_{1} + 2 v\bigr)\bigr) + v \bigl(2 \ell_{1} r_{1} + v \bigl(\ell_{1} + 2 r_{1} + 2 v\bigr)\bigr)\bigr)\bigr)
\end{aligned}
\]

\noindent Regime $\mathrm M$:

\[
\begin{aligned}
&\mathcal{E}^{\mathrm{M}}_{L,L,1}=\alpha \bigl(\alpha \bigl(\beta \bigl(e \bigl(2 \ell_{1} + 2 u + 2 v\bigr) + \ell_{1} u\bigr) + e \bigl(4 \ell_{1} r_{1} + u \bigl(4 r_{1} + 4 v\bigr) + v \bigl(3 \ell_{1} + 4 r_{1} + 4 v\bigr)\bigr) + u \bigl(2 \ell_{1} r_{1}
\\[-1pt]
&\quad{}+ u \bigl(2 r_{1} + 2 v\bigr) + v \bigl(\ell_{1} + 2 r_{1} + 2 v\bigr)\bigr)\bigr) + \beta \bigl(\beta \bigl(e \bigl(2 \ell_{1} + 2 u + 2 v\bigr) + \ell_{1} u\bigr) + e \bigl(e \bigl(6 \ell_{1} + 6 u + 6 v\bigr) + 4 \ell_{1} r_{1}
\\[-1pt]
&\quad{}+ u \bigl(8 \ell_{1} + 4 r_{1} + 4 u + 10 v\bigr) + v \bigl(6 \ell_{1} + 4 r_{1} + 6 v\bigr)\bigr) + u \bigl(2 \ell_{1} r_{1} + u \bigl(2 \ell_{1} + 2 r_{1} + 2 v\bigr) + v \bigl(2 \ell_{1} + 2 r_{1} + 2 v\bigr)\bigr)\bigr)
\\[-1pt]
&\quad{}+ e \bigl(e \bigl(10 \ell_{1} r_{1} + u \bigl(10 r_{1} + 10 v\bigr) + v \bigl(7 \ell_{1} + 10 r_{1} + 10 v\bigr)\bigr) + u \bigl(13 \ell_{1} r_{1} + u \bigl(12 r_{1} + 12 v\bigr) + v \bigl(8 \ell_{1} + 16 r_{1} + 16 v\bigr)\bigr)
\\[-1pt]
&\quad{}+ v \bigl(5 \ell_{1} r_{1} + v \bigl(4 \ell_{1} + 4 r_{1} + 4 v\bigr)\bigr)\bigr) + u \bigl(u \bigl(4 \ell_{1} r_{1} + u \bigl(4 r_{1} + 4 v\bigr) + v \bigl(2 \ell_{1} + 6 r_{1} + 6 v\bigr)\bigr) + v \bigl(2 \ell_{1} r_{1} + v \bigl(\ell_{1}
\\[-1pt]
&\quad{}+ 2 r_{1} + 2 v\bigr)\bigr)\bigr)\bigr) + \beta \bigl(\beta e \bigl(e \bigl(2 \ell_{1} + 2 u + 2 v\bigr) + u \bigl(3 \ell_{1} + 2 u + 4 v\bigr) + v \bigl(2 \ell_{1} + 2 v\bigr)\bigr) + e \bigl(e \bigl(e \bigl(4 \ell_{1} + 4 u + 4 v\bigr)
\\[-1pt]
&\quad{}+ 4 \ell_{1} r_{1} + u \bigl(10 \ell_{1} + 4 r_{1} + 6 u + 14 v\bigr) + v \bigl(9 \ell_{1} + 4 r_{1} + 8 v\bigr)\bigr) + u \bigl(5 \ell_{1} r_{1} + u \bigl(5 \ell_{1} + 4 r_{1} + 2 u + 8 v\bigr) + v \bigl(8 \ell_{1}
\\[-1pt]
&\quad{}+ 6 r_{1} + 8 v\bigr)\bigr) + v \bigl(3 \ell_{1} r_{1} + v \bigl(3 \ell_{1} + 2 r_{1} + 2 v\bigr)\bigr)\bigr)\bigr) + e \bigl(e \bigl(e \bigl(6 \ell_{1} r_{1} + u \bigl(6 r_{1} + 6 v\bigr) + v \bigl(4 \ell_{1} + 6 r_{1} + 6 v\bigr)\bigr)
\\[-1pt]
&\quad{}+ u \bigl(11 \ell_{1} r_{1} + u \bigl(10 r_{1} + 10 v\bigr) + v \bigl(7 \ell_{1} + 14 r_{1} + 14 v\bigr)\bigr) + v \bigl(5 \ell_{1} r_{1} + v \bigl(4 \ell_{1} + 4 r_{1} + 4 v\bigr)\bigr)\bigr) + u \bigl(u \bigl(4 \ell_{1} r_{1}
\\[-1pt]
&\quad{}+ u \bigl(4 r_{1} + 4 v\bigr) + v \bigl(2 \ell_{1} + 6 r_{1} + 6 v\bigr)\bigr) + v \bigl(2 \ell_{1} r_{1} + v \bigl(\ell_{1} + 2 r_{1} + 2 v\bigr)\bigr)\bigr)\bigr)
\end{aligned}
\]

\noindent\textit{Geometry $(\varepsilon_1,\varepsilon_2,j)=(L,L,2)$.}

\noindent Regime $\mathrm I$:

\[
\begin{aligned}
&\mathcal{E}^{\mathrm{I}}_{L,L,2}=\alpha \bigl(\alpha t \bigl(\ell_{1} \bigl(4 \ell_{1} + 4 \ell_{2}\bigr) + v \bigl(4 \ell_{1} + 4 \ell_{2}\bigr)\bigr) + \beta t \bigl(\ell_{1} \bigl(4 \ell_{1} + 4 \ell_{2}\bigr) + v \bigl(4 \ell_{1} + 4 \ell_{2}\bigr)\bigr) + t \bigl(t \bigl(\ell_{1} \bigl(12 \ell_{1} + 10 \ell_{2}\bigr)
\\[-1pt]
&\quad{}+ v \bigl(10 \ell_{1} + 10 \ell_{2}\bigr)\bigr) + u \bigl(\ell_{1} \bigl(8 \ell_{1} + 6 \ell_{2}\bigr) + v \bigl(6 \ell_{1} + 6 \ell_{2}\bigr)\bigr) + v \bigl(\ell_{1} \bigl(4 \ell_{1} + 5 \ell_{2}\bigr) + v \bigl(4 \ell_{1} + 4 \ell_{2}\bigr)\bigr)\bigr)\bigr)
\\[-1pt]
&\quad{}+ \beta t \bigl(t \bigl(\ell_{1} \bigl(4 \ell_{1} + 4 \ell_{2}\bigr) + v \bigl(4 \ell_{1} + 4 \ell_{2}\bigr)\bigr) + u \bigl(\ell_{1} \bigl(4 \ell_{1} + 4 \ell_{2}\bigr) + v \bigl(4 \ell_{1} + 4 \ell_{2}\bigr)\bigr) + v \bigl(\ell_{1} \bigl(2 \ell_{1} + 3 \ell_{2}\bigr) + v \bigl(2 \ell_{1} + 2 \ell_{2}\bigr)\bigr)\bigr)
\\[-1pt]
&\quad{}+ t \bigl(t \bigl(t \bigl(\ell_{1} \bigl(8 \ell_{1} + 6 \ell_{2}\bigr) + v \bigl(6 \ell_{1} + 6 \ell_{2}\bigr)\bigr) + u \bigl(\ell_{1} \bigl(12 \ell_{1} + 8 \ell_{2}\bigr) + v \bigl(8 \ell_{1} + 8 \ell_{2}\bigr)\bigr) + v \bigl(\ell_{1} \bigl(4 \ell_{1} + 5 \ell_{2}\bigr) + v \bigl(4 \ell_{1} + 4 \ell_{2}\bigr)\bigr)\bigr)
\\[-1pt]
&\quad{}+ u \bigl(u \bigl(\ell_{1} \bigl(4 \ell_{1} + 2 \ell_{2}\bigr) + v \bigl(2 \ell_{1} + 2 \ell_{2}\bigr)\bigr) + v \bigl(\ell_{1} \bigl(2 \ell_{1} + 2 \ell_{2}\bigr) + v \bigl(2 \ell_{1} + 2 \ell_{2}\bigr)\bigr)\bigr)\bigr)
\end{aligned}
\]

\noindent Regime $\mathrm M$:

\[
\begin{aligned}
&\mathcal{E}^{\mathrm{M}}_{L,L,2}=\alpha \bigl(\alpha \bigl(e \bigl(\ell_{1} \bigl(4 \ell_{1} + 4 \ell_{2}\bigr) + u \bigl(4 \ell_{1} + 4 \ell_{2}\bigr) + v \bigl(4 \ell_{1} + 4 \ell_{2}\bigr)\bigr) + u \bigl(\ell_{1} \bigl(2 \ell_{1} + 2 \ell_{2}\bigr) + u \bigl(2 \ell_{1} + 2 \ell_{2}\bigr) + v \bigl(2 \ell_{1} + 2 \ell_{2}\bigr)\bigr)\bigr)
\\[-1pt]
&\quad{}+ \beta \bigl(e \bigl(\ell_{1} \bigl(4 \ell_{1} + 4 \ell_{2}\bigr) + u \bigl(4 \ell_{1} + 4 \ell_{2}\bigr) + v \bigl(4 \ell_{1} + 4 \ell_{2}\bigr)\bigr) + u \bigl(\ell_{1} \bigl(2 \ell_{1} + 2 \ell_{2}\bigr) + u \bigl(2 \ell_{1} + 2 \ell_{2}\bigr) + v \bigl(2 \ell_{1} + 2 \ell_{2}\bigr)\bigr)\bigr)
\\[-1pt]
&\quad{}+ e \bigl(e \bigl(\ell_{1} \bigl(12 \ell_{1} + 10 \ell_{2}\bigr) + u \bigl(10 \ell_{1} + 10 \ell_{2}\bigr) + v \bigl(10 \ell_{1} + 10 \ell_{2}\bigr)\bigr) + u \bigl(\ell_{1} \bigl(14 \ell_{1} + 13 \ell_{2}\bigr) + u \bigl(12 \ell_{1} + 12 \ell_{2}\bigr) + v \bigl(16 \ell_{1}
\\[-1pt]
&\quad{}+ 16 \ell_{2}\bigr)\bigr) + v \bigl(\ell_{1} \bigl(4 \ell_{1} + 5 \ell_{2}\bigr) + v \bigl(4 \ell_{1} + 4 \ell_{2}\bigr)\bigr)\bigr) + u \bigl(u \bigl(\ell_{1} \bigl(4 \ell_{1} + 4 \ell_{2}\bigr) + u \bigl(4 \ell_{1} + 4 \ell_{2}\bigr) + v \bigl(6 \ell_{1} + 6 \ell_{2}\bigr)\bigr)
\\[-1pt]
&\quad{}+ v \bigl(\ell_{1} \bigl(2 \ell_{1} + 2 \ell_{2}\bigr) + v \bigl(2 \ell_{1} + 2 \ell_{2}\bigr)\bigr)\bigr)\bigr) + \beta e \bigl(e \bigl(\ell_{1} \bigl(4 \ell_{1} + 4 \ell_{2}\bigr) + u \bigl(4 \ell_{1} + 4 \ell_{2}\bigr) + v \bigl(4 \ell_{1} + 4 \ell_{2}\bigr)\bigr) + u \bigl(\ell_{1} \bigl(4 \ell_{1}
\\[-1pt]
&\quad{}+ 5 \ell_{2}\bigr) + u \bigl(4 \ell_{1} + 4 \ell_{2}\bigr) + v \bigl(6 \ell_{1} + 6 \ell_{2}\bigr)\bigr) + v \bigl(\ell_{1} \bigl(2 \ell_{1} + 3 \ell_{2}\bigr) + v \bigl(2 \ell_{1} + 2 \ell_{2}\bigr)\bigr)\bigr) + e \bigl(e \bigl(e \bigl(\ell_{1} \bigl(8 \ell_{1} + 6 \ell_{2}\bigr)
\\[-1pt]
&\quad{}+ u \bigl(6 \ell_{1} + 6 \ell_{2}\bigr) + v \bigl(6 \ell_{1} + 6 \ell_{2}\bigr)\bigr) + u \bigl(\ell_{1} \bigl(12 \ell_{1} + 11 \ell_{2}\bigr) + u \bigl(10 \ell_{1} + 10 \ell_{2}\bigr) + v \bigl(14 \ell_{1} + 14 \ell_{2}\bigr)\bigr) + v \bigl(\ell_{1} \bigl(4 \ell_{1} + 5 \ell_{2}\bigr)
\\[-1pt]
&\quad{}+ v \bigl(4 \ell_{1} + 4 \ell_{2}\bigr)\bigr)\bigr) + u \bigl(u \bigl(\ell_{1} \bigl(4 \ell_{1} + 4 \ell_{2}\bigr) + u \bigl(4 \ell_{1} + 4 \ell_{2}\bigr) + v \bigl(6 \ell_{1} + 6 \ell_{2}\bigr)\bigr) + v \bigl(\ell_{1} \bigl(2 \ell_{1} + 2 \ell_{2}\bigr) + v \bigl(2 \ell_{1}
\\[-1pt]
&\quad{}+ 2 \ell_{2}\bigr)\bigr)\bigr)\bigr)
\end{aligned}
\]

\noindent\textit{Geometry $(\varepsilon_1,\varepsilon_2,j)=(L,R,0)$.}

\noindent Regime $\mathrm I$:

\[
\begin{aligned}
&\mathcal{E}^{\mathrm{I}}_{L,R,0}=\alpha \bigl(\beta t \bigl(r_{1} \bigl(2 r_{1} + 3 r_{2}\bigr) + t \bigl(2 r_{1} + 2 r_{2}\bigr) + u \bigl(6 r_{1} + 6 r_{2}\bigr) + v \bigl(4 r_{1} + 4 r_{2}\bigr)\bigr) + t \bigl(t \bigl(u \bigl(2 r_{1} + 2 r_{2}\bigr) + v \bigl(2 r_{1}
\\[-1pt]
&\quad{}+ 2 r_{2}\bigr)\bigr) + u \bigl(r_{1} \bigl(2 r_{1} + 3 r_{2}\bigr) + u \bigl(2 r_{1} + 2 r_{2}\bigr) + v \bigl(6 r_{1} + 6 r_{2}\bigr)\bigr) + v \bigl(r_{1} \bigl(4 r_{1} + 5 r_{2}\bigr) + v \bigl(4 r_{1} + 4 r_{2}\bigr)\bigr)\bigr)\bigr)
\\[-1pt]
&\quad{}+ \beta \bigl(\beta t \bigl(r_{1} \bigl(2 r_{1} + 2 r_{2}\bigr) + t \bigl(2 r_{1} + 2 r_{2}\bigr) + u \bigl(4 r_{1} + 4 r_{2}\bigr) + v \bigl(2 r_{1} + 2 r_{2}\bigr)\bigr) + t \bigl(t \bigl(r_{1} \bigl(4 r_{1} + 6 r_{2}\bigr) + t \bigl(4 r_{1} + 4 r_{2}\bigr)
\\[-1pt]
&\quad{}+ u \bigl(14 r_{1} + 14 r_{2}\bigr) + v \bigl(8 r_{1} + 8 r_{2}\bigr)\bigr) + u \bigl(r_{1} \bigl(2 r_{1} + 5 r_{2}\bigr) + u \bigl(8 r_{1} + 8 r_{2}\bigr) + v \bigl(8 r_{1} + 8 r_{2}\bigr)\bigr) + v \bigl(r_{1} \bigl(2 r_{1} + 3 r_{2}\bigr)
\\[-1pt]
&\quad{}+ v \bigl(2 r_{1} + 2 r_{2}\bigr)\bigr)\bigr)\bigr) + t \bigl(t \bigl(t \bigl(u \bigl(2 r_{1} + 2 r_{2}\bigr) + v \bigl(2 r_{1} + 2 r_{2}\bigr)\bigr) + u \bigl(r_{1} r_{2} + u \bigl(2 r_{1} + 2 r_{2}\bigr) + v \bigl(6 r_{1} + 6 r_{2}\bigr)\bigr)
\\[-1pt]
&\quad{}+ v \bigl(r_{1} \bigl(4 r_{1} + 5 r_{2}\bigr) + v \bigl(4 r_{1} + 4 r_{2}\bigr)\bigr)\bigr) + u \bigl(u v \bigl(2 r_{1} + 2 r_{2}\bigr) + v \bigl(r_{1} \bigl(2 r_{1} + 2 r_{2}\bigr) + v \bigl(2 r_{1} + 2 r_{2}\bigr)\bigr)\bigr)\bigr)
\end{aligned}
\]

\noindent Regime $\mathrm M$:

\[
\begin{aligned}
&\mathcal{E}^{\mathrm{M}}_{L,R,0}=\alpha \bigl(\beta \bigl(e \bigl(e \bigl(2 r_{1} + 2 r_{2}\bigr) + r_{1} \bigl(2 r_{1} + 3 r_{2}\bigr) + u \bigl(2 r_{1} + 2 r_{2}\bigr) + v \bigl(4 r_{1} + 4 r_{2}\bigr)\bigr) + u \bigl(r_{1} r_{2} + v \bigl(2 r_{1} + 2 r_{2}\bigr)\bigr)\bigr)
\\[-1pt]
&\quad{}+ e \bigl(e v \bigl(2 r_{1} + 2 r_{2}\bigr) + u v \bigl(4 r_{1} + 4 r_{2}\bigr) + v \bigl(r_{1} \bigl(4 r_{1} + 5 r_{2}\bigr) + v \bigl(4 r_{1} + 4 r_{2}\bigr)\bigr)\bigr) + u \bigl(u v \bigl(2 r_{1} + 2 r_{2}\bigr) + v \bigl(r_{1} \bigl(2 r_{1}
\\[-1pt]
&\quad{}+ 2 r_{2}\bigr) + v \bigl(2 r_{1} + 2 r_{2}\bigr)\bigr)\bigr)\bigr) + \beta \bigl(\beta e \bigl(e \bigl(2 r_{1} + 2 r_{2}\bigr) + r_{1} \bigl(2 r_{1} + 2 r_{2}\bigr) + u \bigl(2 r_{1} + 2 r_{2}\bigr) + v \bigl(2 r_{1} + 2 r_{2}\bigr)\bigr)
\\[-1pt]
&\quad{}+ e \bigl(e \bigl(e \bigl(4 r_{1} + 4 r_{2}\bigr) + r_{1} \bigl(4 r_{1} + 6 r_{2}\bigr) + u \bigl(6 r_{1} + 6 r_{2}\bigr) + v \bigl(8 r_{1} + 8 r_{2}\bigr)\bigr) + u \bigl(r_{1} \bigl(2 r_{1} + 4 r_{2}\bigr) + u \bigl(2 r_{1} + 2 r_{2}\bigr)
\\[-1pt]
&\quad{}+ v \bigl(6 r_{1} + 6 r_{2}\bigr)\bigr) + v \bigl(r_{1} \bigl(2 r_{1} + 3 r_{2}\bigr) + v \bigl(2 r_{1} + 2 r_{2}\bigr)\bigr)\bigr)\bigr) + e \bigl(e \bigl(e v \bigl(2 r_{1} + 2 r_{2}\bigr) + u v \bigl(4 r_{1} + 4 r_{2}\bigr) + v \bigl(r_{1} \bigl(4 r_{1}
\\[-1pt]
&\quad{}+ 5 r_{2}\bigr) + v \bigl(4 r_{1} + 4 r_{2}\bigr)\bigr)\bigr) + u \bigl(u v \bigl(2 r_{1} + 2 r_{2}\bigr) + v \bigl(r_{1} \bigl(2 r_{1} + 2 r_{2}\bigr) + v \bigl(2 r_{1} + 2 r_{2}\bigr)\bigr)\bigr)\bigr)
\end{aligned}
\]

\noindent\textit{Geometry $(\varepsilon_1,\varepsilon_2,j)=(L,R,1)$.}

\noindent Regime $\mathrm I$:

\[
\begin{aligned}
&\mathcal{E}^{\mathrm{I}}_{L,R,1}=\alpha \bigl(\alpha \bigl(\beta t \bigl(2 r_{1} + 4 u + 4 v\bigr) + t \bigl(t \bigl(4 u + 4 v\bigr) + u \bigl(2 r_{1} + 4 u + 8 v\bigr) + v \bigl(3 r_{1} + 4 v\bigr)\bigr)\bigr) + \beta \bigl(\beta t \bigl(2 r_{1} + 4 u
\\[-1pt]
&\quad{}+ 4 v\bigr) + t \bigl(3 \ell_{1} r_{1} + t \bigl(2 \ell_{1} + 6 r_{1} + 16 u + 18 v\bigr) + u \bigl(6 \ell_{1} + 6 r_{1} + 12 u + 22 v\bigr) + v \bigl(4 \ell_{1} + 6 r_{1} + 8 v\bigr)\bigr)\bigr)
\\[-1pt]
&\quad{}+ t \bigl(t \bigl(t \bigl(12 u + 12 v\bigr) + u \bigl(2 \ell_{1} + 4 r_{1} + 20 u + 34 v\bigr) + v \bigl(2 \ell_{1} + 7 r_{1} + 14 v\bigr)\bigr) + u \bigl(3 \ell_{1} r_{1} + u \bigl(2 \ell_{1} + 2 r_{1} + 8 u
\\[-1pt]
&\quad{}+ 18 v\bigr) + v \bigl(6 \ell_{1} + 6 r_{1} + 14 v\bigr)\bigr) + v \bigl(5 \ell_{1} r_{1} + v \bigl(4 \ell_{1} + 4 r_{1} + 4 v\bigr)\bigr)\bigr)\bigr) + \beta \bigl(\beta t \bigl(2 \ell_{1} r_{1} + t \bigl(2 \ell_{1} + 2 r_{1} + 4 u + 6 v\bigr)
\\[-1pt]
&\quad{}+ u \bigl(4 \ell_{1} + 2 r_{1} + 4 u + 8 v\bigr) + v \bigl(2 \ell_{1} + 2 r_{1} + 2 v\bigr)\bigr) + t \bigl(t \bigl(6 \ell_{1} r_{1} + t \bigl(4 \ell_{1} + 4 r_{1} + 12 u + 16 v\bigr) + u \bigl(14 \ell_{1} + 8 r_{1}
\\[-1pt]
&\quad{}+ 20 u + 38 v\bigr) + v \bigl(8 \ell_{1} + 9 r_{1} + 12 v\bigr)\bigr) + u \bigl(5 \ell_{1} r_{1} + u \bigl(8 \ell_{1} + 4 r_{1} + 8 u + 20 v\bigr) + v \bigl(8 \ell_{1} + 8 r_{1} + 12 v\bigr)\bigr) + v \bigl(3 \ell_{1} r_{1}
\\[-1pt]
&\quad{}+ v \bigl(2 \ell_{1} + 3 r_{1} + 2 v\bigr)\bigr)\bigr)\bigr) + t \bigl(t \bigl(t \bigl(t \bigl(8 u + 8 v\bigr) + u \bigl(2 \ell_{1} + 2 r_{1} + 20 u + 30 v\bigr) + v \bigl(2 \ell_{1} + 4 r_{1} + 10 v\bigr)\bigr) + u \bigl(\ell_{1} r_{1}
\\[-1pt]
&\quad{}+ u \bigl(2 \ell_{1} + 2 r_{1} + 16 u + 30 v\bigr) + v \bigl(6 \ell_{1} + 5 r_{1} + 18 v\bigr)\bigr) + v \bigl(5 \ell_{1} r_{1} + v \bigl(4 \ell_{1} + 4 r_{1} + 4 v\bigr)\bigr)\bigr) + u \bigl(u \bigl(u \bigl(4 u + 8 v\bigr)
\\[-1pt]
&\quad{}+ v \bigl(2 \ell_{1} + 6 v\bigr)\bigr) + v \bigl(2 \ell_{1} r_{1} + v \bigl(2 \ell_{1} + r_{1} + 2 v\bigr)\bigr)\bigr)\bigr)
\end{aligned}
\]

\noindent Regime $\mathrm M$:

\[
\begin{aligned}
&\mathcal{E}^{\mathrm{M}}_{L,R,1}=\alpha \bigl(\alpha \bigl(\beta \bigl(e \bigl(2 r_{1} + 4 v\bigr) + u \bigl(r_{1} + 2 v\bigr)\bigr) + e \bigl(4 e v + 6 u v + v \bigl(3 r_{1} + 4 v\bigr)\bigr) + u \bigl(2 u v + v \bigl(r_{1} + 2 v\bigr)\bigr)\bigr)
\\[-1pt]
&\quad{}+ \beta \bigl(\beta \bigl(e \bigl(2 r_{1} + 4 v\bigr) + u \bigl(r_{1} + 2 v\bigr)\bigr) + e \bigl(e \bigl(2 \ell_{1} + 6 r_{1} + 2 u + 18 v\bigr) + 3 \ell_{1} r_{1} + u \bigl(2 \ell_{1} + 8 r_{1} + 2 u + 22 v\bigr)
\\[-1pt]
&\quad{}+ v \bigl(4 \ell_{1} + 6 r_{1} + 8 v\bigr)\bigr) + u \bigl(\ell_{1} r_{1} + u \bigl(2 r_{1} + 6 v\bigr) + v \bigl(2 \ell_{1} + 2 r_{1} + 4 v\bigr)\bigr)\bigr) + e \bigl(e \bigl(12 e v + 24 u v + v \bigl(2 \ell_{1} + 7 r_{1}
\\[-1pt]
&\quad{}+ 14 v\bigr)\bigr) + u \bigl(16 u v + v \bigl(4 \ell_{1} + 8 r_{1} + 18 v\bigr)\bigr) + v \bigl(5 \ell_{1} r_{1} + v \bigl(4 \ell_{1} + 4 r_{1} + 4 v\bigr)\bigr)\bigr) + u \bigl(u \bigl(4 u v + v \bigl(2 \ell_{1} + 2 r_{1} + 6 v\bigr)\bigr)
\\[-1pt]
&\quad{}+ v \bigl(2 \ell_{1} r_{1} + v \bigl(2 \ell_{1} + r_{1} + 2 v\bigr)\bigr)\bigr)\bigr) + \beta \bigl(\beta e \bigl(e \bigl(2 \ell_{1} + 2 r_{1} + 2 u + 6 v\bigr) + 2 \ell_{1} r_{1} + u \bigl(2 \ell_{1} + 3 r_{1} + 2 u + 6 v\bigr)
\\[-1pt]
&\quad{}+ v \bigl(2 \ell_{1} + 2 r_{1} + 2 v\bigr)\bigr) + e \bigl(e \bigl(e \bigl(4 \ell_{1} + 4 r_{1} + 4 u + 16 v\bigr) + 6 \ell_{1} r_{1} + u \bigl(6 \ell_{1} + 10 r_{1} + 6 u + 28 v\bigr) + v \bigl(8 \ell_{1} + 9 r_{1}
\\[-1pt]
&\quad{}+ 12 v\bigr)\bigr) + u \bigl(4 \ell_{1} r_{1} + u \bigl(2 \ell_{1} + 5 r_{1} + 2 u + 12 v\bigr) + v \bigl(6 \ell_{1} + 8 r_{1} + 10 v\bigr)\bigr) + v \bigl(3 \ell_{1} r_{1} + v \bigl(2 \ell_{1} + 3 r_{1} + 2 v\bigr)\bigr)\bigr)\bigr)
\\[-1pt]
&\quad{}+ e \bigl(e \bigl(e \bigl(8 e v + 18 u v + v \bigl(2 \ell_{1} + 4 r_{1} + 10 v\bigr)\bigr) + u \bigl(14 u v + v \bigl(4 \ell_{1} + 7 r_{1} + 16 v\bigr)\bigr) + v \bigl(5 \ell_{1} r_{1} + v \bigl(4 \ell_{1} + 4 r_{1} + 4 v\bigr)\bigr)\bigr)
\\[-1pt]
&\quad{}+ u \bigl(u \bigl(4 u v + v \bigl(2 \ell_{1} + 2 r_{1} + 6 v\bigr)\bigr) + v \bigl(2 \ell_{1} r_{1} + v \bigl(2 \ell_{1} + r_{1} + 2 v\bigr)\bigr)\bigr)\bigr)
\end{aligned}
\]

\noindent\textit{Geometry $(\varepsilon_1,\varepsilon_2,j)=(L,R,2)$.}

\noindent Regime $\mathrm I$:

\[
\begin{aligned}
&\mathcal{E}^{\mathrm{I}}_{L,R,2}=\alpha \bigl(\alpha \bigl(\beta t \bigl(2 \ell_{1} + 2 \ell_{2} + 2 v\bigr) + t \bigl(u \bigl(2 \ell_{1} + 2 \ell_{2} + 2 v\bigr) + v \bigl(4 \ell_{1} + 3 \ell_{2} + 4 v\bigr)\bigr)\bigr) + \beta \bigl(\beta t \bigl(2 \ell_{1} + 2 \ell_{2} + 2 v\bigr)
\\[-1pt]
&\quad{}+ t \bigl(\ell_{1} \bigl(2 \ell_{1} + 3 \ell_{2}\bigr) + t \bigl(6 \ell_{1} + 6 \ell_{2} + 6 v\bigr) + u \bigl(6 \ell_{1} + 6 \ell_{2} + 6 v\bigr) + v \bigl(8 \ell_{1} + 6 \ell_{2} + 6 v\bigr)\bigr)\bigr) + t \bigl(t \bigl(u \bigl(4 \ell_{1} + 4 \ell_{2}
\\[-1pt]
&\quad{}+ 4 v\bigr) + v \bigl(10 \ell_{1} + 7 \ell_{2} + 10 v\bigr)\bigr) + u \bigl(\ell_{1} \bigl(2 \ell_{1} + 3 \ell_{2}\bigr) + u \bigl(2 \ell_{1} + 2 \ell_{2} + 2 v\bigr) + v \bigl(10 \ell_{1} + 6 \ell_{2} + 8 v\bigr)\bigr) + v \bigl(\ell_{1} \bigl(4 \ell_{1}
\\[-1pt]
&\quad{}+ 5 \ell_{2}\bigr) + v \bigl(8 \ell_{1} + 4 \ell_{2} + 4 v\bigr)\bigr)\bigr)\bigr) + \beta \bigl(\beta t \bigl(\ell_{1} \bigl(2 \ell_{1} + 2 \ell_{2}\bigr) + t \bigl(2 \ell_{1} + 2 \ell_{2} + 2 v\bigr) + u \bigl(2 \ell_{1} + 2 \ell_{2} + 2 v\bigr) + v \bigl(4 \ell_{1}
\\[-1pt]
&\quad{}+ 2 \ell_{2} + 2 v\bigr)\bigr) + t \bigl(t \bigl(\ell_{1} \bigl(4 \ell_{1} + 6 \ell_{2}\bigr) + t \bigl(4 \ell_{1} + 4 \ell_{2} + 4 v\bigr) + u \bigl(8 \ell_{1} + 8 \ell_{2} + 8 v\bigr) + v \bigl(12 \ell_{1} + 9 \ell_{2} + 8 v\bigr)\bigr)
\\[-1pt]
&\quad{}+ u \bigl(\ell_{1} \bigl(2 \ell_{1} + 5 \ell_{2}\bigr) + u \bigl(4 \ell_{1} + 4 \ell_{2} + 4 v\bigr) + v \bigl(8 \ell_{1} + 8 \ell_{2} + 6 v\bigr)\bigr) + v \bigl(\ell_{1} \bigl(2 \ell_{1} + 3 \ell_{2}\bigr) + v \bigl(4 \ell_{1} + 3 \ell_{2} + 2 v\bigr)\bigr)\bigr)\bigr)
\\[-1pt]
&\quad{}+ t \bigl(t \bigl(t \bigl(u \bigl(2 \ell_{1} + 2 \ell_{2} + 2 v\bigr) + v \bigl(6 \ell_{1} + 4 \ell_{2} + 6 v\bigr)\bigr) + u \bigl(\ell_{1} \ell_{2} + u \bigl(2 \ell_{1} + 2 \ell_{2} + 2 v\bigr) + v \bigl(8 \ell_{1} + 5 \ell_{2} + 8 v\bigr)\bigr)
\\[-1pt]
&\quad{}+ v \bigl(\ell_{1} \bigl(4 \ell_{1} + 5 \ell_{2}\bigr) + v \bigl(8 \ell_{1} + 4 \ell_{2} + 4 v\bigr)\bigr)\bigr) + u \bigl(u v \bigl(2 \ell_{1} + 2 v\bigr) + v \bigl(\ell_{1} \bigl(2 \ell_{1} + 2 \ell_{2}\bigr) + v \bigl(4 \ell_{1} + \ell_{2} + 2 v\bigr)\bigr)\bigr)\bigr)
\end{aligned}
\]

\noindent Regime $\mathrm M$:

\[
\begin{aligned}
&\mathcal{E}^{\mathrm{M}}_{L,R,2}=\alpha \bigl(\alpha \bigl(\beta \bigl(e \bigl(2 \ell_{1} + 2 \ell_{2} + 2 u + 2 v\bigr) + \ell_{2} u\bigr) + e \bigl(4 u v + v \bigl(4 \ell_{1} + 3 \ell_{2} + 4 v\bigr)\bigr) + u \bigl(2 u v + v \bigl(2 \ell_{1} + \ell_{2} + 2 v\bigr)\bigr)\bigr)
\\[-1pt]
&\quad{}+ \beta \bigl(\beta \bigl(e \bigl(2 \ell_{1} + 2 \ell_{2} + 2 u + 2 v\bigr) + \ell_{2} u\bigr) + e \bigl(e \bigl(6 \ell_{1} + 6 \ell_{2} + 6 u + 6 v\bigr) + \ell_{1} \bigl(2 \ell_{1} + 3 \ell_{2}\bigr) + u \bigl(6 \ell_{1} + 8 \ell_{2} + 4 u
\\[-1pt]
&\quad{}+ 10 v\bigr) + v \bigl(8 \ell_{1} + 6 \ell_{2} + 6 v\bigr)\bigr) + u \bigl(\ell_{1} \ell_{2} + u \bigl(2 \ell_{2} + 2 v\bigr) + v \bigl(2 \ell_{1} + 2 \ell_{2} + 2 v\bigr)\bigr)\bigr) + e \bigl(e \bigl(10 u v + v \bigl(10 \ell_{1} + 7 \ell_{2}
\\[-1pt]
&\quad{}+ 10 v\bigr)\bigr) + u \bigl(12 u v + v \bigl(16 \ell_{1} + 8 \ell_{2} + 16 v\bigr)\bigr) + v \bigl(\ell_{1} \bigl(4 \ell_{1} + 5 \ell_{2}\bigr) + v \bigl(8 \ell_{1} + 4 \ell_{2} + 4 v\bigr)\bigr)\bigr) + u \bigl(u \bigl(4 u v + v \bigl(6 \ell_{1}
\\[-1pt]
&\quad{}+ 2 \ell_{2} + 6 v\bigr)\bigr) + v \bigl(\ell_{1} \bigl(2 \ell_{1} + 2 \ell_{2}\bigr) + v \bigl(4 \ell_{1} + \ell_{2} + 2 v\bigr)\bigr)\bigr)\bigr) + \beta \bigl(\beta e \bigl(e \bigl(2 \ell_{1} + 2 \ell_{2} + 2 u + 2 v\bigr) + \ell_{1} \bigl(2 \ell_{1} + 2 \ell_{2}\bigr)
\\[-1pt]
&\quad{}+ u \bigl(4 \ell_{1} + 3 \ell_{2} + 2 u + 4 v\bigr) + v \bigl(4 \ell_{1} + 2 \ell_{2} + 2 v\bigr)\bigr) + e \bigl(e \bigl(e \bigl(4 \ell_{1} + 4 \ell_{2} + 4 u + 4 v\bigr) + \ell_{1} \bigl(4 \ell_{1} + 6 \ell_{2}\bigr) + u \bigl(10 \ell_{1}
\\[-1pt]
&\quad{}+ 10 \ell_{2} + 6 u + 14 v\bigr) + v \bigl(12 \ell_{1} + 9 \ell_{2} + 8 v\bigr)\bigr) + u \bigl(\ell_{1} \bigl(2 \ell_{1} + 4 \ell_{2}\bigr) + u \bigl(4 \ell_{1} + 5 \ell_{2} + 2 u + 8 v\bigr) + v \bigl(10 \ell_{1} + 8 \ell_{2} + 8 v\bigr)\bigr)
\\[-1pt]
&\quad{}+ v \bigl(\ell_{1} \bigl(2 \ell_{1} + 3 \ell_{2}\bigr) + v \bigl(4 \ell_{1} + 3 \ell_{2} + 2 v\bigr)\bigr)\bigr)\bigr) + e \bigl(e \bigl(e \bigl(6 u v + v \bigl(6 \ell_{1} + 4 \ell_{2} + 6 v\bigr)\bigr) + u \bigl(10 u v + v \bigl(14 \ell_{1} + 7 \ell_{2}
\\[-1pt]
&\quad{}+ 14 v\bigr)\bigr) + v \bigl(\ell_{1} \bigl(4 \ell_{1} + 5 \ell_{2}\bigr) + v \bigl(8 \ell_{1} + 4 \ell_{2} + 4 v\bigr)\bigr)\bigr) + u \bigl(u \bigl(4 u v + v \bigl(6 \ell_{1} + 2 \ell_{2} + 6 v\bigr)\bigr) + v \bigl(\ell_{1} \bigl(2 \ell_{1} + 2 \ell_{2}\bigr)
\\[-1pt]
&\quad{}+ v \bigl(4 \ell_{1} + \ell_{2} + 2 v\bigr)\bigr)\bigr)\bigr)
\end{aligned}
\]

\noindent\textit{Geometry $(\varepsilon_1,\varepsilon_2,j)=(R,L,0)$.}

\noindent Regime $\mathrm I$:

\[
\begin{aligned}
&\mathcal{E}^{\mathrm{I}}_{R,L,0}=\alpha \bigl(\beta t \bigl(r_{1} \bigl(2 r_{1} + 3 r_{2}\bigr) + t \bigl(2 r_{1} + 2 r_{2} + 2 v\bigr) + u \bigl(6 r_{1} + 6 r_{2} + 6 v\bigr) + v \bigl(4 r_{1} + 4 r_{2} + 2 v\bigr)\bigr) + t \bigl(t u \bigl(2 r_{1}
\\[-1pt]
&\quad{}+ 2 r_{2} + 2 v\bigr) + u \bigl(r_{1} \bigl(2 r_{1} + 3 r_{2}\bigr) + u \bigl(2 r_{1} + 2 r_{2} + 2 v\bigr) + v \bigl(4 r_{1} + 2 r_{2} + 2 v\bigr)\bigr)\bigr)\bigr) + \beta \bigl(\beta t \bigl(r_{1} \bigl(2 r_{1} + 2 r_{2}\bigr)
\\[-1pt]
&\quad{}+ t \bigl(2 r_{1} + 2 r_{2} + 2 v\bigr) + u \bigl(4 r_{1} + 4 r_{2} + 4 v\bigr) + v \bigl(4 r_{1} + 3 r_{2} + 2 v\bigr)\bigr) + t \bigl(t \bigl(r_{1} \bigl(4 r_{1} + 6 r_{2}\bigr) + t \bigl(4 r_{1} + 4 r_{2} + 4 v\bigr)
\\[-1pt]
&\quad{}+ u \bigl(14 r_{1} + 14 r_{2} + 14 v\bigr) + v \bigl(10 r_{1} + 10 r_{2} + 6 v\bigr)\bigr) + u \bigl(r_{1} \bigl(2 r_{1} + 5 r_{2}\bigr) + u \bigl(8 r_{1} + 8 r_{2} + 8 v\bigr) + v \bigl(10 r_{1} + 12 r_{2} + 8 v\bigr)\bigr)
\\[-1pt]
&\quad{}+ v \bigl(r_{1} \bigl(2 r_{1} + 4 r_{2}\bigr) + v \bigl(4 r_{1} + 5 r_{2} + 2 v\bigr)\bigr)\bigr)\bigr) + t \bigl(t \bigl(t u \bigl(2 r_{1} + 2 r_{2} + 2 v\bigr) + u \bigl(r_{1} r_{2} + u \bigl(2 r_{1} + 2 r_{2} + 2 v\bigr)
\\[-1pt]
&\quad{}+ v \bigl(2 r_{1} + r_{2} + 2 v\bigr)\bigr)\bigr) + u \bigl(u v \bigl(2 r_{1} + 2 v\bigr) + v \bigl(r_{1} \bigl(2 r_{1} + 2 r_{2}\bigr) + v \bigl(4 r_{1} + r_{2} + 2 v\bigr)\bigr)\bigr)\bigr)
\end{aligned}
\]

\noindent Regime $\mathrm M$:

\[
\begin{aligned}
&\mathcal{E}^{\mathrm{M}}_{R,L,0}=\alpha \bigl(\beta \bigl(e \bigl(e \bigl(2 r_{1} + 2 r_{2} + 2 u + 2 v\bigr) + r_{1} \bigl(2 r_{1} + 3 r_{2}\bigr) + u \bigl(6 r_{1} + 4 r_{2} + 2 u + 6 v\bigr) + v \bigl(4 r_{1} + 4 r_{2} + 2 v\bigr)\bigr)
\\[-1pt]
&\quad{}+ u \bigl(r_{1} \bigl(2 r_{1} + 3 r_{2}\bigr) + u \bigl(2 r_{1} + 2 v\bigr) + v \bigl(4 r_{1} + 2 r_{2} + 2 v\bigr)\bigr)\bigr) + e \bigl(e u \bigl(2 r_{1} + 2 v\bigr) + u \bigl(r_{1} \bigl(2 r_{1} + 3 r_{2}\bigr) + u \bigl(2 r_{1}
\\[-1pt]
&\quad{}+ 2 v\bigr) + v \bigl(4 r_{1} + 2 r_{2} + 2 v\bigr)\bigr)\bigr)\bigr) + \beta \bigl(\beta \bigl(e \bigl(e \bigl(2 r_{1} + 2 r_{2} + 2 u + 2 v\bigr) + r_{1} \bigl(2 r_{1} + 2 r_{2}\bigr) + u \bigl(4 r_{1} + 3 r_{2} + 2 u + 4 v\bigr)
\\[-1pt]
&\quad{}+ v \bigl(4 r_{1} + 3 r_{2} + 2 v\bigr)\bigr) + u \bigl(r_{1} \bigl(2 r_{1} + 2 r_{2}\bigr) + u \bigl(2 r_{1} + 2 v\bigr) + v \bigl(4 r_{1} + r_{2} + 2 v\bigr)\bigr)\bigr) + e \bigl(e \bigl(e \bigl(4 r_{1} + 4 r_{2} + 4 u + 4 v\bigr)
\\[-1pt]
&\quad{}+ r_{1} \bigl(4 r_{1} + 6 r_{2}\bigr) + u \bigl(12 r_{1} + 10 r_{2} + 6 u + 14 v\bigr) + v \bigl(10 r_{1} + 10 r_{2} + 6 v\bigr)\bigr) + u \bigl(r_{1} \bigl(8 r_{1} + 11 r_{2}\bigr) + u \bigl(12 r_{1} + 5 r_{2}
\\[-1pt]
&\quad{}+ 2 u + 14 v\bigr) + v \bigl(22 r_{1} + 14 r_{2} + 14 v\bigr)\bigr) + v \bigl(r_{1} \bigl(2 r_{1} + 4 r_{2}\bigr) + v \bigl(4 r_{1} + 5 r_{2} + 2 v\bigr)\bigr)\bigr) + u \bigl(u \bigl(r_{1} \bigl(4 r_{1} + 4 r_{2}\bigr)
\\[-1pt]
&\quad{}+ u \bigl(4 r_{1} + 4 v\bigr) + v \bigl(10 r_{1} + 2 r_{2} + 6 v\bigr)\bigr) + v \bigl(r_{1} \bigl(2 r_{1} + 2 r_{2}\bigr) + v \bigl(4 r_{1} + r_{2} + 2 v\bigr)\bigr)\bigr)\bigr) + e \bigl(e \bigl(e u \bigl(2 r_{1} + 2 v\bigr)
\\[-1pt]
&\quad{}+ u \bigl(r_{1} \bigl(4 r_{1} + 5 r_{2}\bigr) + u \bigl(6 r_{1} + 6 v\bigr) + v \bigl(10 r_{1} + 3 r_{2} + 6 v\bigr)\bigr)\bigr) + u \bigl(u \bigl(r_{1} \bigl(4 r_{1} + 4 r_{2}\bigr) + u \bigl(4 r_{1} + 4 v\bigr) + v \bigl(10 r_{1}
\\[-1pt]
&\quad{}+ 2 r_{2} + 6 v\bigr)\bigr) + v \bigl(r_{1} \bigl(2 r_{1} + 2 r_{2}\bigr) + v \bigl(4 r_{1} + r_{2} + 2 v\bigr)\bigr)\bigr)\bigr)
\end{aligned}
\]

\noindent\textit{Geometry $(\varepsilon_1,\varepsilon_2,j)=(R,L,1)$.}

\noindent Regime $\mathrm I$:

\[
\begin{aligned}
&\mathcal{E}^{\mathrm{I}}_{R,L,1}=\alpha \bigl(\alpha \bigl(\beta t \bigl(2 r_{1} + 4 u + 2 v\bigr) + t \bigl(4 t u + u \bigl(2 r_{1} + 4 u + 2 v\bigr)\bigr)\bigr) + \beta \bigl(\beta t \bigl(2 r_{1} + 4 u + 2 v\bigr) + t \bigl(3 \ell_{1} r_{1}
\\[-1pt]
&\quad{}+ t \bigl(2 \ell_{1} + 6 r_{1} + 16 u + 6 v\bigr) + u \bigl(6 \ell_{1} + 6 r_{1} + 12 u + 14 v\bigr) + v \bigl(4 \ell_{1} + 4 r_{1} + 4 v\bigr)\bigr)\bigr) + t \bigl(t \bigl(12 t u + u \bigl(2 \ell_{1} + 4 r_{1}
\\[-1pt]
&\quad{}+ 20 u + 12 v\bigr)\bigr) + u \bigl(3 \ell_{1} r_{1} + u \bigl(2 \ell_{1} + 2 r_{1} + 8 u + 10 v\bigr) + v \bigl(2 \ell_{1} + 4 r_{1} + 4 v\bigr)\bigr)\bigr)\bigr) + \beta \bigl(\beta t \bigl(2 \ell_{1} r_{1} + t \bigl(2 \ell_{1} + 2 r_{1}
\\[-1pt]
&\quad{}+ 4 u + 2 v\bigr) + u \bigl(4 \ell_{1} + 2 r_{1} + 4 u + 6 v\bigr) + v \bigl(3 \ell_{1} + 2 r_{1} + 2 v\bigr)\bigr) + t \bigl(t \bigl(6 \ell_{1} r_{1} + t \bigl(4 \ell_{1} + 4 r_{1} + 12 u + 4 v\bigr)
\\[-1pt]
&\quad{}+ u \bigl(14 \ell_{1} + 8 r_{1} + 20 u + 24 v\bigr) + v \bigl(10 \ell_{1} + 6 r_{1} + 6 v\bigr)\bigr) + u \bigl(5 \ell_{1} r_{1} + u \bigl(8 \ell_{1} + 4 r_{1} + 8 u + 16 v\bigr) + v \bigl(12 \ell_{1} + 6 r_{1}
\\[-1pt]
&\quad{}+ 10 v\bigr)\bigr) + v \bigl(4 \ell_{1} r_{1} + v \bigl(5 \ell_{1} + 2 r_{1} + 2 v\bigr)\bigr)\bigr)\bigr) + t \bigl(t \bigl(t \bigl(8 t u + u \bigl(2 \ell_{1} + 2 r_{1} + 20 u + 14 v\bigr)\bigr) + u \bigl(\ell_{1} r_{1} + u \bigl(2 \ell_{1}
\\[-1pt]
&\quad{}+ 2 r_{1} + 16 u + 22 v\bigr) + v \bigl(\ell_{1} + 4 r_{1} + 8 v\bigr)\bigr)\bigr) + u \bigl(u \bigl(u \bigl(4 u + 8 v\bigr) + v \bigl(2 r_{1} + 6 v\bigr)\bigr) + v \bigl(2 \ell_{1} r_{1} + v \bigl(\ell_{1} + 2 r_{1}
\\[-1pt]
&\quad{}+ 2 v\bigr)\bigr)\bigr)\bigr)
\end{aligned}
\]

\noindent Regime $\mathrm M$:

\[
\begin{aligned}
&\mathcal{E}^{\mathrm{M}}_{R,L,1}=\alpha \bigl(\alpha \bigl(\beta \bigl(e \bigl(2 r_{1} + 2 u + 2 v\bigr) + u \bigl(2 r_{1} + 2 v\bigr)\bigr) + e u \bigl(2 r_{1} + 2 v\bigr)\bigr) + \beta \bigl(\beta \bigl(e \bigl(2 r_{1} + 2 u + 2 v\bigr) + u \bigl(2 r_{1} + 2 v\bigr)\bigr)
\\[-1pt]
&\quad{}+ e \bigl(e \bigl(2 \ell_{1} + 6 r_{1} + 6 u + 6 v\bigr) + 3 \ell_{1} r_{1} + u \bigl(4 \ell_{1} + 14 r_{1} + 4 u + 18 v\bigr) + v \bigl(4 \ell_{1} + 4 r_{1} + 4 v\bigr)\bigr) + u \bigl(3 \ell_{1} r_{1}
\\[-1pt]
&\quad{}+ u \bigl(6 r_{1} + 6 v\bigr) + v \bigl(2 \ell_{1} + 4 r_{1} + 4 v\bigr)\bigr)\bigr) + e \bigl(e u \bigl(8 r_{1} + 8 v\bigr) + u \bigl(3 \ell_{1} r_{1} + u \bigl(6 r_{1} + 6 v\bigr) + v \bigl(2 \ell_{1} + 4 r_{1} + 4 v\bigr)\bigr)\bigr)\bigr)
\\[-1pt]
&\quad{}+ \beta \bigl(\beta \bigl(e \bigl(e \bigl(2 \ell_{1} + 2 r_{1} + 2 u + 2 v\bigr) + 2 \ell_{1} r_{1} + u \bigl(3 \ell_{1} + 4 r_{1} + 2 u + 6 v\bigr) + v \bigl(3 \ell_{1} + 2 r_{1} + 2 v\bigr)\bigr) + u \bigl(2 \ell_{1} r_{1}
\\[-1pt]
&\quad{}+ u \bigl(2 r_{1} + 2 v\bigr) + v \bigl(\ell_{1} + 2 r_{1} + 2 v\bigr)\bigr)\bigr) + e \bigl(e \bigl(e \bigl(4 \ell_{1} + 4 r_{1} + 4 u + 4 v\bigr) + 6 \ell_{1} r_{1} + u \bigl(10 \ell_{1} + 14 r_{1} + 6 u + 20 v\bigr)
\\[-1pt]
&\quad{}+ v \bigl(10 \ell_{1} + 6 r_{1} + 6 v\bigr)\bigr) + u \bigl(11 \ell_{1} r_{1} + u \bigl(5 \ell_{1} + 14 r_{1} + 2 u + 18 v\bigr) + v \bigl(14 \ell_{1} + 14 r_{1} + 16 v\bigr)\bigr) + v \bigl(4 \ell_{1} r_{1}
\\[-1pt]
&\quad{}+ v \bigl(5 \ell_{1} + 2 r_{1} + 2 v\bigr)\bigr)\bigr) + u \bigl(u \bigl(4 \ell_{1} r_{1} + u \bigl(4 r_{1} + 4 v\bigr) + v \bigl(2 \ell_{1} + 6 r_{1} + 6 v\bigr)\bigr) + v \bigl(2 \ell_{1} r_{1} + v \bigl(\ell_{1} + 2 r_{1} + 2 v\bigr)\bigr)\bigr)\bigr)
\\[-1pt]
&\quad{}+ e \bigl(e \bigl(e u \bigl(6 r_{1} + 6 v\bigr) + u \bigl(5 \ell_{1} r_{1} + u \bigl(10 r_{1} + 10 v\bigr) + v \bigl(3 \ell_{1} + 8 r_{1} + 8 v\bigr)\bigr)\bigr) + u \bigl(u \bigl(4 \ell_{1} r_{1} + u \bigl(4 r_{1} + 4 v\bigr)
\\[-1pt]
&\quad{}+ v \bigl(2 \ell_{1} + 6 r_{1} + 6 v\bigr)\bigr) + v \bigl(2 \ell_{1} r_{1} + v \bigl(\ell_{1} + 2 r_{1} + 2 v\bigr)\bigr)\bigr)\bigr)
\end{aligned}
\]

\noindent\textit{Geometry $(\varepsilon_1,\varepsilon_2,j)=(R,L,2)$.}

\noindent Regime $\mathrm I$:

\[
\begin{aligned}
&\mathcal{E}^{\mathrm{I}}_{R,L,2}=\alpha \bigl(\alpha \bigl(\beta t \bigl(2 \ell_{1} + 2 \ell_{2}\bigr) + t u \bigl(2 \ell_{1} + 2 \ell_{2}\bigr)\bigr) + \beta \bigl(\beta t \bigl(2 \ell_{1} + 2 \ell_{2}\bigr) + t \bigl(\ell_{1} \bigl(2 \ell_{1} + 3 \ell_{2}\bigr) + t \bigl(6 \ell_{1} + 6 \ell_{2}\bigr)
\\[-1pt]
&\quad{}+ u \bigl(6 \ell_{1} + 6 \ell_{2}\bigr) + v \bigl(4 \ell_{1} + 4 \ell_{2}\bigr)\bigr)\bigr) + t \bigl(t u \bigl(4 \ell_{1} + 4 \ell_{2}\bigr) + u \bigl(\ell_{1} \bigl(2 \ell_{1} + 3 \ell_{2}\bigr) + u \bigl(2 \ell_{1} + 2 \ell_{2}\bigr) + v \bigl(4 \ell_{1} + 4 \ell_{2}\bigr)\bigr)\bigr)\bigr)
\\[-1pt]
&\quad{}+ \beta \bigl(\beta t \bigl(\ell_{1} \bigl(2 \ell_{1} + 2 \ell_{2}\bigr) + t \bigl(2 \ell_{1} + 2 \ell_{2}\bigr) + u \bigl(2 \ell_{1} + 2 \ell_{2}\bigr) + v \bigl(2 \ell_{1} + 2 \ell_{2}\bigr)\bigr) + t \bigl(t \bigl(\ell_{1} \bigl(4 \ell_{1} + 6 \ell_{2}\bigr) + t \bigl(4 \ell_{1} + 4 \ell_{2}\bigr)
\\[-1pt]
&\quad{}+ u \bigl(8 \ell_{1} + 8 \ell_{2}\bigr) + v \bigl(6 \ell_{1} + 6 \ell_{2}\bigr)\bigr) + u \bigl(\ell_{1} \bigl(2 \ell_{1} + 5 \ell_{2}\bigr) + u \bigl(4 \ell_{1} + 4 \ell_{2}\bigr) + v \bigl(6 \ell_{1} + 6 \ell_{2}\bigr)\bigr) + v \bigl(\ell_{1} \bigl(2 \ell_{1} + 4 \ell_{2}\bigr)
\\[-1pt]
&\quad{}+ v \bigl(2 \ell_{1} + 2 \ell_{2}\bigr)\bigr)\bigr)\bigr) + t \bigl(t \bigl(t u \bigl(2 \ell_{1} + 2 \ell_{2}\bigr) + u \bigl(\ell_{1} \ell_{2} + u \bigl(2 \ell_{1} + 2 \ell_{2}\bigr) + v \bigl(4 \ell_{1} + 4 \ell_{2}\bigr)\bigr)\bigr) + u \bigl(u v \bigl(2 \ell_{1} + 2 \ell_{2}\bigr)
\\[-1pt]
&\quad{}+ v \bigl(\ell_{1} \bigl(2 \ell_{1} + 2 \ell_{2}\bigr) + v \bigl(2 \ell_{1} + 2 \ell_{2}\bigr)\bigr)\bigr)\bigr)
\end{aligned}
\]

\noindent Regime $\mathrm M$:

\[
\begin{aligned}
&\mathcal{E}^{\mathrm{M}}_{R,L,2}=\alpha \bigl(\alpha \bigl(\beta \bigl(e \bigl(2 \ell_{1} + 2 \ell_{2}\bigr) + u \bigl(2 \ell_{1} + 2 \ell_{2}\bigr)\bigr) + e u \bigl(2 \ell_{1} + 2 \ell_{2}\bigr)\bigr) + \beta \bigl(\beta \bigl(e \bigl(2 \ell_{1} + 2 \ell_{2}\bigr) + u \bigl(2 \ell_{1} + 2 \ell_{2}\bigr)\bigr)
\\[-1pt]
&\quad{}+ e \bigl(e \bigl(6 \ell_{1} + 6 \ell_{2}\bigr) + \ell_{1} \bigl(2 \ell_{1} + 3 \ell_{2}\bigr) + u \bigl(14 \ell_{1} + 14 \ell_{2}\bigr) + v \bigl(4 \ell_{1} + 4 \ell_{2}\bigr)\bigr) + u \bigl(\ell_{1} \bigl(2 \ell_{1} + 3 \ell_{2}\bigr) + u \bigl(6 \ell_{1} + 6 \ell_{2}\bigr)
\\[-1pt]
&\quad{}+ v \bigl(4 \ell_{1} + 4 \ell_{2}\bigr)\bigr)\bigr) + e \bigl(e u \bigl(8 \ell_{1} + 8 \ell_{2}\bigr) + u \bigl(\ell_{1} \bigl(2 \ell_{1} + 3 \ell_{2}\bigr) + u \bigl(6 \ell_{1} + 6 \ell_{2}\bigr) + v \bigl(4 \ell_{1} + 4 \ell_{2}\bigr)\bigr)\bigr)\bigr)
\\[-1pt]
&\quad{}+ \beta \bigl(\beta \bigl(e \bigl(e \bigl(2 \ell_{1} + 2 \ell_{2}\bigr) + \ell_{1} \bigl(2 \ell_{1} + 2 \ell_{2}\bigr) + u \bigl(4 \ell_{1} + 4 \ell_{2}\bigr) + v \bigl(2 \ell_{1} + 2 \ell_{2}\bigr)\bigr) + u \bigl(\ell_{1} \bigl(2 \ell_{1} + 2 \ell_{2}\bigr) + u \bigl(2 \ell_{1} + 2 \ell_{2}\bigr)
\\[-1pt]
&\quad{}+ v \bigl(2 \ell_{1} + 2 \ell_{2}\bigr)\bigr)\bigr) + e \bigl(e \bigl(e \bigl(4 \ell_{1} + 4 \ell_{2}\bigr) + \ell_{1} \bigl(4 \ell_{1} + 6 \ell_{2}\bigr) + u \bigl(14 \ell_{1} + 14 \ell_{2}\bigr) + v \bigl(6 \ell_{1} + 6 \ell_{2}\bigr)\bigr) + u \bigl(\ell_{1} \bigl(8 \ell_{1} + 11 \ell_{2}\bigr)
\\[-1pt]
&\quad{}+ u \bigl(14 \ell_{1} + 14 \ell_{2}\bigr) + v \bigl(14 \ell_{1} + 14 \ell_{2}\bigr)\bigr) + v \bigl(\ell_{1} \bigl(2 \ell_{1} + 4 \ell_{2}\bigr) + v \bigl(2 \ell_{1} + 2 \ell_{2}\bigr)\bigr)\bigr) + u \bigl(u \bigl(\ell_{1} \bigl(4 \ell_{1} + 4 \ell_{2}\bigr) + u \bigl(4 \ell_{1}
\\[-1pt]
&\quad{}+ 4 \ell_{2}\bigr) + v \bigl(6 \ell_{1} + 6 \ell_{2}\bigr)\bigr) + v \bigl(\ell_{1} \bigl(2 \ell_{1} + 2 \ell_{2}\bigr) + v \bigl(2 \ell_{1} + 2 \ell_{2}\bigr)\bigr)\bigr)\bigr) + e \bigl(e \bigl(e u \bigl(6 \ell_{1} + 6 \ell_{2}\bigr) + u \bigl(\ell_{1} \bigl(4 \ell_{1} + 5 \ell_{2}\bigr)
\\[-1pt]
&\quad{}+ u \bigl(10 \ell_{1} + 10 \ell_{2}\bigr) + v \bigl(8 \ell_{1} + 8 \ell_{2}\bigr)\bigr)\bigr) + u \bigl(u \bigl(\ell_{1} \bigl(4 \ell_{1} + 4 \ell_{2}\bigr) + u \bigl(4 \ell_{1} + 4 \ell_{2}\bigr) + v \bigl(6 \ell_{1} + 6 \ell_{2}\bigr)\bigr) + v \bigl(\ell_{1} \bigl(2 \ell_{1} + 2 \ell_{2}\bigr)
\\[-1pt]
&\quad{}+ v \bigl(2 \ell_{1} + 2 \ell_{2}\bigr)\bigr)\bigr)\bigr)
\end{aligned}
\]

\noindent\textit{Geometry $(\varepsilon_1,\varepsilon_2,j)=(R,R,0)$.}

\noindent Regime $\mathrm I$:

\[
\begin{aligned}
&\mathcal{E}^{\mathrm{I}}_{R,R,0}=\beta t \bigl(t \bigl(2 r_{1} r_{2} + v \bigl(2 r_{1} + 2 r_{2}\bigr)\bigr) + u \bigl(r_{1} \bigl(4 r_{1} + 6 r_{2}\bigr) + v \bigl(6 r_{1} + 6 r_{2}\bigr)\bigr) + v \bigl(r_{1} \bigl(2 r_{1} + 4 r_{2}\bigr) + v \bigl(2 r_{1} + 2 r_{2}\bigr)\bigr)\bigr)
\\[-1pt]
&\quad{}+ t \bigl(t u \bigl(r_{1} \bigl(4 r_{1} + 2 r_{2}\bigr) + v \bigl(2 r_{1} + 2 r_{2}\bigr)\bigr) + u \bigl(u \bigl(r_{1} \bigl(4 r_{1} + 2 r_{2}\bigr) + v \bigl(2 r_{1} + 2 r_{2}\bigr)\bigr) + v \bigl(r_{1} \bigl(2 r_{1} + 2 r_{2}\bigr) + v \bigl(2 r_{1}
\\[-1pt]
&\quad{}+ 2 r_{2}\bigr)\bigr)\bigr)\bigr)
\end{aligned}
\]

\noindent Regime $\mathrm M$:

\[
\begin{aligned}
&\mathcal{E}^{\mathrm{M}}_{R,R,0}=\beta \bigl(e \bigl(e \bigl(2 r_{1} r_{2} + u \bigl(2 r_{1} + 2 r_{2}\bigr) + v \bigl(2 r_{1} + 2 r_{2}\bigr)\bigr) + u \bigl(r_{1} \bigl(2 r_{1} + 4 r_{2}\bigr) + u \bigl(2 r_{1} + 2 r_{2}\bigr) + v \bigl(6 r_{1} + 6 r_{2}\bigr)\bigr)
\\[-1pt]
&\quad{}+ v \bigl(r_{1} \bigl(2 r_{1} + 4 r_{2}\bigr) + v \bigl(2 r_{1} + 2 r_{2}\bigr)\bigr)\bigr) + u \bigl(u v \bigl(2 r_{1} + 2 r_{2}\bigr) + v \bigl(r_{1} \bigl(2 r_{1} + 2 r_{2}\bigr) + v \bigl(2 r_{1} + 2 r_{2}\bigr)\bigr)\bigr)\bigr)
\\[-1pt]
&\quad{}+ e \bigl(e u v \bigl(2 r_{1} + 2 r_{2}\bigr) + u \bigl(u v \bigl(2 r_{1} + 2 r_{2}\bigr) + v \bigl(r_{1} \bigl(2 r_{1} + 2 r_{2}\bigr) + v \bigl(2 r_{1} + 2 r_{2}\bigr)\bigr)\bigr)\bigr)
\end{aligned}
\]

\noindent\textit{Geometry $(\varepsilon_1,\varepsilon_2,j)=(R,R,1)$.}

\noindent Regime $\mathrm I$:

\[
\begin{aligned}
&\mathcal{E}^{\mathrm{I}}_{R,R,1}=\alpha \bigl(\beta t \bigl(t \bigl(2 r_{1} + 2 v\bigr) + u \bigl(6 r_{1} + 6 v\bigr) + v \bigl(4 r_{1} + 2 v\bigr)\bigr) + t \bigl(t u \bigl(2 r_{1} + 2 v\bigr) + u \bigl(u \bigl(2 r_{1} + 2 v\bigr) + v \bigl(2 r_{1} + 2 v\bigr)\bigr)\bigr)\bigr)
\\[-1pt]
&\quad{}+ \beta \bigl(\beta t \bigl(t \bigl(2 r_{1} + 2 v\bigr) + u \bigl(4 r_{1} + 4 v\bigr) + v \bigl(3 r_{1} + 2 v\bigr)\bigr) + t \bigl(t \bigl(2 \ell_{1} r_{1} + t \bigl(4 r_{1} + 4 v\bigr) + u \bigl(14 r_{1} + 14 v\bigr) + v \bigl(2 \ell_{1}
\\[-1pt]
&\quad{}+ 10 r_{1} + 6 v\bigr)\bigr) + u \bigl(6 \ell_{1} r_{1} + u \bigl(8 r_{1} + 8 v\bigr) + v \bigl(6 \ell_{1} + 12 r_{1} + 8 v\bigr)\bigr) + v \bigl(4 \ell_{1} r_{1} + v \bigl(2 \ell_{1} + 5 r_{1} + 2 v\bigr)\bigr)\bigr)\bigr)
\\[-1pt]
&\quad{}+ t \bigl(t \bigl(t u \bigl(2 r_{1} + 2 v\bigr) + u \bigl(2 \ell_{1} r_{1} + u \bigl(2 r_{1} + 2 v\bigr) + v \bigl(2 \ell_{1} + r_{1} + 2 v\bigr)\bigr)\bigr) + u \bigl(u \bigl(2 \ell_{1} r_{1} + v \bigl(2 \ell_{1} + 2 v\bigr)\bigr)
\\[-1pt]
&\quad{}+ v \bigl(2 \ell_{1} r_{1} + v \bigl(2 \ell_{1} + r_{1} + 2 v\bigr)\bigr)\bigr)\bigr)
\end{aligned}
\]

\noindent Regime $\mathrm M$:

\[
\begin{aligned}
&\mathcal{E}^{\mathrm{M}}_{R,R,1}=\alpha \bigl(\beta \bigl(e \bigl(e \bigl(2 r_{1} + 2 u + 2 v\bigr) + u \bigl(4 r_{1} + 2 u + 6 v\bigr) + v \bigl(4 r_{1} + 2 v\bigr)\bigr) + u \bigl(2 u v + v \bigl(2 r_{1} + 2 v\bigr)\bigr)\bigr) + e \bigl(2 e u v
\\[-1pt]
&\quad{}+ u \bigl(2 u v + v \bigl(2 r_{1} + 2 v\bigr)\bigr)\bigr)\bigr) + \beta \bigl(\beta \bigl(e \bigl(e \bigl(2 r_{1} + 2 u + 2 v\bigr) + u \bigl(3 r_{1} + 2 u + 4 v\bigr) + v \bigl(3 r_{1} + 2 v\bigr)\bigr) + u \bigl(2 u v + v \bigl(r_{1} + 2 v\bigr)\bigr)\bigr)
\\[-1pt]
&\quad{}+ e \bigl(e \bigl(e \bigl(4 r_{1} + 4 u + 4 v\bigr) + 2 \ell_{1} r_{1} + u \bigl(2 \ell_{1} + 10 r_{1} + 6 u + 14 v\bigr) + v \bigl(2 \ell_{1} + 10 r_{1} + 6 v\bigr)\bigr) + u \bigl(4 \ell_{1} r_{1} + u \bigl(2 \ell_{1}
\\[-1pt]
&\quad{}+ 5 r_{1} + 2 u + 14 v\bigr) + v \bigl(6 \ell_{1} + 14 r_{1} + 14 v\bigr)\bigr) + v \bigl(4 \ell_{1} r_{1} + v \bigl(2 \ell_{1} + 5 r_{1} + 2 v\bigr)\bigr)\bigr) + u \bigl(u \bigl(4 u v + v \bigl(2 \ell_{1} + 2 r_{1} + 6 v\bigr)\bigr)
\\[-1pt]
&\quad{}+ v \bigl(2 \ell_{1} r_{1} + v \bigl(2 \ell_{1} + r_{1} + 2 v\bigr)\bigr)\bigr)\bigr) + e \bigl(e \bigl(2 e u v + u \bigl(6 u v + v \bigl(2 \ell_{1} + 3 r_{1} + 6 v\bigr)\bigr)\bigr) + u \bigl(u \bigl(4 u v + v \bigl(2 \ell_{1} + 2 r_{1}
\\[-1pt]
&\quad{}+ 6 v\bigr)\bigr) + v \bigl(2 \ell_{1} r_{1} + v \bigl(2 \ell_{1} + r_{1} + 2 v\bigr)\bigr)\bigr)\bigr)
\end{aligned}
\]

\noindent\textit{Geometry $(\varepsilon_1,\varepsilon_2,j)=(R,R,2)$.}

\noindent Regime $\mathrm I$:

\[
\begin{aligned}
&\mathcal{E}^{\mathrm{I}}_{R,R,2}=\alpha \bigl(\alpha \bigl(\beta t \bigl(4 u + 2 v\bigr) + t \bigl(4 t u + u \bigl(4 u + 2 v\bigr)\bigr)\bigr) + \beta \bigl(\beta t \bigl(4 u + 2 v\bigr) + t \bigl(t \bigl(2 \ell_{2} + 16 u + 6 v\bigr) + u \bigl(8 \ell_{1}
\\[-1pt]
&\quad{}+ 6 \ell_{2} + 12 u + 14 v\bigr) + v \bigl(4 \ell_{1} + 4 \ell_{2} + 4 v\bigr)\bigr)\bigr) + t \bigl(t \bigl(12 t u + u \bigl(8 \ell_{1} + 2 \ell_{2} + 20 u + 12 v\bigr)\bigr) + u \bigl(u \bigl(8 \ell_{1} + 2 \ell_{2} + 8 u
\\[-1pt]
&\quad{}+ 10 v\bigr) + v \bigl(4 \ell_{1} + 2 \ell_{2} + 4 v\bigr)\bigr)\bigr)\bigr) + \beta \bigl(\beta t \bigl(t \bigl(2 \ell_{2} + 4 u + 2 v\bigr) + u \bigl(4 \ell_{1} + 4 \ell_{2} + 4 u + 6 v\bigr) + v \bigl(2 \ell_{1} + 3 \ell_{2} + 2 v\bigr)\bigr)
\\[-1pt]
&\quad{}+ t \bigl(t \bigl(2 \ell_{1} \ell_{2} + t \bigl(4 \ell_{2} + 12 u + 4 v\bigr) + u \bigl(16 \ell_{1} + 14 \ell_{2} + 20 u + 24 v\bigr) + v \bigl(6 \ell_{1} + 10 \ell_{2} + 6 v\bigr)\bigr) + u \bigl(\ell_{1} \bigl(4 \ell_{1} + 6 \ell_{2}\bigr)
\\[-1pt]
&\quad{}+ u \bigl(12 \ell_{1} + 8 \ell_{2} + 8 u + 16 v\bigr) + v \bigl(14 \ell_{1} + 12 \ell_{2} + 10 v\bigr)\bigr) + v \bigl(\ell_{1} \bigl(2 \ell_{1} + 4 \ell_{2}\bigr) + v \bigl(4 \ell_{1} + 5 \ell_{2} + 2 v\bigr)\bigr)\bigr)\bigr)
\\[-1pt]
&\quad{}+ t \bigl(t \bigl(t \bigl(8 t u + u \bigl(12 \ell_{1} + 2 \ell_{2} + 20 u + 14 v\bigr)\bigr) + u \bigl(\ell_{1} \bigl(4 \ell_{1} + 2 \ell_{2}\bigr) + u \bigl(20 \ell_{1} + 2 \ell_{2} + 16 u + 22 v\bigr) + v \bigl(12 \ell_{1} + \ell_{2} + 8 v\bigr)\bigr)\bigr)
\\[-1pt]
&\quad{}+ u \bigl(u \bigl(\ell_{1} \bigl(4 \ell_{1} + 2 \ell_{2}\bigr) + u \bigl(8 \ell_{1} + 4 u + 8 v\bigr) + v \bigl(10 \ell_{1} + 6 v\bigr)\bigr) + v \bigl(\ell_{1} \bigl(2 \ell_{1} + 2 \ell_{2}\bigr) + v \bigl(4 \ell_{1} + \ell_{2} + 2 v\bigr)\bigr)\bigr)\bigr)
\end{aligned}
\]

\noindent Regime $\mathrm M$:

\[
\begin{aligned}
&\mathcal{E}^{\mathrm{M}}_{R,R,2}=\alpha \bigl(\alpha \bigl(\beta \bigl(e \bigl(2 u + 2 v\bigr) + 2 u v\bigr) + 2 e u v\bigr) + \beta \bigl(\beta \bigl(e \bigl(2 u + 2 v\bigr) + 2 u v\bigr) + e \bigl(e \bigl(2 \ell_{2} + 6 u + 6 v\bigr) + u \bigl(4 \ell_{1}
\\[-1pt]
&\quad{}+ 4 \ell_{2} + 4 u + 18 v\bigr) + v \bigl(4 \ell_{1} + 4 \ell_{2} + 4 v\bigr)\bigr) + u \bigl(6 u v + v \bigl(4 \ell_{1} + 2 \ell_{2} + 4 v\bigr)\bigr)\bigr) + e \bigl(8 e u v + u \bigl(6 u v + v \bigl(4 \ell_{1} + 2 \ell_{2} + 4 v\bigr)\bigr)\bigr)\bigr)
\\[-1pt]
&\quad{}+ \beta \bigl(\beta \bigl(e \bigl(e \bigl(2 \ell_{2} + 2 u + 2 v\bigr) + u \bigl(2 \ell_{1} + 3 \ell_{2} + 2 u + 6 v\bigr) + v \bigl(2 \ell_{1} + 3 \ell_{2} + 2 v\bigr)\bigr) + u \bigl(2 u v + v \bigl(2 \ell_{1} + \ell_{2} + 2 v\bigr)\bigr)\bigr)
\\[-1pt]
&\quad{}+ e \bigl(e \bigl(e \bigl(4 \ell_{2} + 4 u + 4 v\bigr) + 2 \ell_{1} \ell_{2} + u \bigl(6 \ell_{1} + 10 \ell_{2} + 6 u + 20 v\bigr) + v \bigl(6 \ell_{1} + 10 \ell_{2} + 6 v\bigr)\bigr) + u \bigl(\ell_{1} \bigl(2 \ell_{1} + 4 \ell_{2}\bigr)
\\[-1pt]
&\quad{}+ u \bigl(4 \ell_{1} + 5 \ell_{2} + 2 u + 18 v\bigr) + v \bigl(18 \ell_{1} + 14 \ell_{2} + 16 v\bigr)\bigr) + v \bigl(\ell_{1} \bigl(2 \ell_{1} + 4 \ell_{2}\bigr) + v \bigl(4 \ell_{1} + 5 \ell_{2} + 2 v\bigr)\bigr)\bigr) + u \bigl(u \bigl(4 u v
\\[-1pt]
&\quad{}+ v \bigl(6 \ell_{1} + 2 \ell_{2} + 6 v\bigr)\bigr) + v \bigl(\ell_{1} \bigl(2 \ell_{1} + 2 \ell_{2}\bigr) + v \bigl(4 \ell_{1} + \ell_{2} + 2 v\bigr)\bigr)\bigr)\bigr) + e \bigl(e \bigl(6 e u v + u \bigl(10 u v + v \bigl(8 \ell_{1} + 3 \ell_{2} + 8 v\bigr)\bigr)\bigr)
\\[-1pt]
&\quad{}+ u \bigl(u \bigl(4 u v + v \bigl(6 \ell_{1} + 2 \ell_{2} + 6 v\bigr)\bigr) + v \bigl(\ell_{1} \bigl(2 \ell_{1} + 2 \ell_{2}\bigr) + v \bigl(4 \ell_{1} + \ell_{2} + 2 v\bigr)\bigr)\bigr)\bigr)
\end{aligned}
\]

\paragraph{Verification protocol.}
Every row above is obtained by substituting the stated coordinates into
the exact clipping identity, ordering each selected triple on the line,
replacing its Product-Gap weight by the product of its two consecutive
gaps, resolving the minima according to the relevant radial regime,
and collecting in the displayed nonnegative variables. The Horner
forms are exact, and coefficientwise nonnegativity proves every row.

% END OF FULL EDITABLE EXTERIOR TABLES.
% Expected contents: 45 shadow rows, 6 three-witness rows, 24 four-witness rows.

\endgroup

\subsection{Deletion induction and continuity}

The three report regions cover the line, so every reduced atom is
nonnegative. Therefore,

$$
\Psi(O)\geq\Psi(O\setminus\{p\})
$$

whenever $p$ is an ordinary report closest to $y$. Repeatedly delete a
closest ordinary report. Once fewer than three ordinary reports remain,
there are no ordinary triples, so $Z=B=0$ and $\Psi=0$. Restoring the
deleted reports gives $\Psi(O)\geq0$. The expected-cost identity then
implies

$$
C_y(r\mid O)\geq C_y(y\mid O).
$$

For fixed labels, every consecutive-gap product and every
nearest-distance term is continuous in all report coordinates.
Approximating an arbitrary profile by generic profiles and taking a
limit handles repeated reports, nonunique closest reports, coordinate
ties, and radial ties. Thus the three-facility Product-Gap mechanism is
strategyproof in expectation.

\section{Exact Rational Certificate for the Four-Facility Counterexample}
\label{app:k4-certificate}

This appendix verifies the strict manipulation in
Proposition~\ref{prop:k4-manipulation} using exact rational arithmetic.
The strategic agent has true location $0$ and reports
$h=1/1000$. There are $m=100$ ordinary agents at $0$. After deleting
these labels, let $P$ contain one agent at $-1$, one agent at $1$, and,
for every $j\in\{1,\ldots,100\}$, exactly $2^j$ agents at

$$
-1-\frac{1}{100\cdot2^j}.
$$

Define

\[
\begin{aligned}
Z
&=\sum_{F\in\binom{P}{4}}W_4(F),
&
B
&=\sum_{F\in\binom{P}{4}}W_4(F)d(0,F),\\
S_0
&=\sum_{A\in\binom{P}{3}}W_4(A\cup\{0\}),
&
S_h
&=\sum_{A\in\binom{P}{3}}W_4(A\cup\{h\}),\\
T
&=\sum_{A\in\binom{P}{3}}W_4(A\cup\{h\})d(0,A\cup\{h\}),
&
J
&=\sum_{C\in\binom{P}{2}}W_4(C\cup\{0,h\}).
\end{aligned}
\]

Every report in $P$ has absolute value at least $1$, so $T=hS_h$.
Positive-weight selected sets use at most one label at a fixed
coordinate. Hence

$$
U(0)=Z+(m+1)S_0,
\qquad
M(0)=B,
$$

and

$$
U(h)=Z+S_h+m(S_0+J),
\qquad
M(h)=B+T.
$$

Therefore,

\[
\begin{aligned}
\Phi_4(O;h)
&=U(0)M(h)-U(h)M(0)\\
&=(Z+S_0)(B+T)-(Z+S_h)B
+m(S_0T-BJ).
\end{aligned}
\]

We evaluate these quantities by coordinate support. Let the distinct
supports of a finite profile be $x_1<\cdots<x_s$, with multiplicities
$\mu_1,\ldots,\mu_s$. Let $G_\ell(j)$ be the total labeled
Product-Gap weight of all increasing selections of $\ell$ distinct
supports whose rightmost support is $x_j$. Then

$$
G_1(j)=\mu_j
$$

and, for $\ell\geq2$,

$$
G_\ell(j)
=
\mu_j\sum_{i<j}G_{\ell-1}(i)(x_j-x_i).
$$

Thus the total four-set mass is

$$
\mathsf W(P)=\sum_{j=1}^sG_4(j).
$$

Applying the same recurrence after adding $0$ and $h$ gives

$$
S_0=\mathsf W(P\cup\{0\})-\mathsf W(P),
$$

$$
S_h=\mathsf W(P\cup\{h\})-\mathsf W(P),
$$

and

\[
J
=
\mathsf W(P\cup\{0,h\})
-\mathsf W(P\cup\{0\})
-\mathsf W(P\cup\{h\})
+\mathsf W(P).
\]

Every ordinary four-set containing $-1$ or $1$ has nearest distance
$1$. An all-left four-set whose rightmost selected group is the support
indexed by $j$ has nearest distance
$1+1/(100\cdot2^j)$. Consequently,

$$
B
=
Z+
\sum_{j=1}^{100}\frac{G_4(j)}{100\cdot2^j}.
$$

All preceding operations are exact rational additions and
multiplications. They give

$$
\Phi_4(O;h)=-\frac{N_\Phi}{D_\Phi}<0,
$$

where

\[
\begin{aligned}
N_\Phi={}&
119729241211910914414021691080179703547390611638925094828474771\\
&101411794241079540697997847748037180981054038680063806949043408217,
\\[1mm]
D_\Phi={}&
401734511064747568885490523085290650630550748445698208825344\\
&00000000000.
\end{aligned}
\]

Both integers are positive, so the determinant is strictly negative.
For readability,

$$
\frac{M(0)}{U(0)}
\approx
0.00724556377223489236
$$

and

$$
\frac{M(h)}{U(h)}
\approx
0.00724554178194980469.
$$

The decimal values are only a check; the exact negative rational
determinant certifies the profitable deviation.

\end{document}